\documentclass[11pt,twoside]{atmp}

\usepackage{amsmath,amssymb}

\usepackage[T1]{fontenc}
\usepackage{amssymb}
\usepackage{amsthm}
\usepackage{amsmath}
\usepackage[all]{xy}
\usepackage{graphicx}
\usepackage{mathrsfs}
\usepackage{afterpage}
\usepackage{float}
\usepackage[bottom]{footmisc}
\usepackage[margin=2cm,font=normalsize,labelfont=bf]{caption}
\usepackage[normalem]{ulem}
\usepackage[section]{placeins}
\usepackage{verbatim}
\usepackage{amstext}

\newtheorem{theorem}{Theorem}[section]
\newtheorem{definition}[theorem]{Definition}
\newtheorem{rem}[theorem]{Remark}

\newtheorem{exa}[theorem]{Example}
\newtheorem{proposition}[theorem]{Proposition}
\newtheorem{lemma}[theorem]{Lemma}

\newenvironment{remark}{\begin{rem}\normalfont}{\end{rem}}
\newenvironment{example}{\begin{exa}\normalfont}{\end{exa}}

\makeatletter

\def\Z {\mathbb{Z}}

\def\C {\mathbb{C}}
\def\im{\mathrm{i}}
\def\id{\mathrm{id}}

\renewcommand{\varepsilon}{\epsilon}

\newcommand{\alxy}[1]{\begin{aligned}\xymatrix{#1}\end{aligned}}
\newcommand{\alxydim}[2]{\begin{aligned}\xymatrix#1{#2}\end{aligned}}

\newcommand\erf[1]{(\ref{#1})}
\newcommand\prf[1]{page \pageref{#1}}

\renewcommand{\emph}[1]{\def\reserved@a{it}\ifx\f@shape\reserved@a\uline{#1}\else\textit{#1}\fi}

\newcommand{\trivlin}{\hbox{$1\hskip -1.2pt\vrule depth 0pt height 1.6ex width 0.7pt \vrule depth 0pt height 0.3pt width 0.12em$}}

\makeatother

\renewcommand{\to}{\!\xymatrix@C=0.5cm{\ar[r] &}}
\renewcommand{\mapsto}{\xymatrix@C=0.5cm{\ar@{|->}[r] &}}
\renewcommand{\Rightarrow}{\xymatrix@C=0.5cm{\ar@{=>}[r] &}}
\newcommand{\incl}{\xymatrix@C=0.5cm{\ar@{^(->}[r] &}}

\begin{document}

\title[Bundle Gerbes for Orientifold Sigma Models]
{Bundle Gerbes for\\Orientifold Sigma Models}


\author[K. Gaw\k{e}dzki, R. R. Suszek, K. Waldorf]
{Krzysztof Gaw\k{e}dzki\,$^1$\\
Rafa\l ~R.~Suszek\,$^2$\\
Konrad Waldorf\,$^3$}

\address{$^1$ Laboratoire de Physique, ENS-Lyon\\46 All\'ee d'Italie, F-69364 Lyon, France \bigskip \\
$^2$ Department of Mathematics, King's College London\\
Strand, London WC2R 2LS, United Kingdom
\bigskip \\
$^3$ Department Mathematik, Universit\"at Hamburg\\Bundesstra\ss e 55, D-20146 Hamburg, Germany}


\begin{abstract}
Bundle gerbes with connection and their modules play an important
role in the theory of two-dimensional sigma models with a background
Wess-Zumino flux: their holonomy determines the contribution of the
flux to the Feynman amplitudes of classical fields. We discuss
additional structures on bundle gerbes and gerbe modules needed in
similar constructions for orientifold sigma models describing closed
and open strings.
\end{abstract}

\maketitle

\newpage
\section{Introduction}

Bundle gerbes \cite{murray} are geometric structures related to
sheaves of line bundles, see \cite{schweigert2,murray3} for recent
historical essays. They appear naturally in the mathematical context
of lifting principal $G$-bundles to $\hat G$-bundles for central
extensions $\hat G$ of a Lie group $G$ by a circle. In the
physical context, they arise in  studies of quantum field theory
anomalies \cite{carey5} or, together with bundle gerbe modules, in a
construction of groups of string theory charges \cite{bouwknegt1}.
The present paper has been mainly motivated by the role that bundle
gerbes equipped with hermitian connections play in the theory of
two-dimensional sigma models with a Wess-Zumino (WZ) term
\cite{wess1,novikov1} in the action functional. Classically, the
fields of such a sigma model are maps $\phi$ from a two-dimensional
surface $\Sigma$, called the worldsheet, to the target manifold $M$
equipped with a metric and a closed 3-form $H$. The WZ term
describes the background $H$-flux. Locally, it is given by integrals
of the pullbacks $\phi^*B$ of local Kalb-Ramond 2-forms $B$ on $M$
such that $\mathrm{d}B=H$. The ambiguities in defining such a
functional $S_{\mathrm{WZ}}(\phi)$ globally in topologically
non-trivial situations were originally
studied with cohomology techniques in \cite{alvarez1} and
\cite{gawedzki3} for closed worldsheets, and in \cite{kapustin1} for
worldsheets with boundary. They may be sorted out systematically
using bundle gerbes with connection over the manifold $M$ and, in
the case with boundary, bundle gerbe modules, see
\cite{gawedzki1,carey2,gawedzki4}. In particular, a choice of a
bundle gerbe $\mathcal{G}$ with connection, whose curvature 3-form
is $H$, determines unambiguously the Feynman amplitudes
$\mathrm{e}^{\im S_{\mathrm{WZ}}(\phi)}$ on closed oriented worldsheets
$\Sigma$. In the beginning of Section \ref{sec1} of this article, we
recall the definition \cite{murray} of bundle gerbes with connection
and review their 1-morphisms and the 2-morphisms between
1-morphisms, all together forming a 2-category \cite{stevenson1,
waldorf1}.

The WZ term in the action functional plays an essential role in
Wess-Zumino-Witten (WZW) sigma models \cite{witten1}, assuring their
conformal symmetry on the quantum level and rendering them soluble.
The target space of a WZW sigma model is a compact Lie group $G$
equipped with a bundle gerbe whose curvature is a bi-invariant
closed 3-form $H$. Bundle gerbes and their modules are specially
useful in treating the case \cite{felder2} of WZW models with
non-simply connected target groups $G'$ that are quotients of their
covering groups $G$ by a finite subgroup $\Gamma_0$ of the center
$Z(G)$ of $G$ \cite{gawedzki2}. A gerbe $\mathcal{G}'$ over
$G'=G/\Gamma_0$ may be thought of as a gerbe $\mathcal{G}$ over $G$
equipped with a $\Gamma_0$-equivariant structure that picks up in a
consistent way isomorphisms between the pullback gerbes
$\gamma^*\mathcal{G}$ and $\mathcal{G}$ for each element $\gamma\in
\Gamma_0$. The notion of  equivariant structures on gerbes extends
to the case of gerbes over general manifolds $M$ on which a finite
group $\Gamma_0$ acts preserving the curvature 3-form $H$, possibly
with fixed points. A gerbe over $M$ with such an equivariant
structure may be thought of as a gerbe over the orbifold
$M/\Gamma_0$ and it may be used to define the WZ action functional
for sigma models with the orbifold target. Gerbes with equivariant
structures with respect to actions of continuous groups will be
discussed elsewhere, see also \cite{meinrenken1}. They find
application in gauged sigma models with WZ term.

Motivated by the theory of unoriented strings
\cite{schwarz1,pradisi5}, one would like to define the WZ action
functional for unoriented (in particular unorientable) worldsheets
$\Sigma$. More exactly, one considers  so-called orientifold sigma
models. Their classical fields $\hat\phi$ map the oriented double
$\hat\Sigma$, which is equipped with an orientation-changing
involution $\sigma$ such that $\Sigma=\hat\Sigma/\sigma$, to the
target $M$ equipped with an involution $k$ so that
$\hat\phi\circ\sigma=k\circ\hat\phi$. Assuming that $k^*H=-H$, one
may define the WZ action functional for such fields using a gerbe
$\mathcal{G}$ over $M$ with curvature $H$ additionally equipped with
a  Jandl structure \cite{schreiber1}.  Such a structure on
$\mathcal{G}$ may be considered as a twisted version of a
$\Z_2$-equivariant structure for the $\Z_2$ action on $M$ defined by
$k$. It picks up in a consistent way an isomorphism between
$k^*\mathcal{G}$ and the \emph{dual} gerbe $\mathcal{G}^*$.

One may consider more general orientifold sigma models with
the WZ term, corresponding to an action on $M$ of a finite
group $\Gamma$ with elements $\gamma$ such that
$\gamma^*H=\epsilon(\gamma)H$ for a homomorphism $\epsilon$ from
$\Gamma$ to $\{\pm1\}\equiv\Z_2$. The notions of
$\Gamma_0$-equivariant and Jandl structures on a gerbe $\mathcal{G}$
may be merged into the one of a $(\Gamma,\epsilon)$-equivariant
structure, which we shall also call a  twisted-equivariant
structure. Such a structure consistently picks up isomorphisms
between $\gamma^*\mathcal{G}$ and either $\mathcal{G}$ or
$\mathcal{G}^*$, according to the sign of $\epsilon(\gamma)$. The
twisted-equivariant structures on gerbes are introduced in
Section\,1 and are the main topic of the present article. A special
case of such structures occurs when the normal subgroup
$\Gamma_0={\rm ker}\,{\epsilon}$ of $\Gamma$ acts on $M$ without
fixed points. If $\epsilon\equiv1$ so that $\Gamma_0=\Gamma$, we are
back to the correspondence between gerbes over $M$ with
$\Gamma_0$-equivariant structure and gerbes over $M'=M/\Gamma_0$. If
$\epsilon$ is non-trivial, so that $\Gamma/\Gamma_0=\Z_2$, then the
action of $\Gamma$ on $M$ induces a $\Z_2$-action on $M'$, with the
non-trivial element of $\Z_2$ acting as an involution $k'$ inverting
the sign of the projected 3-form $H'$. In this situation, gerbes
over $M$ with curvature $H$ and $(\Gamma,\epsilon)$-equivariant
structure correspond to gerbes over $M'$ with curvature $H'$ and a
Jandl structure. This descent theory for bundle gerbes is discussed
in Section\,2.

The present article provides a geometric theory extending and
completing the discussion of our previous paper \cite{gawedzki6}
that was devoted to the study of gerbes with twisted-equivariant
structures over simple simply connected compact Lie groups. Such
gerbes are needed for applications to the orientifolds of  WZW
models. In \cite{gawedzki6}, we used a local description of
gerbes and cohomological tools. Section\,3 of the present paper
establishes the relation between the geometric and cohomological
languages.

For oriented worldsheets $\Sigma$ with boundary, the classical
fields $\phi:\Sigma\to M$ are often constrained to take values in
special submanifolds $D$ of $M$ on the boundary components of
$\Sigma$. Such submanifolds are called (D-)branes in string theory.
The extension of the definition of the WZ action to this case
requires a choice of a bundle gerbe $\mathcal{G}$ with curvature $H$
and of gerbe modules over the submanifolds $D$. Gerbe modules may be
viewed as vector bundles with connection twisted by the gerbe. In
the context of the 2-category of bundle gerbes, they can also be
viewed as particular 1-morphisms \cite{waldorf1}. In Section 4, we
adapt this notion to the case of gerbes with
$(\Gamma,\epsilon)$-equivariant structures needed for applications
to orientifold sigma models on worldsheets with boundary. We also
discuss a presentation of such gerbe modules in terms of local data
and develop their descent theory.

The Feynman amplitudes $\mathrm{e}^{\im S_{\mathrm{WZ}}(\phi)}$
of fields $\phi$ defined on closed worldsheets are given by
the holonomy of gerbes \cite{gawedzki1}. In the case of unoriented
worldsheets, the gerbes have to be equipped additionally with a
Jandl structure \cite{schreiber1}. For oriented worldsheets with
boundary, the holonomy giving the Feynman amplitudes receives also
contributions from the gerbe modules over the brane
worldvolumes that provide the boundary conditions of
the theory \cite{gawedzki1,carey2,gawedzki4}. In this article, we
introduce a generalization of both notions for unoriented
worldsheets with boundary using gerbe modules for  gerbes with Jandl
structure. This is discussed in both the geometric and the local
language in Section\,5.

In Conclusions, we summarize the contents of the present paper
and sketch the directions for further work that
includes extending the discussion \cite{gawedzki6} of the
orientifold WZW models on closed worldsheets to the ones on
worldsheets with boundary.

\textbf{Acknowledgements.} K.G. and R.R.S. acknowledge the support of Agence National
de Recherche under the contract
ANR-05-BLAN-0029-03. K.W. acknowledges the support of the
Collaborative Research Centre 676 ``Particles, Strings and the Early
Universe'' and thanks ENS-Lyon for kind hospitality.

\section{Twisted-Equivariant Bundle Gerbes}

\label{sec1}

We review bundle gerbes and their algebraic structure in Section
\ref{sec1_1} and define twisted-equivariant structures on them
in Section \ref{sec4}.
Twisted-equivariant structures include two extremal versions:
the untwisted one, which is just an ordinary equivariant structure,
and the twisted $\Z_2$-equivariant one, which coincides, as we discuss
in Section \ref{sec15}, with a Jandl structure.

\subsection{The  2-Category of Bundle
Gerbes}

\label{sec1_1}

In the whole article, we work with the following conventions:
\begin{itemize}
\item
\emph{Vector bundles} are hermitian vector bundles with unitary
connection, and isomorphisms of vector bundles respect the hermitian
structure and the connections. These conventions in particular apply
to line bundles.

\item
If $\pi:Y \to M$ is a surjective submersion between smooth manifolds,
we denote by
\begin{equation*}
Y^{[k]}\ :=\ Y \times_M
... \times_M Y
\end{equation*}
the $k$-fold fibre product of $Y$ with itself (composed of the elements
in the Cartesian product whose components have the same projection to
$M$). The fibre
products are, again, smooth manifolds in such a way that the
canonical projections $\pi_{i_1...i_r}:Y^{[k]} \to Y^{[r]}$ are
smooth maps.
\end{itemize}

In the following, we collect the basic definitions.

\begin{definition}[\cite{murray}]
\label{def_gerbe}
A \emph{bundle gerbe} $\mathcal{G}$ over
a smooth manifold
$M$ is a surjective submersion $\pi: Y
\to M$, a line bundle $L$ over $Y^{[2]}$,
a 2-form $C \in \Omega^2(Y)$, and
an isomorphism
\begin{equation*}
\mu : \pi_{12}^{*}L \otimes \pi_{23}^{*}L
\to \pi_{13}^{*}L
\end{equation*}
of line bundles over $Y^{[3]}$, such that
two axioms are satisfied:
\begin{list}{}{\leftmargin=1cm\labelwidth=1cm\labelsep=0.25cm}

\item[\normalfont(G1)]
The curvature of $L$ is fixed by
\begin{equation*}
\mathrm{curv}(L) = \pi_2^{*}C -
\pi_1^{*}C\text{.}
\end{equation*}

\item[\normalfont(G2)]
$\mu$ is associative
in the sense that the diagram
\begin{equation*}
\alxydim{@C=3cm@R=1.2cm}{\pi_{12}^{*}L \otimes \pi_{23}^{*}L \otimes \pi_{34}^{*}L \ar[r]^{\pi_{123}^{*}\mu\otimes
\id} \ar[d]_{\id \otimes \pi_{234}^{*}\mu} & \pi_{13}^{*}L \otimes \pi_{34}^{*}L
\ar[d]^{\pi_{134}^{*}\mu} \\ \pi_{12}^{*}L \otimes \pi_{24}^{*}L \ar[r]_{\pi_{124}^{*}
\mu} & \pi_{14}^{*}L}
\end{equation*}
of isomorphisms of  line bundles
  over $Y^{[4]}$ is commutative.
\end{list}
\end{definition}

\begin{example}
\label{ex2} On any smooth manifold $M$, there is a family
$\mathcal{I}_{\omega}$ of \emph{trivial bundle gerbes} over $M$,
labelled by 2-forms $\omega\in\Omega^2(M)$. The surjective
submersion of $\mathcal{I}_{\omega}$ is $Y:=M$ and the identity
$\pi:=\id_M$, the line bundle over $Y^{[2]}\cong M$ is the trivial
line bundle (equipped with the trivial flat connection), and the
isomorphism $\mu$ is the identity between trivial line bundles. Its
2-form is the given 2-form $C:=\omega$. \end{example}

Associated to a bundle gerbe $\mathcal{G}$ over $M$ is a 3-form
$H\in\Omega^3(M)$ called the \emph{curvature} of $\mathcal{G}$. It
is the unique 3-form which satisfies $\pi^{*}H=\mathrm{d}C$. The
trivial bundle gerbe $\mathcal{I}_{\omega}$ has the curvature
$\mathrm{d}\omega$.

We would like to compare two bundle gerbes using a notion of
morphisms between bundle gerbes. The morphisms between such gerbes
which we consider here have been introduced in
\cite{waldorf1}. For simplicity, we work with the convention that we
do not label or write down pullbacks along canonical projection
maps, such as in \erf{2} and (1M1) below.

\begin{definition}
Let $\mathcal{G}_1$ and $\mathcal{G}_2$ be bundle gerbes over $M$. A \uline{1-morphism} \begin{equation*}
\mathcal{A}:\mathcal{G}_1
\to \mathcal{G}_2
\end{equation*}
consists of a surjective submersion $\zeta:Z \to Y_1 \times_M Y_2$,
a vector bundle $A$ over $Z$, and
an isomorphism
\begin{equation}
\label{2}
\alpha : L_1  \otimes
\zeta_{2}^* A \to  \zeta_{1}^* A \otimes L_2
\end{equation}
of vector bundles over $Z \times_M Z$, such that two axioms are satisfied:
\begin{list}{}{\leftmargin=1.2cm\labelwidth=1.2cm\labelsep=0.2cm}
\item[\normalfont(1M1)]
The curvature of $A$ obeys
\begin{equation*}
\frac{_1}{^n}\mathrm{tr}(\mathrm{curv}(A))=C_2 - C_1\text{,}
\end{equation*}
where $n$ is the rank of $A$.

\item[\normalfont(1M2)]
 The isomorphism $\alpha$
commutes with the isomorphisms $\mu_1$ and
$\mu_2$ of the gerbes $\mathcal{G}_1$ and $\mathcal{G}_2$ in the sense
that the diagram
\begin{equation*}
\alxydim{@C=2.5cm@R=1.2cm}{\zeta_{12}^{*}L_1 \otimes \zeta_{23}^{*}L_1 \otimes \zeta_3^{*}A \ar[r]^-{\mu_1
\otimes \id} \ar[d]_{\id \otimes \zeta_{23}^{*}\alpha} & \zeta_{13}^{*}L_1
\otimes \zeta_3^{*}A \ar[dd]^{\zeta_{13}^{*}\alpha} \\ \zeta_{12}^{*}L_1
\otimes \zeta_2^{*}A \otimes \zeta_{23}^{*}L_2 \ar[d]_{\zeta_{12}^{*}\alpha
\otimes \id} & \\ \zeta_1^{*}A \otimes \zeta_{12}^{*}L_2 \otimes \zeta_{23}^{*}L_2
\ar[r]_-{\id \otimes \mu_2} & \zeta_1^{*}A \otimes \zeta_{13}^{*}L_2}
\label{1}
\end{equation*}
of isomorphisms of vector bundles over $Z\times_M Z \times_M Z$ is
commutative.
\end{list}
\end{definition}

These 1-morphisms are generalizations of so-called \emph{stable
isomorphisms} \cite{murray2}. They are generalized in two aspects:
we admit vector bundles of rank possibly higher than 1 (this makes
it possible to describe gerbe modules by morphisms), and these
vector bundles live over a more general space $Z$ than
just the fibre product $Y_1 \times_M Y_2$ (this makes the
composition of morphisms easier).

A 1-morphism $\mathcal{A}:\mathcal{G} \to \mathcal{G}'$ requires
that the curvatures of the bundle gerbes $\mathcal{G}$ and
$\mathcal{G}'$ coincide. This follows from axiom (1M1) and the fact
that the trace of the curvature of a vector bundle is a closed form.

\begin{example}
\label{ex3}
Every bundle gerbe $\mathcal{G}$ has an associated 1-morphism
\begin{equation*}
\id_{\mathcal{G}}:\mathcal{G} \to \mathcal{G}
\end{equation*}
defined by the identity surjective
submersion $\id_Z$ of $Z:= Y^{[2]}$, the line bundle $A:=L$ of the bundle
gerbe $\mathcal{G}$ itself, and the isomorphism
\begin{equation*}
\alxydim{@C=1.5cm}{\pi_{13}^{*}L \otimes \pi_{34}^{*}L \ar[r]^-{\pi_{134}^{*}\mu} & \pi_{14}^{*}L
\ar[r]^-{\pi_{124}^{*}\mu^{-1}} & \pi_{12}^{*}L \otimes \pi_{24}^{*}L}
\end{equation*}
of line bundles over $Z \times_M Z = Y^{[4]}$, where we have identified $\zeta_1=\pi_{12}$
and $\zeta_2=\pi_{34}$. The axioms for this 1-morphism follow from the axioms
of the bundle gerbe $\mathcal{G}$.
\end{example}

The 2-categorial aspects of the theory of bundle gerbes enter when
one wants to compare two 1-morphisms.

\begin{definition}
\label{defmorph} Let $\mathcal{A}:\mathcal{G}_1 \to \mathcal{G}_2$
and $\mathcal{A}':\mathcal{G}_1 \to \mathcal{G}_2$ be 1-morphisms
between bundle gerbes over $M$. A \uline{2-morphism}
\begin{equation*}
\beta: \mathcal{A}
\Rightarrow \mathcal{A}'
\end{equation*}
is a  surjective submersion $\omega:W \to Z_1 \times_P Z_2$, where
$P:=Y_1 \times_M Y_2$, together with a morphism $\beta_W: A_1 \to
A_2$ of vector bundles over $W$, such that the diagram
\begin{equation}
\label{4}
\alxydim{@C=1.5cm@R=1.2cm}{
L_1 \otimes \omega_2^{*}A_1 \ar[r]^-{\alpha_1}
\ar[d]_{\id \otimes \omega_2^{*} \beta_W}
& \omega_1^{*}A_1  \otimes L_2
\ \ar[d]^{\omega_1^{*}\beta_W\otimes
\id}  \\ L_1
\otimes \omega_2^{*}A_2 \ar[r]_-{\alpha_2}
& \omega_1^{*}A_2  \otimes L_2}
\end{equation}
of morphisms of vector bundles over $W \times_M W$ is commutative.
We shall often omit the subscript in $\beta_W$ if it is clear
from the context that the notation refers to the bundle isomorphism.

\end{definition}

Due to technical reasons, one has to define a certain equivalence
relation on the space of 2-morphisms \cite{waldorf1}, whose precise
form is not important for this article. A 2-morphism
$\beta:\mathcal{A} \Rightarrow \mathcal{A}'$ is invertible if and
only if the morphism $\beta_W$ of vector bundles is invertible.
This, in turn, is the case if and only if
the ranks of the vector bundles of $\mathcal{A}$ and $\mathcal{A}'$
coincide.

Bundle gerbes over $M$, 1-morphisms and 2-morphisms as defined above
form a strictly associative 2-category $\mathfrak{BGrb}(M)$
\cite{waldorf1}. We describe below what that
means. Most importantly for us, we can
compose 1-morphisms:

\begin{definition}
\label{def2}
The composition of two 1-morphisms $\mathcal{A}:\mathcal{G}_1 \to \mathcal{G}_2$
and $\mathcal{A}':\mathcal{G}_2 \to \mathcal{G}_3$ is the 1-morphism
\begin{equation*}
\mathcal{A}'\circ \mathcal{A}:\mathcal{G}_1 \to \mathcal{G}_3
\end{equation*}
defined by the following data: its surjective submersion is $\zeta: \tilde Z \to Y_1
\times_M Y_3$ with $\tilde Z:= Z \times_{Y_2} Z'$ and the canonical
projections to $Y_1$ and $Y_3$, its vector bundle over $\tilde Z$ is $\tilde A:= A \otimes
A'$, and its isomorphism is given by
\begin{equation*}
\alxydim{@R=1.2cm@C=0.3cm}{L_1 \otimes \tilde \zeta_2^{*}\tilde A \ar@{=}[r] & L_1 \otimes \zeta_2^{*}A \otimes
\zeta_2^{\prime*}A' \ar[d]^{\alpha
\otimes \id} & \\ &
\zeta_1^{*}A \otimes L_2 \otimes
\zeta_2^{\prime*}A' \ar[d]^{\id
\otimes \alpha'} & \\ &
\zeta_1^{*}A \otimes \zeta_1^{\prime*}A'
\otimes L_3 \ar@{=}[r] & \tilde
\zeta_1^{*}\tilde A \otimes L_3\text{.}}
\end{equation*}
\end{definition}
The axioms for this 1-morphism are easy to
check. If we tacitly assume the category of vector spaces to be
strictly monoidal, it turns out that the composition of 1-morphisms
defined in this manner is, indeed, strictly associative,
\begin{equation*}
(\mathcal{A}''\circ\mathcal{A}')\circ\mathcal{A} = \mathcal{A}''\circ(\mathcal{A}'\circ\mathcal{A})\text{.}
\end{equation*}
The simplicity of Definition \ref{def2} (compared, e.g., to the one
given in \cite{stevenson1}) and the strict associativity of the
composition of 1-morphisms are consequences of our generalized
definition of 1-morphisms. One can now show

\begin{proposition}[\cite{waldorf1}]
A 1-morphism $\mathcal{A}:\mathcal{G} \to \mathcal{G}'$ is
invertible, also called \emph{1-isomorphism}, if and only if its
vector bundle is of rank one.
\end{proposition}

In a 2-category, invertibility means that
there exists a 1-isomorphism $\mathcal{A}^{-1}$ acting in the
opposite direction, together with 2-isomorphisms
\begin{equation}
\label{18}
i_l: \mathcal{A}^{-1} \circ \mathcal{A} \Rightarrow \id_{\mathcal{G}}
\quad\text{ and }\quad
i_r: \id_{\mathcal{G}'} \Rightarrow  \mathcal{A} \circ \mathcal{A}^{-1}
\end{equation}
which satisfy certain coherence axioms \cite{waldorf1}.
The  2-category $\mathfrak{BGrb}(M)$ of bundle gerbes over $M$ also provides the following
structure:
\begin{itemize}
\item[a)]
the \emph{vertical composition} of  two 2-morphisms
$\beta_1:\mathcal{A} \Rightarrow \mathcal{A}'$ and
$\beta_2:\mathcal{A}' \Rightarrow \mathcal{A}''$ to a new 2-morphism
\begin{equation*}
\beta_2
\bullet\beta_1:\mathcal{A} \Rightarrow \mathcal{A}''
\end{equation*}
which is associative and has  units $\id_{\mathcal{A}}$ for any 1-morphism
$\mathcal{A}$.

\item[b)]
the \emph{horizontal composition} of two 2-morphisms
$\beta_{12}:\mathcal{A}_{12} \Rightarrow \mathcal{A}_{12}'$ and
$\beta_{23}:\mathcal{A}_{23}\Rightarrow \mathcal{A}_{23}'$ to a new
2-morphism
\begin{equation*}
\beta_{23}
\circ \beta_{12}:\mathcal{A}_{23}
\circ \mathcal{A}_{12} \Rightarrow
\mathcal{A}'_{23}
\circ \mathcal{A}'_{12}\text{,}
\end{equation*}
which is compatible with the vertical composition.

\item[c)]
natural 2-isomorphisms
\begin{equation}
\label{19}
\rho_{\mathcal{A}}: \id_{\mathcal{G}_2} \circ \mathcal{A}
\Rightarrow \mathcal{A}
\quad\text{ and }\quad
\lambda_{\mathcal{A}}:\mathcal{A}
\circ \id_{\mathcal{G}_1}
\Rightarrow \mathcal{A}
\end{equation}
associated to any 1-morphism $\mathcal{A}:\mathcal{G}_1 \to \mathcal{G}_2$,
which satisfy the equality
\begin{equation}
\label{16}
\id_{\mathcal{A}'} \circ \rho_{\mathcal{A}} = \lambda_{\mathcal{A}'} \circ
\id_{\mathcal{A}}\text{.}
\end{equation}
\end{itemize}

The 2-category of bundle gerbes has pullbacks: for every smooth map
$f: M \to N$, there is a strict 2-functor
\begin{equation}
\label{20}
f^{*}:\mathfrak{BGrb}(N) \to \mathfrak{BGrb}(M)\text{.}
\end{equation}
Thus, for any bundle gerbe $\mathcal{G}$, we have a pullback bundle
gerbe $f^{*}\mathcal{G}$; for any 1-morphism
$\mathcal{A}:\mathcal{G}_1 \to \mathcal{G}_2$, a pullback 1-morphism
$f^{*}\mathcal{A}:f^{*}\mathcal{G}_1 \to f^{*}\mathcal{G}_2$; and
for every 2-morphism $\beta:\mathcal{A} \Rightarrow \mathcal{A}'$, a
pullback 2-morphism $f^{*}\beta:f^{*}\mathcal{A} \Rightarrow
f^{*}\mathcal{A}'$. These pullbacks are essentially defined as
pullbacks of the surjective submersions and the structure thereon,
details can be found in \cite{waldorf1}. If a bundle gerbe
$\mathcal{G}$ has curvature $H$, its pullback
$f^{*}\mathcal{G}$ has curvature $f^{*}H$. The
strictness of the 2-functor (\ref{20}) means that
$f^{*}\id_{\mathcal{G}}=\id_{f^{*}\mathcal{G}}$ and
$f^{*}(\mathcal{A}'\circ \mathcal{A}) = f^{*}\mathcal{A}' \circ
f^{*}\mathcal{A}$ whenever $\mathcal{A}$ and $\mathcal{A}'$ are
composable 1-morphisms. If $g:X \to M$ is another map, we find $(f
\circ g)^{*}= g^{*}\circ f^{*}$.

In order to concentrate  on what we need in this article, we  define
the dual $\mathcal{G}^{*}$ of a bundle gerbe $\mathcal{G}$ without
emphasizing its role in the 2-categorial
context. $\mathcal{G}^{*}$ consists of the
same surjective submersion $\pi:Y \to M$ as $\mathcal{G}$, the
2-form $-C \in \Omega^2(Y)$, the line bundle $L^{*}$ over $Y^{[2]}$
and the inverse of the dual of the isomorphism  $\mu$, which is an
isomorphism
\begin{equation*}
\mu^{*-1}:\pi_{12}^{*}L^{*}\otimes \pi_{23}^{*}L^{*}\to
\pi_{13}^{*}L^{*}
\end{equation*}
of line bundles over $Y^{[3]}$. If $H$ is the curvature of
$\mathcal{G}$,  the curvature of $\mathcal{G}^{*}$ is $-H$. If
$\mathcal{A}:\mathcal{G}_1 \to \mathcal{G}_2$ is a 1-morphism,  we
define an \emph{adjoint} 1-morphism
\begin{equation*}
\mathcal{A}^{\dagger}: \mathcal{G}_1^{*} \to \mathcal{G}_2^{*}
\end{equation*}
in the following way: it consists of the same surjective submersion
$\zeta: Z \to Y_1 \times_M Y_2$ as $\mathcal{A}$, it has the vector
bundle $A^{*}$ over $Z$, and the isomorphism \begin{equation*}
\alpha^{*-1}: L_1^{*} \otimes \zeta_1^{*}A^{*} \to \zeta_2^{*}A^{*}
\otimes L_2^{*}
\end{equation*}
of vector bundles over $Z \times_M Z$. The axioms for this
1-morphism follow immediately from those for
$\mathcal{A}$. Finally, for a 2-\textit{iso}morphism
$\beta:\mathcal{A}_1 \Rightarrow \mathcal{A}_2$, we define an
\emph{adjoint} 2-isomorphism
\begin{equation*}
\beta^{\dagger}: \mathcal{A}_1^{\dagger} \Rightarrow \mathcal{A}_2^{\dagger}\text{.}
\end{equation*}
It has the same surjective submersion $\omega:W \to Z_1 \times_P
Z_2$ as $\beta$, and the isomorphism $\beta_W^{*-1}: A_1^{*} \to
A_2^{*}$ of vector bundles over $W$. Notice that all these
operations are strictly involutive:
\begin{equation}
\label{102}
\mathcal{G}^{**} = \mathcal{G}
\quad\text{, }\quad
\mathcal{A}^{\dagger\dagger}=\mathcal{A}
\quad\text{ and }\quad
\beta^{\dagger\dagger}=\beta\text{.}
\end{equation}

\begin{remark}
In the context of some more structures in the 2-category
$\mathfrak{BGrb}(M)$, as described in
\cite{waldorf1}, namely a duality 2-functor $()^{*}$ and a functor
assigning inverses $\mathcal{A}^{-1}$ to 1-isomorphisms
$\mathcal{A}$, and certain 2-morphisms
$\beta^{\#}:\mathcal{A}_2^{-1} \Rightarrow \mathcal{A}_1^{-1}$ to
2-isomorphisms $\beta: \mathcal{A}_1 \Rightarrow \mathcal{A}_2$, we
find $\mathcal{A}^{\dagger} = \mathcal{A}^{*-1}$ and
$\beta^{\dagger} = \beta^{\#-1}$.
\end{remark}

\subsection{Twisted-Equivariant Structures}

\label{sec4}

\newcommand{\act}{}

An \textit{orientifold group} $(\Gamma,\epsilon)$ for a smooth
manifold $M$ is a finite group $\Gamma$ acting smoothly on the left
on $M$, together with a group homomorphism $\epsilon: \Gamma \to
\Z_2 = \lbrace -1,1\rbrace$. We label the diffeomorphisms
implementing the action by the group elements
themselves, for instance $\gamma: M \to M$. Notice that $\gamma_2
\circ \gamma_1 = \gamma_2\gamma_1$.

Next, we define an action of the orientifold group
$(\Gamma,\epsilon)$ on bundle gerbes over $M$ and their 1- and
2-morphisms. The value $\epsilon(\gamma)$ indicates whether a group
element $\gamma\in\Gamma$ acts just by pullback along $\gamma^{-1}$
or also by  additionally taking adjoints. Explicitly, for a bundle
gerbe $\mathcal{G}$, we set
\begin{equation*}
\gamma\act\mathcal{G} :=
\begin{cases}(\gamma^{-1})^{*}\mathcal{G} & \text{if }\epsilon(\gamma)=1\ \\
(\gamma^{-1})^{*}\mathcal{G}^{*}
 & \text{if }\epsilon(\gamma)=-1\text{.}\ \\
\end{cases}
\end{equation*}
Similarly, for a 1-morphism $\mathcal{A}:\mathcal{G} \to
\mathcal{H}$, we have a 1-morphism
\begin{equation*}
\gamma\act\mathcal{A}: \gamma\act\mathcal{G} \to \gamma\act\mathcal{H}
\end{equation*}
defined by
\begin{equation*}
\gamma\act\mathcal{A}:=
\begin{cases}(\gamma^{-1})^{*}\mathcal{A} & \text{if }\epsilon(\gamma)=1\ \\
(\gamma^{-1})^{*}\mathcal{A}^{\dagger} & \text{if } \epsilon(\gamma)=-1\text{.}\ \\
\end{cases}
\end{equation*}
Finally, for a 2-\textit{iso}morphism $\beta:\mathcal{A} \Rightarrow
\mathcal{A}'$, we have a 2-isomorphism
\begin{equation*}
\gamma\act\beta:\gamma\act\mathcal{A}\Rightarrow
\gamma\act\mathcal{A}'
\end{equation*}
defined by
\begin{equation*}
\gamma\act\beta :=
\begin{cases}(\gamma^{-1})^{*}\beta & \text{if } \epsilon(\gamma)=1\\
(\gamma^{-1})^{*}\beta^{\dagger} & \text{if }\epsilon(\gamma)=-1\text{.}\ \\
\end{cases}
\end{equation*}
Note that our conventions and (\ref{102}) imply
\begin{equation*}
(\gamma_1\gamma_2)\act=\gamma_1\act\gamma_2\act\text{,}
\end{equation*}
so that $\gamma\act$ is a left action on gerbes and their 1- and 2-morphisms. We use the same notation for differential forms, i.e. $\gamma\act \omega:= \varepsilon(\gamma)(\gamma^{-1})^{*}\omega$ for any differential form $\omega$ on $M$. If $H$ is the curvature of a bundle gerbe $\mathcal{G}$,  the curvature of $\gamma\act\mathcal{G}$ is $\gamma\act H$.

\begin{definition}
\label{def1}
Let $(\Gamma,\epsilon)$ be an orientifold group for $M$ and let $\mathcal{G}$ be a bundle gerbe over $M$. A \emph{$(\Gamma,\epsilon)$-equivariant  structure} on  $\mathcal{G}$
consists of
1-isomorphisms
\begin{equation*}
\mathcal{A}_{\gamma}: \mathcal{G}
\to \gamma\act\mathcal{G}
\end{equation*}
for each $\gamma \in \Gamma$, and of 2-isomorphisms
\begin{equation*}
\varphi_{\gamma_1,\gamma_2}: \gamma_1\act\mathcal{A}_{\gamma_2}
\circ \mathcal{A}_{\gamma_1}
\Rightarrow \mathcal{A}_{\gamma_1\gamma_2}
\end{equation*}
for each pair $\gamma_1,\gamma_2\in
\Gamma$, such that the diagram
\begin{equation}
\label{5}
\alxydim{@C=2cm@R=1.2cm}{\gamma_1\act\gamma_2\act\mathcal{A}_{\gamma_3} \circ
\gamma_1\act\mathcal{A}_{\gamma_2}
\circ \mathcal{A}_{\gamma_1} \ar@{=>}[d]_{\gamma_1\act\varphi_{\gamma_2,\gamma_3}
 \circ \id} \ar@{=>}[r]^-{\id \circ \varphi_{\gamma_1,\gamma_2}}
& \gamma_1\act\gamma_2\act\mathcal{A}_{\gamma_3}
\circ \mathcal{A}_{\gamma_1\gamma_2} \ar@{=>}[d]^{\varphi_{\gamma_1\gamma_2,\gamma_3}}
\\ \gamma_1\act\mathcal{A}_{\gamma_2\gamma_3} \circ
\mathcal{A}_{\gamma_1}
\ar@{=>}[r]_-{\varphi_{\gamma_1,\gamma_2\gamma_3}}
& \mathcal{A}_{\gamma_1\gamma_2\gamma_3}}
\end{equation}
of 2-isomorphisms is commutative. \end{definition}

We call a bundle gerbe $\mathcal{G}$ with $(\Gamma,\epsilon)$-equivariant structure a\emph{ $(\Gamma,\epsilon)$-equivariant bundle gerbe} or twisted-equivariant bundle gerbe.
 If $\epsilon$ is constant,  a $(\Gamma,\epsilon)$-equivariant bundle gerbe $\mathcal{G}$
is just called a \emph{$\Gamma$-equivariant bundle gerbe}. The
curvature $H$ of a twisted-equivariant bundle gerbe satisfies
$\gamma\act H=H$.  A twisted-equivariant structure on a bundle gerbe
$\mathcal{G}$ is called \emph{normalized} if the following choices
concerning the neutral group element $1\in \Gamma$ have been
made:
\begin{itemize}
\item[(a)]
the 1-isomorphism $\mathcal{A}_1: \mathcal{G} \to \mathcal{G}$ is
the identity 1-isomorphism $\id_{\mathcal{G}}$;

\item[(b)]
the 2-isomorphism $\varphi_{1,\gamma}: \mathcal{A}_\gamma \circ
\id_{\mathcal{G}} \Rightarrow \mathcal{A}_\gamma$ is the natural
2-isomorphism $\lambda_{\mathcal{A}_\gamma}$ from the 2-category of
bundle gerbes;

\item[(c)]
accordingly, the 2-isomorphism $\varphi_{\gamma,1}:
\id_{\gamma\act\mathcal{G}} \circ \mathcal{A}_{\gamma} \Rightarrow
\mathcal{A}_\gamma$ is the natural 2-isomorphism
$\rho_{\mathcal{A}_\gamma}$ from the 2-category of bundle gerbes.
\end{itemize}
Bundle gerbes with normalized twisted-equivariant structures will
give rise to elements in normalized group cohomology, as we shall
see in Section \ref{sec3}. A twisted-equivariant structure on a
bundle gerbe $\mathcal{G}$ is called \emph{descended} if all
surjective submersions, i.e. the surjective submersions
$\zeta^{\gamma}$ of the 1-isomorphisms $\mathcal{A}_{\gamma}$ and
the surjective submersions $\omega^{\gamma_1,\gamma_2}$ of the
2-isomorphisms $\varphi_{\gamma_1,\gamma_2}$, are identities.
This will be important in Section \ref{sec2}.

\begin{example}
\label{ex1} As an example, let us equip the trivial bundle gerbe
$\mathcal{I}_{\omega}$ from Example \ref{ex2} with a twisted-equivariant structure, for any orientifold group $(\Gamma,\epsilon)$
of $M$. This is possible for 2-forms $\omega\in\Omega^2(M)$ with
$\gamma\act\omega=\omega$ for all $\gamma\in\Gamma$. Since then
$\gamma\act\mathcal{I}_{\omega}=\mathcal{I}_{\gamma\act\omega}=\mathcal{I}_{\omega}$,
we may choose $\mathcal{A}_{\gamma}:=\id_{\mathcal{I}_{\omega}}$ for
all $\gamma\in\Gamma$. Accordingly, we can also choose
\begin{equation*}
\varphi_{\gamma_1,\gamma_2}:=
\rho_{\id_{\mathcal{I}_{\omega}}}=\lambda_{\id_{\mathcal{I}_{\omega}}}:
\id_{\mathcal{I}_{\omega}} \circ \id_{\mathcal{I}_{\omega}}
\Rightarrow \id_{\mathcal{I}_{\omega}}\text{.}
\end{equation*}
Diagram (\ref{5}) commutes due to condition (\ref{16}). We denote
this canonical $(\Gamma,\epsilon)$-equivariant structure by
$\mathcal{J}_{\omega}$. It is normalized and descended.
\end{example}

We recall that there exist canonical bundle gerbes over all compact
simple Lie groups \cite{meinrenken1,gawedzki2}. All normalized
twisted-equivariant structures on these canonical bundle gerbes were
classified (up to equivalence defined below) in \cite{gawedzki6}
using cohomological considerations (see also Section \ref{sec3}):
they arise in numbers ranging from two to sixteen. The corresponding
geometrical constructions will appear in \cite{gawedzki7}.

Let us formulate the definition of a $(\Gamma,\epsilon)$-equivariant
structure in terms of line bundles and their isomorphisms. For
convenience, we assume the $(\Gamma,\epsilon)$-equivariant structure
to be descended (see also Lemma \ref{lem1} below). The pullback
bundle gerbe $(\gamma^{-1})^{*}\mathcal{G}$ has the surjective
submersion $\pi_{\gamma}:Y_\gamma \to M$ in the commutative pullback
diagram
\begin{equation*}
\alxydim{@=1.2cm}{Y_\gamma\ \ar[r] \ar[d]_{\pi_{\gamma}}
& Y \ar[d]^{\pi} \\ M \ar[r]_{\gamma^{-1}}
& M}\text{,}
\end{equation*}
with $Y_{\gamma}:=Y$ and  $\pi_{\gamma}:=\gamma
\circ \pi$, and the rest of the
data is the same as for $\mathcal{G}$.
The 1-isomorphism $\mathcal{A}_{\gamma}:\mathcal{G}\to
\gamma\act\mathcal{G}$ is  now a line
bundle $A_{\gamma}$ over $Z^{\gamma}:=
Y \times_M Y_\gamma$
and  an isomorphism
\begin{equation}
\label{3}
\alpha_{\gamma}: \pi_{13}^{*}L
\otimes \pi_{34}^{*}A_{\gamma} \to \pi_{12}^{*}A_{\gamma}
\otimes \pi_{24}^{*}L^{\epsilon(\gamma)}
\end{equation}
of line bundles over $Z^{\gamma} \times_M Z^{\gamma}$,
satisfying the compatibility axiom
(1M2), namely
\begin{equation}
\label{6}
\alxydim{@C=3cm@R=1.2cm}{\pi_{13}^{*}L \otimes \pi_{35}^{*}L \otimes \pi_{56}^{*}A_{\gamma}
\ar[r]^-{\pi_{135}^{*}\mu \otimes \id} \ar[d]_{\id \otimes \pi_{3456}^{*}\alpha_{\gamma}}
& \pi_{15}^{*}L \otimes \pi_{56}^{*}A_{\gamma} \ar[dd]^{\pi_{1256}^{*}\alpha_{\gamma}}\\  \pi_{13}^{*}L \otimes \pi_{34}^{*}A_{\gamma} \otimes \pi_{46}^{*}L^{\epsilon(\gamma)}
\ar[d]_{\pi_{1234}^{*}\alpha_{\gamma} \otimes \id} & \\ \pi_{12}^{*}A_{\gamma}
\otimes \pi_{24}^{*}L^{\epsilon(\gamma)} \otimes \pi_{46}^{*}L^{\epsilon(\gamma)} \ar[r]_-{\id \otimes \pi_{246}^{*}\mu^{\epsilon(\gamma)}} & \pi_{12}^{*}A_{\gamma}
\otimes \pi_{26}^{*}L^{\epsilon(\gamma)}\text{.}}
\end{equation}
Here, and in the following, we have regarded the fibre products of
$Z^{\gamma}$ with itself as a subset of $Y^{4}$ and $Y^{6}$
respectively, and used the projections $\pi_{ij}:Y^{k} \to Y^{2}$
carefully: in (\ref{3}) we have well-defined projections
$\pi_{12},\pi_{34}:Z^{\gamma}  \times _M Z^{\gamma}\to Z^{\gamma}$
and $\pi_{13},\pi_{24}:Z^{\gamma} \times_M Z^{\gamma} \to Y^{[2]}$,
and similarly in (\ref{6}). Furthermore,
$L^{\epsilon(\gamma)}$ stands for the dual line bundle $L^*$ and
$\mu^{\epsilon(\gamma)}$ for the isomorphism $\mu^{*-1}$  if
$\epsilon(\gamma)=-1$.

The  surjective submersion of the  1-isomorphism
$\gamma_1\act\mathcal{A}_{\gamma_2}$ is the identity on $Z^{\gamma_2}_{\gamma_1}:= Y_{\gamma_1}
\times_M
Y_{\gamma_1\gamma_2}$,  whose projection to the base space $M$ makes the diagram
\begin{equation*}
\alxydim{@=1.2cm}{Z^{\gamma_2}_{\gamma_1} \ar[d] \ar@{=}[r] & Z^{\gamma_2} \ar[d] \\
M \ar[r]_{\gamma_1^{-1}} & M}
\end{equation*}
 commutative. Further, $\gamma_1\act\mathcal{A}_{\gamma_2}$
consists of the line bundle $A_{\gamma_2}^{\epsilon(\gamma_1)}$ over
$Z_{\gamma_1}^{\gamma_2}$, and of the isomorphism
$\alpha_{\gamma_2}^{\epsilon(\gamma_1)}$, which stands for
$\alpha_{\gamma_2}^{*-1}$ if $\epsilon(\gamma_1)=-1$. Next, applying
Definition \ref{def2} for the composition of
two 1-morphisms to $\gamma_1\act\mathcal{A}_{\gamma_2} \circ
\mathcal{A}_{\gamma_1}$, we have to form the fibre product
\begin{equation*}
Z^{\gamma_1,\gamma_2}:=
  Z^{\gamma_1}\times_{Y_{\gamma_1}}Z^{\gamma_2}_{\gamma_1} \cong
Y \times_M
Y_{\gamma_1} \times_M
Y_{\gamma_1\gamma_2}
\end{equation*}
with the surjective submersion $\pi_{13}:Z^{\gamma_1,\gamma_2}
\to Y \times_M Y_{\gamma_1\gamma_2}$. The line bundle of the composition $\gamma_1\act\mathcal{A}_{\gamma_2}
\circ \mathcal{A}_{\gamma_1}$  is the line bundle  $\pi_{12}^{*}A_{\gamma_1}
\otimes \pi_{23}^{*}A^{\epsilon(\gamma_1)}_{\gamma_2}$ over $Z^{\gamma_1,\gamma_2}$, and its isomorphism
is
\begin{multline*}
(\id \otimes \pi_{2356}^{*}\alpha^{\epsilon(\gamma_1)}_{\gamma_2}) \circ (\pi_{1245}^{*}\alpha_{\gamma_1}\otimes \id):\pi_{14}^{*}L \otimes \pi_{45}^{*}A_{\gamma_1} \otimes \pi_{56}^{*}A_{\gamma_2}^{\epsilon(\gamma_1)}
\\\to \pi_{12}^{*}A_{\gamma_1} \otimes \pi_{23}^{*}A_{\gamma_2}^{\epsilon(\gamma_1)} \otimes \pi_{36}^{*}L^{\epsilon(\gamma_1\gamma_2)}
\end{multline*}

Finally, we come to the 2-isomorphisms
$\varphi_{\gamma_1,\gamma_2}$, whose surjective submersion $\omega$
is by assumption  the identity on $Z^{\gamma_1,\gamma_2}$, so that
they induce the isomorphisms
\begin{equation*}
\varphi_{\gamma_1,\gamma_2}:\pi_{12}^{*}A_{\gamma_1}
\otimes \pi_{23}^{*}A_{\gamma_2}^{\epsilon(\gamma_1)}
\to \pi_{13}^{*}A_{\gamma_1\gamma_2}
\end{equation*}
of line bundles over $Z^{\gamma_1,\gamma_2}$ satisfying the
compatibility condition (\ref{4}) for 2-morphisms, which, here,
amounts to the commutativity of the diagram (turned by 90 degrees
compared to (\ref{4}) for presentational reasons)
\begin{equation}
\label{8}
\alxydim{@C=2cm@R=1.2cm}{\pi_{14}^{*}L \otimes \pi_{45}^{*}A_{\gamma_1} \otimes \pi_{56}^{*}A_{\gamma_2}^{\epsilon(\gamma_1)}
\ar[r]^-{\id \otimes \pi_{456}^{*}\varphi_{\gamma_1,\gamma_2}}
\ar[d]_-{\pi_{1245}^{*}\alpha_{\gamma_1}\otimes \id} & \pi_{14}^{*}L\otimes
\pi_{46}^{*}A_{\gamma_1\gamma_2}
\ar[dd]^{\pi_{1346}^{*}\alpha_{\gamma_1\gamma_2}}
\\ \pi_{12}^{*}A_{\gamma_1} \otimes \pi_{25}^{*}L^{\epsilon(\gamma_1)}
\otimes \pi_{56}^{*}A_{\gamma_2}^{\epsilon(\gamma_1)}
\ar[d]_{\id \otimes \pi_{2356}^{*}\alpha^{\epsilon(\gamma_1)}_{\gamma_2}}
& \\ \pi_{12}^{*}A_{\gamma_1} \otimes \pi_{23}^{*}A_{\gamma_2}^{\epsilon(\gamma_1)} \otimes \pi_{36}^{*}L^{\epsilon(\gamma_1\gamma_2)} \ar[r]_-{\pi_{123}^{*}\varphi_{\gamma_1,\gamma_2}
\otimes \id}
&  \pi_{13}^{*}A_{\gamma_1\gamma_2}
\otimes \pi_{36}^{*}L^{\epsilon(\gamma_1\gamma_2)}}
\end{equation}
of isomorphisms of line bundles over $Z^{\gamma_1,\gamma_2} \times_M
Z^{\gamma_1,\gamma_2}$. The commutativity of diagram (\ref{5})
from Definition \ref{def1} is
equivalent to that of the diagram
\begin{equation}
\label{7}
\alxydim{@C=2.2cm@R=1.2cm}{\pi_{12}^{*}A_{\gamma_1}
\otimes \pi_{23}^{*}A_{\gamma_2}^{\epsilon(\gamma_1)}
\otimes \pi_{34}^{*}A^{\epsilon(\gamma_1\gamma_2)}_{\gamma_3}
 \ar[d]_{\pi_{123}^{*}\varphi_{\gamma_1,\gamma_2}
\otimes \id} \ar[r]^-{\id \otimes \pi_{234}^{*}\varphi^{\epsilon(\gamma_1)}_{\gamma_2,\gamma_3}}
& \pi_{12}^{*}A_{\gamma_1}
\otimes \pi_{24}^{*}A^{\epsilon(\gamma_1)}_{\gamma_2\gamma_3}
 \ar[d]^{\pi_{124}^{*}\varphi_{\gamma_1,\gamma_2\gamma_3}}
\\ \pi_{13}^{*}A_{\gamma_1\gamma_2}
\otimes \pi_{34}^{*}A_{\gamma_3}^{\epsilon(\gamma_1\gamma_2)}
\ar[r]_-{\pi_{134}^{*}\varphi_{\gamma_1\gamma_2,\gamma_3}}
& \pi_{14}^{*}A_{\gamma_1\gamma_2\gamma_3}}
\end{equation}
of isomorphisms of line bundles over
$Z^{\gamma_1,\gamma_2,\gamma_3}\cong Y\times_M Y_{\gamma_1} \times_M
Y_{\gamma_1\gamma_2} \times_M Y_{\gamma_1\gamma_2\gamma_3}$.

Summarizing, a descended $(\Gamma,\epsilon)$-equivariant  structure
on the bundle gerbe $\mathcal{G}$ is
\begin{enumerate}
\item
a line bundle $A_{\gamma}$ over $Z^{\gamma}$ of curvature
$\mathrm{curv}(A_{\gamma}) = \epsilon(\gamma)\pi_{2}^{*}C -
\pi_1^{*}C$ for each $\gamma\in\Gamma$;
\item
for each $\gamma\in\Gamma$, an isomorphism
\begin{equation*}
\alpha_{\gamma}: \pi_{13}^{*}L
\otimes \pi_{34}^{*}A_{\gamma} \to \pi_{12}^{*}A_{\gamma}
\otimes \pi_{24}^{*}L^{\epsilon(\gamma)}
\end{equation*}
of line bundles over $Z^{\gamma} \times_M Z^{\gamma}$ such that
the diagram (\ref{6}) is commutative;

\item
for each pair $(\gamma_1,\gamma_2)\in\Gamma\times\Gamma$, an
isomorphism
\begin{equation*}
\varphi_{\gamma_1,\gamma_2}:\pi_{12}^{*}A_{\gamma_1}
\otimes \pi_{23}^{*}A_{\gamma_2}^{\epsilon(\gamma_1)}
\to \pi_{13}^{*}A_{\gamma_1\gamma_2}
\end{equation*}
of line bundles over $Z^{\gamma_1,\gamma_2}$ such that
the diagrams (\ref{8}) and (\ref{7}) are commutative.

\end{enumerate}
If the $(\Gamma,\epsilon)$-equivariant  structure
is normalized, we have
 $A_1:=L$ and
 the isomorphism $\alpha_{1}:=\pi_{124}^{*}\mu^{-1} \circ
\pi_{134}^{*}\mu$. The normalization constraints $\varphi_{\gamma,1}=\rho_{\mathcal{A}_{\gamma}}$
and $\varphi_{1,\gamma}=\lambda_{\mathcal{A}_{\gamma}}$ imply $\varphi_{1,1}=\mu$.

Next, we would like to compare two $(\Gamma,\epsilon)$-equivariant
bundle gerbes.

\begin{definition}
\label{def4}
Let $(\Gamma,\epsilon)$ be an orientifold group for $M$ and let $\mathcal{G}^a$ and $\mathcal{G}^b$ be bundle gerbes over $M$ equipped with $(\Gamma,\epsilon)$-equivariant  structures $\mathcal{J}^a=(\mathcal{A}_{\gamma}^a,\varphi_{\gamma_1,\gamma_2}^a)$ and $\mathcal{J}^b=(\mathcal{A}^b_{\gamma},\varphi^b_{\gamma_1,\gamma_2})$, respectively.
An \emph{equivariant 1-morphism}
\begin{equation*}
(\mathcal{B},\eta_{\gamma}): (\mathcal{G}^a,\mathcal{J}^a) \to (\mathcal{G}^b,\mathcal{J}^b)
\end{equation*}
is a 1-morphism $\mathcal{B}:\mathcal{G}^a \to \mathcal{G}^b$ of the underlying bundle gerbes together with a family of 2-isomorphisms
\begin{equation*}
\eta_{\gamma}: \gamma\act\mathcal{B} \circ \mathcal{A}^a_{\gamma} \Rightarrow \mathcal{A}^b_{\gamma} \circ \mathcal{B}\text{,}
\end{equation*}
one for each $\gamma\in\Gamma$,
such that the diagram
\begin{equation}
\label{17}
\alxydim{@C=3cm@R=1.2cm}{\gamma_1\act\gamma_2\act\mathcal{B} \circ \gamma_1\act\mathcal{A}^a_{\gamma_2} \circ \mathcal{A}^a_{\gamma_1} \ar@{=>}[d]_-{\gamma\act_1\eta_{\gamma_2} \circ \id_{\mathcal{A}^a_{\gamma_1}}} \ar@{=>}[r]^-{\id_{\gamma_1\act\gamma_2\act\mathcal{B}} \circ \varphi^a_{\gamma_1,\gamma_2}} & \gamma_1\gamma_2\act\mathcal{B} \circ \mathcal{A}^a_{\gamma_1\gamma_2} \ar@{=>}[dd]^-{\eta_{\gamma_1\gamma_2}} \\ \gamma\act_1\mathcal{A}_{\gamma_2}^{b}\circ \gamma_1\act \mathcal{B}\circ \mathcal{A}^a_{\gamma_1} \ar@{=>}[d]_-{\id_{\gamma\act_1\mathcal{A}_{\gamma_2}^b}\circ \eta_{\gamma_1}} & \\ \gamma_1\act\mathcal{A}^{b}_{\gamma_2}\circ \mathcal{A}^{b}_{\gamma_1} \circ \mathcal{B} \ar@{=>}[r]_-{\varphi^b_{\gamma_1,\gamma_2} \circ \id_{\mathcal{B}}} & \mathcal{A}^b_{\gamma_1\gamma_2} \circ \mathcal{B}}
\end{equation}
of 2-isomorphisms is commutative.
\end{definition}

Equivariant 1-morphisms  can be composed in a natural way: if
\begin{equation*}
\alxydim{@C=1.5cm}{(\mathcal{G}^a,\mathcal{J}^a) \ar[r]^{(\mathcal{B},\eta_{\gamma})} & (\mathcal{G}^b,\mathcal{J}^b) \ar[r]^{(\mathcal{B}',\eta_{\gamma}')} & (\mathcal{G}^c,\mathcal{J}^c)}
\end{equation*}
are two composable equivariant 1-morphisms, their composition consists of   the 1-morphism $\mathcal{B}' \circ \mathcal{B}: \mathcal{G}^a \to \mathcal{G}^c$ and  the 2-morphisms
\begin{equation*}
\alxydim{@C=1.1cm}{\gamma\act(\mathcal{B}' \circ \mathcal{B}) \circ \mathcal{A}^a_{\gamma}=\gamma\act\mathcal{B}' \circ \gamma\act\mathcal{B} \circ \mathcal{A}^a_{\gamma}\ar@{=>}[r]^-{\id \circ \eta_{\gamma}} & \gamma\act\mathcal{B}' \circ \mathcal{A}^b_{\gamma}\circ \mathcal{B} \ar@{=>}[r]^-{\eta'_{\gamma} \circ \id} & \mathcal{A}^c_{\gamma} \circ (\mathcal{B}' \circ \mathcal{B})}\text{.}
\end{equation*}
This composition is associative. We also have an identity
equivariant 1-morphism associated to a
$(\Gamma,\epsilon)$-equivariant bundle gerbe
$(\mathcal{G},\mathcal{J})$ given by $(\id_{\mathcal{G}},
\lambda_{\mathcal{A}_{\gamma}}^{-1} \bullet
\rho_{\mathcal{A}_{\gamma}})$. An equivariant 1-morphism
$(\mathcal{B},\eta_{\gamma})$ is called \emph{invertible} or
equivariant 1-\emph{iso}morphism if the 1-morphism $\mathcal{B}$ is
invertible. In this case, an inverse is given by
$(\mathcal{B}^{-1},\eta_{\gamma}^{-1})$. Hence, equivariant
1-isomorphisms furnish an equivalence relation on the set of
$(\Gamma,\epsilon)$-equivariant bundle gerbes over $M$.

\begin{definition}
Two twisted-equivariant bundle gerbes over $M$ are called
\emph{equivalent} if there exists an equivariant 1-isomorphism
between them.
\end{definition}

 The set of equivalence classes of twisted-equivariant bundle gerbes over $M$ will be further investigated in Sections \ref{sec2} and
 \ref{sec3}. Let us anticipate here the following fact.

\begin{lemma}
\label{lem1}
Every twisted-equivariant bundle gerbe is equivalent to one with descended twisted-equivariant structure.
\end{lemma}

\begin{proof}
We recall Theorem 1 of \cite{waldorf1}: for every 1-morphism
$\mathcal{A}:\mathcal{G} \to \mathcal{H}$, there exists a
``descended'' 1-morphism $\mathrm{Des}(\mathcal{A}):\mathcal{G} \to
\mathcal{H}$ whose surjective submersion $\zeta:Z \to Y \times_M Y'$
is the identity, together with a 2-isomorphism
$\sigma_{\mathcal{A}}: \mathcal{A} \Rightarrow
\mathrm{Des}(\mathcal{A})$. For every 2-morphism
$\varphi:\mathcal{A}^a \Rightarrow \mathcal{A}^b$, there exists a
2-morphism
\begin{equation*}
\mathrm{Des}(\varphi):\mathrm{Des}(\mathcal{A}^a) \Rightarrow \mathrm{Des}(\mathcal{A}^b)
\end{equation*}
such that the diagram
\begin{equation}
\label{45}
\alxydim{@=1.2cm}{\mathcal{A}^a \ar@{=>}[r]^-{\sigma_{\mathcal{A}^a}} \ar@{=>}[d]_{\varphi} & \mathrm{Des}(\mathcal{A}^a) \ar@{=>}[d]^{\mathrm{Des}(\varphi)} \\ \mathcal{A}^b \ar@{=>}[r]_-{\sigma_{\mathcal{A}^b}} & \mathrm{Des}(\mathcal{A}^b)}
\end{equation}
is commutative. For a given
$(\Gamma,\varepsilon)$-equivariant structure
$\mathcal{J}=(\mathcal{A}_{\gamma},\varphi_{\gamma_1,\gamma_2})$ on
a bundle gerbe $\mathcal{G}$, we define $\mathcal{A}'_{\gamma} :=
\mathrm{Des}(\mathcal{A}_{\gamma})$ and
$\varphi_{\gamma_1,\gamma_2}' :=
\mathrm{Des}(\varphi_{\gamma_1,\gamma_2})$. Due to the commutativity
of (\ref{45}), the new $\varphi'_{\gamma_1,\gamma_2}$ still
satisfy condition (\ref{5}) for
$(\Gamma,\epsilon)$-equivariant structures. Then, the choices
$\mathcal{B}=\id_{\mathcal{G}}$ and
$\eta_{\gamma}=\sigma_{\mathcal{A}_{\gamma}}$  define an
equivariant 1-isomorphism which establishes the claimed equivalence.
\end{proof}

We call an equivariant 1-morphism between bundle gerbes with
normalized $(\Gamma,\epsilon)$-equivariant  structures
\emph{normalized} if $\eta_1: \mathcal{B} \circ \id_{\mathcal{G}^a}
\Rightarrow \id_{\mathcal{G}^b} \circ \mathcal{B}$ is given by the
natural 2-morphisms of the 2-category, $\eta_1=
\rho^{-1}_{\mathcal{B}} \bullet \lambda_{\mathcal{B}}$. We call
an equivariant 1-morphism
\emph{descended} if the surjective submersion of $\mathcal{B}$ is
the identity.

Let us, again, describe
what an equivariant 1-morphism $(\mathcal{B},\eta_{\gamma})$ is in
terms of line bundles and isomorphisms thereof. We assume it to be
descended for simplicity. The 1-isomorphism $\mathcal{B}:
\mathcal{G}^a \to \mathcal{G}^b$ consists of a vector bundle $B$
over $Z:= Y^a \times_M Y^b$ and of an isomorphism $\beta:
\pi_{13}^{*}L^a \otimes \pi_{34}^{*}B \to \pi_{12}^{*}B \otimes
\pi_{24}^{*}L^b$ over $Z^{[2]}$ satisfying axiom
(1M1). The composition $\gamma\act\mathcal{B}
\circ \mathcal{A}^a_{\gamma}$ we have to consider is the 1-morphism
with the vector bundle $\pi_{12}^{*}A^a_{\gamma} \otimes
\pi_{23}^{*}B^{\epsilon(\gamma)}$ over
\begin{equation*}
Z_1^{\gamma} := (Z^a)^{\gamma} \times_{Y^a_{\gamma}} Z_\gamma \cong Y^a
\times_M Y^a_{\gamma} \times_M Y^b_{\gamma}\text{,}
\end{equation*}
and with the isomorphism
\begin{multline}
\label{30} (\id \otimes \pi_{2356}^{*}\beta^{\epsilon(\gamma)})
\circ (\pi_{1245}^{*}\alpha^a_{\gamma} \otimes \id): \pi_{14}^{*}L^a
\otimes \pi_{45}^{*}A^a_{\gamma} \otimes
\pi_{56}^{*}B^{\epsilon(\gamma)}\\ \to \pi_{12}^{*}A^a_{\gamma}
\otimes \pi_{23}^{*}B^{\epsilon(\gamma)} \otimes
\pi_{36}^{*}(L^b)^{\epsilon(\gamma)}
\end{multline}
of vector bundles over
$(Z_1^\gamma)^{[2]}$. The other composition, $\mathcal{A}^b_{\gamma}
\circ \mathcal{B}$, is the 1-morphism with the vector bundle
$\pi_{12}^{*}B\otimes \pi_{23}^{*}A_{\gamma}^b$ over
\begin{equation*}
Z_2^{\gamma} := Z \times_{Y^b} (Z^b)^{\gamma} \cong Y^a \times_M Y^b
\times_M Y^b_{\gamma}\text{,}
\end{equation*}
and with the isomorphism
\begin{multline}
\label{31} (\id \otimes \pi_{2356}^{*}\alpha_{\gamma}^{b}) \circ
(\pi_{1245}^{*}\beta \otimes \id): \pi_{14}^{*}L^{a} \otimes
\pi_{45}^{*}B \otimes \pi_{56}^{*}A^b_{\gamma} \\\to \pi_{12}^{*}B
\otimes \pi_{23}^{*}A^{b}_{\gamma} \otimes
\pi_{36}^{*}(L^b)^{\epsilon(\gamma)}
\end{multline}
of vector bundles over
$(Z^{\gamma}_2)^{[2]}$.
 The 2-isomorphisms $\eta_{\gamma}$ correspond now to isomorphisms
\begin{equation}
\label{34}
\eta_{{\gamma}}: \pi_{12}^{*}A^a_{\gamma} \otimes
\pi_{24}^{*}B^{\epsilon(\gamma)}\to \pi_{13}^{*}B\otimes
\pi_{34}^{*}A^b_{\gamma}
\end{equation}
of vector bundles over $Z_1^{\gamma} \times_{P} Z_2^{\gamma}\cong
Y^a\times_MY^a_\gamma\times_MY^b\times_MY^b_\gamma$, \,where
$P:=Y^{a} \times_M Y^b_{\gamma}$, and these isomorphisms satisfy the
compatibility condition
\begin{equation}
\label{32} \alxydim{@=1.2cm}{\pi_{15}^{*}L^a \otimes
\pi_{56}^{*}A^a_{\gamma} \otimes \pi_{68}^{*}B^{\epsilon(\gamma)}
\ar[d]_{\id \otimes \pi_{5678}^{*}\eta_{{\gamma}}} \ar[r] &
\pi_{12}^{*}A^a_{\gamma} \otimes \pi_{24}^{*}B^{\epsilon(\gamma)}
\otimes \pi_{48}^{*}(L^b)^{\epsilon(\gamma)}
\ar[d]^{\pi_{1234}^{*}\eta_{{\gamma}} \otimes \id}
\\ \pi_{15}^{*}L^a \otimes \pi_{57}^{*}B \otimes \pi_{78}^{*}A^b_{\gamma}
\ar[r]_{} & \pi_{13}^{*}B \otimes \pi_{34}^{*}A^b_{\gamma} \otimes
\pi_{48}^{*}(L^b)^{\epsilon(\gamma)}\text{,}}
\end{equation}
where the horizontal arrows are given by (\ref{30}) and (\ref{31}),
respectively. Finally, the commutativity of diagram (\ref{17})
implies the commutativity of the diagram
\begin{equation}
\hspace{-0.3cm}
\label{33} \alxydim{@C=1.9cm@R=1.2cm}{\pi_{12}^{*}A^a_{\gamma_1}
\otimes \pi_{23}^{*}(A_{\gamma_2}^a)^{\epsilon(\gamma_1)} \otimes
\pi_{36}^{*}B^{\epsilon(\gamma_1\gamma_2)}
\ar[r]^-{\pi_{123}^{*}\varphi^a_{\gamma_1,\gamma_2} \otimes \id}
\ar[d]_{\id \otimes
\pi_{2356}^{*}\eta^{\epsilon(\gamma_1)}_{{\gamma_2}}} &
\pi_{13}^{*}A^a_{\gamma_1\gamma_2} \otimes
\pi_{36}^{*}B^{\epsilon(\gamma_1\gamma_2)}
\ar[dd]^{\pi_{1346}^{*}\eta_{{\gamma_1\gamma_2}}} \\
\pi_{12}^{*}A^a_{\gamma_1} \otimes
\pi_{25}^{*}B^{\epsilon(\gamma_1)} \otimes
\pi_{56}^{*}(A_{\gamma_2}^b)^{\epsilon(\gamma_1)}
\ar[d]_{\pi_{1245}^{*}\eta_{\gamma_1} \otimes \id} & \\
\pi_{14}^{*}B\otimes \pi_{45}^{*}A^b_{\gamma_1}\otimes
\pi_{56}^{*}(A_{\gamma_2}^b)^{\epsilon(\gamma_1)} \ar[r]_-{\id
\otimes \pi_{456}^{*}\varphi^b_{{\gamma_1,\gamma_2}}} &
\pi_{14}^{*}B\otimes \pi_{46}^{*}A^b_{\gamma_1\gamma_2}\text{.}}
\hspace{-0.5cm}
\end{equation}

For completeness, and as a preparation for Section \ref{sec5}, we
would also like to introduce equivariant
2-morphisms. Suppose that we have $(\Gamma,\epsilon)$-equivariant
bundle gerbes $(\mathcal{G}^a,\mathcal{J}^a)$ and
$(\mathcal{G}^b,\mathcal{J}^b)$, and that we have two equivariant
1-morphisms $(\mathcal{B},\eta_{\gamma})$ and
$(\mathcal{B}',\eta_{\gamma}')$ between these. An \emph{equivariant
2-morphism} \begin{equation*} \phi: (\mathcal{B},\eta_{\gamma})
\Rightarrow (\mathcal{B}',\eta_{\gamma}')
\end{equation*}
is a 2-morphism $\phi:\mathcal{B} \Rightarrow \mathcal{B}'$ which is
compatible with the 2-morphisms $\eta_{\gamma}$ and
$\eta_{\gamma}'$ in the sense that the diagram
\begin{equation}
\label{35} \alxydim{@=1.2cm}{\gamma\act\mathcal{B} \circ
\mathcal{A}^a_{\gamma} \ar@{=>}[d]_{\gamma\act\phi \circ
\id_{\mathcal{A}^a_{\gamma}}} \ar@{=>}[r]^-{\eta_{\gamma}} &
\mathcal{A}^b_{\gamma} \circ \mathcal{B}
\ar@{=>}[d]^{\id_{\mathcal{A}^b_{\gamma}} \circ \phi} \\
\gamma\act\mathcal{B}' \circ \mathcal{A}^a_{\gamma}
\ar@{=>}[r]_{\eta_{\gamma}'} & \mathcal{A}^b_{\gamma} \circ
\mathcal{B}'}
\end{equation}
of 2-morphisms is commutative.

In terms of morphisms between vector bundles, $\phi$ is just a
morphism $\phi: B \to B'$ of vector bundles over $Z = Y^a \times_M
Y^b$ which is compatible with the isomorphisms $\beta$ and $\beta'$
in the sense that the diagram
\begin{equation}
\label{37}
\alxydim{@=1.2cm}{
\pi_{13}^{*}L^a \otimes
\pi_{34}^{*}B \ar[r]^-{\beta}
\ar[d]_{\id \otimes \pi_{34}^{*}\phi}
& \pi_{12}^{*}B \otimes \pi_{24}^{*}L^b
\ \ar[d]^{\pi_{12}^{*}\phi\otimes
\id}  \\ \pi_{13}^{*}L^a \otimes
\pi_{34}^{*}B' \ar[r]_-{\beta'}
& \pi_{12}^{*}B' \otimes \pi_{24}^{*}L^b}
\end{equation}
is commutative, and diagram (\ref{35}) imposes the commutativity of
\begin{equation}
\label{62}
\alxydim{@=1.2cm}{
\pi_{12}^{*}A^a_{\gamma} \otimes \pi_{24}^{*}B^{\epsilon(\gamma)} \ar[r]^-{\eta_{\gamma}}
\ar[d]_{\id \otimes \pi_{24}^{*}\phi^{\epsilon(\gamma)}}
& \pi_{13}^{*}B\otimes \pi_{34}^{*}A^b_{\gamma}
\ \ar[d]^{\pi_{13}^{*}\phi\otimes
\id}  \\ \pi_{12}^{*}A^{a}_{\gamma} \otimes \pi_{24}^{*}B'^{\epsilon(\gamma)} \ar[r]_-{\eta'_{\gamma}}
& \pi_{13}^{*}B'\otimes \pi_{34}^{*}A^b_{\gamma}\text{.}}
\end{equation}

Naturally, twisted-equivariant bundle gerbes, equivariant
1-morphisms and equivariant 2-morphisms form, again, a 2-category,
but we will not stress this point.

\subsection{Jandl Gerbes}

\label{sec15}

In this section, we consider the particular  orientifold group
$(\Z_2,\id)$.  The non-trivial group element of $\Z_2$ is denoted by
$k$, and its action $k:M \to M$ is an involution. According to
Definition \ref{def1}, a (normalized)  $\Z_2^{\id}$-equivariant
structure is a single 1-isomorphism
\begin{equation*}
\mathcal{A}_k:\mathcal{G}
\to k^{*}\mathcal{G}^{*}
\end{equation*}
and single 2-isomorphism
\begin{equation*}
\varphi_{k,k}:  k^{*}\mathcal{A}_k^{\dagger} \circ \mathcal{A}_k \Rightarrow
\id_{\mathcal{G}}
\end{equation*}
such that
\begin{equation}
\label{40}
\lambda_{\mathcal{A}_k} \bullet (\id \circ \varphi_{k,k})= \rho_{\mathcal{A}_k}
\bullet (k^{*}\varphi_{k,k}^{\dagger} \circ \id)\text{.}
\end{equation}
It is easy to see that this is exactly the same as a Jandl structure
\cite{schreiber1}: the 1-isomorphism $\mathcal{A}:=
k^{*}\mathcal{A}_k:k^{*}\mathcal{G} \to \mathcal{G}^{*}$ and the
2-isomorphism  $\varphi$ defined by
\begin{equation*}
\alxydim{@C=1.3cm}{k^{*}\mathcal{A} \ar@{=>}[r]^-{\rho_{k*\mathcal{A}}^{-1}} & \id_{\mathcal{G}}
\circ k^{*}\mathcal{A} \ar@{=>}[r]^-{i_r \circ \id_{k^{*}\mathcal{A}}} &  \mathcal{A}^{*} \circ \mathcal{A}^{\dagger} \circ k^{*}\mathcal{A} \ar@{=>}[r]^-{\id_{\mathcal{A}^{*}}
 \circ \varphi_{k,k}}
 & \mathcal{A}^{*} \circ \id_{\mathcal{G}} \ar@{=>}[r]^-{\lambda_{\mathcal{A}^{*}}} & \mathcal{A}^{*}}
\end{equation*}
yield a Jandl structure as described in \cite{waldorf1}. Hence,
we will call a $(\Z_2,\id)$-equivariant structure just
\emph{Jandl structure} and a $(\Z_{2},\id)$-equivariant bundle gerbe
\emph{Jandl gerbe}.

Now, we elaborate the details of a Jandl structure, which we may
assume to be descended according to Lemma \ref{lem1}.  The
1-isomorphism $\mathcal{A}_k$ consists of a line bundle $A_k$ over
$Z^{k}:= Y \times_M Y_{k}$ of curvature
\begin{equation}
\label{80}
\mathrm{curv}(A_k)=-(\pi_2^{*}C + \pi_1^{*}C)
\end{equation}
in the notation
of Section \ref{sec4}, together with an isomorphism
\begin{equation}
\label{81}
\alpha_k: \pi_{13}^{*}L \otimes \pi_{34}^{*}A_k \to \pi_{12}^{*}A_k
\otimes \pi_{24}^{*}L^{*}
\end{equation}
of line bundles over $Z^{k} \times_M Z^{k}$ satisfying axiom (1M2).
The composition $k^{*}\mathcal{A}_k^{\dagger}
\circ \mathcal{A}_k$ is the 1-isomorphism with the surjective submersion $\id$
on
\begin{equation*}
Z^{k,k}=Z^{k} \times_{Y_k} Z^{k}_k \cong Y \times_M Y_k \times_M Y
\text{,}
\end{equation*}
the
line bundle $\pi_{12}^{*}A_k \otimes \pi_{23}^{*}A_k^{*}$ over $Z^{k,k}$ and
the isomorphism
\begin{multline*}
(\id\otimes \pi_{2356}^{*}\alpha_k^{*-1}) \circ (\pi_{1245}^{*}\alpha_k
\otimes\id): \pi_{14}^{*}L \otimes \pi_{45}^{*}A_k \otimes
\pi_{56}^{*}A_k^{*}\\ \to \pi_{12}^{*}A_k \otimes \pi_{23}^{*}A_k^{*}
\otimes \pi_{36}^{*}L
\end{multline*}
of line bundles over $Z^{k,k} \times_M Z^{k,k}$.
The 2-isomorphism
$\varphi_{k,k}$ corresponds to a bundle isomorphism
\begin{equation}
\label{82} \varphi_{k,k}: \pi_{12}^{*}A_k \otimes \pi_{23}^{*}A_k^{*} \to
\pi_{13}^{*}L\text{,}
\end{equation}
compatible with $\alpha_k$ by virtue of the
commutativity of the diagram
\begin{equation*}
\alxydim{@=1.2cm}{\pi_{14}^{*}L \otimes \pi_{45}^{*}A_k \otimes
\pi_{56}^{*}A_k^{*} \ar[d]_{\id \otimes \pi_{456}^{*}\varphi_{k,k}} \ar[r]
&  \pi_{12}^{*}A_k \otimes \pi_{23}^{*}A_k^{*} \otimes \pi_{36}^{*}L
\ar[d]^{\pi_{123}^{*}\varphi_{k,k} \otimes \id} \\ \pi_{14}^{*}L \otimes
\pi_{46}^{*}L \ar[r] &\pi_{13}^{*}L \otimes \pi_{36}^{*}L\text{.}}
\end{equation*}
Finally, equation (\ref{40}) gives the commutativity of
\begin{equation}
\label{22}
\alxydim{@C=2.4cm@R=1.2cm}{\pi_{12}^{*}A_k
\otimes \pi_{23}^{*}A_k^{*}
\otimes \pi_{34}^{*}A_k
 \ar[d]_{\pi_{123}^{*}\varphi_{k,k}
\otimes \id} \ar[r]^-{\id \otimes \pi_{234}^{*}\varphi^{*-1}_{k,k}}
& \pi_{12}^{*}A_k
\otimes \pi_{24}^{*}L^{*}
 \ar[d]^{\pi_{124}^{*}\rho}
\\ \pi_{13}^{*}L
\otimes \pi_{34}^{*}A_k
\ar[r]_-{\pi_{134}^{*}\lambda}
& \pi_{14}^{*}A_k\text{.}}
\end{equation}

\def\bun{\mathcal{B}\!un}

\label{71}

It is worthwhile to discuss a \emph{trivialized Jandl gerbe}.
This is a Jandl gerbe $(\mathcal{G},\mathcal{J})$
equipped with a trivialization, i.e. a 1-isomorphism
$\mathcal{T}:\mathcal{G} \to \mathcal{I}_{\rho}$. Here, the
2-categorial formalism can be used fruitfully. In \cite{waldorf1},
a functor
\begin{equation}
\label{55}
\bun: \mathfrak{Hom}(\mathcal{I}_{\rho_1},\mathcal{I}_{\rho_2}) \to
\mathfrak{Bun}(M)
\end{equation}
is defined: for every 1-morphism $\mathcal{A}:\mathcal{I}_{\rho_1}
\to \mathcal{I}_{\rho_2}$ between trivial bundle gerbes over $M$, it
provides a vector bundle $\bun(\mathcal{A})$ over $M$; and for every
2-morphism $\beta:\mathcal{A} \Rightarrow \mathcal{A}'$, it provides
a morphism
\begin{equation*}
\bun(\beta): \bun(\mathcal{A})
\to \bun(\mathcal{A}')
\end{equation*}
of vector bundles over $M$. If the 1-morphism
$\mathcal{A}:\mathcal{I}_{\rho_1} \to \mathcal{I}_{\rho_2}$ has a
vector bundle $A$ over $\zeta: Z \to M \times_M M \cong M$, the
vector bundle $\bun(\mathcal{A})$ is uniquely characterized by the
property that $\zeta^{*}\bun(\mathcal{A})\cong A$. Accordingly, the
rank $n$ of $\bun(\mathcal{A})$ is equal to the rank of $A$, and its
curvature satisfies, by axiom (1M1),
\begin{equation*}
\frac{_1}{^n} \mathrm{tr}(\mathrm{curv}(\bun(\mathcal{A}))) = \rho_2
- \rho_1\text{.}
\end{equation*}
The functor $\bun$ has the following
compatibility properties:
\begin{itemize}
\item
$\bun(\mathcal{A}_2 \circ \mathcal{A}_1) = \bun(\mathcal{A}_1)
\otimes \bun(\mathcal{A}_2)$

\item
$\bun(\id_{\mathcal{I}_{\rho}})=\trivlin$
\item
$\bun(f^{*}\mathcal{A}) = f^{*}\bun(\mathcal{A})$
\item
$\bun(\mathcal{A}^{\dagger}) = \bun(\mathcal{A})^{*}$,
\end{itemize}
in which $\trivlin$ denotes the trivial line bundle with the trivial
flat connection.

Let us return to the Jandl gerbe $(\mathcal{G},\mathcal{J})$ and the
trivialization $\mathcal{T}: \mathcal{G} \to \mathcal{I}_{\rho}$.
First, we form a 1-isomorphism $\mathcal{R}:\mathcal{I}_{\rho} \to
\mathcal{I}_{-k^{*}\rho}$ by composing
\begin{equation}
\label{69} \alxydim{@C=1.2cm}{\mathcal{I}_{\rho}
\ar[r]^-{\mathcal{T}^{-1}} & \mathcal{G} \ar[r]^-{\mathcal{A}_k} &
k^{*}\mathcal{G}^{*} \ar[r]^-{k^{*}\mathcal{T}^{\dagger}} &
\mathcal{I}_{-k^{*}\rho}\text{,}}
\end{equation}
and a 2-isomorphism $\psi: k^{*}\mathcal{R}^{\dagger} \circ \mathcal{R} \Rightarrow \id_{\mathcal{I}_{\rho}}$ by composing
\begin{equation}
\label{70} \alxydim{@=1.2cm@C=0.3cm}{k^{*}\mathcal{R}^{\dagger}
\circ \mathcal{R} \ar@{=}[r] & \mathcal{T} \circ
k^{*}\mathcal{A}^{\dagger}_k \circ k^{*}\mathcal{T}^{\dagger-1}
\circ k^{*}\mathcal{T}^{\dagger} \circ \mathcal{A}_k \circ
\mathcal{T}^{-1}  \ar@{=>}[d]^{\id \circ i_l \circ \id} \\ &
\mathcal{T} \circ k^{*}\mathcal{A}_k^{\dagger}
\circ\id_{k^*\mathcal{G}^{*}}\circ \mathcal{A}_k \circ
\mathcal{T}^{-1} \ar@{=>}[d]^{\id \circ \rho_{\mathcal{A}_k} \circ
\id} \\ & \mathcal{T} \circ k^{*}\mathcal{A}^{\dagger}_k
\circ\mathcal{A}_k \circ \mathcal{T}^{-1} \ar@{=>}[d]^{\id \circ
\varphi_{k,k} \circ \id} \\ & \mathcal{T} \circ \id_{\mathcal{G}}
\circ \mathcal{T}^{-1} \ar@{=>}[d]_-{\lambda_{\mathcal{T}} \circ
\id}
\\& \mathcal{T} \circ \mathcal{T}^{-1} \ar@{=>}[r]_-{i_r^{-1}} &
\id_{\mathcal{I}_{\rho}}\text{.}}
\end{equation}
In this definition, we have used the
2-isomorphisms $i_l$ and $i_r$ from \erf{18} associated to the
inverse 1-isomorphism $k^{*}\mathcal{T}^{\dagger-1}$, and the
2-isomorphisms $\rho$ and $\lambda$ from \erf{19}. Equation \erf{16}
assures that it is not important whether one uses
$\lambda$ or $\rho$.

Now we apply the functor $\bun$ to the 1-isomorphism
$\mathcal{R}$ and the 2-isomorphism
$\psi$. The first yields a line bundle $R :=
\bun(\mathcal{R})$ over $M$ of curvature
$-(k^{*}\rho+\rho)$, and the
second (using the above rules) an isomorphism
\begin{equation*}
\phi := \bun(\psi): R \otimes k^{*}R^{*} \to \trivlin
\end{equation*}
of line bundles over $M$. Finally,
condition \erf{40} implies that
\begin{equation*}
\phi \otimes \id_{R} = \id_R \otimes k^{*}\phi^{\dagger}
\end{equation*}
as isomorphisms from $R \otimes k^{*}R^{*} \otimes R$ to $R$. In
other words, the pair $(R,\phi)$ is a $k$-equivariant line bundle
over $M$. Summarizing, every trivialized Jandl gerbe gives rise to
an equivariant line bundle.

\begin{remark}
One could also use the functor $\bun$ to express a trivialized
twisted-equivariant structure in terms of bundles over $M$ in the
case of a general orientifold group $(\Gamma,\varepsilon)$. The
result is not (as probably expected) a $\Gamma$-equivariant line
bundle over $M$ but a \emph{family} $R_{\gamma}$ of line bundles
over $M$ of curvature $\gamma\act\rho-\rho$, together with
isomorphisms
\begin{equation*}
\phi_{\gamma_1,\gamma_2}: R_{\gamma_1} \otimes \gamma_1\act R_{\gamma_2} \to R_{\gamma_1\gamma_2}
\end{equation*}
of line bundles which satisfy a coherence condition on triples $\gamma_1,\gamma_2,\gamma_3\in\Gamma$.
\end{remark}

It will be useful to investigate the relation between the
equivariant line bundles associated to two equivariantly isomorphic
Jandl gerbes. Suppose that $(\mathcal{G}^a,\mathcal{J}^a)$ and
$(\mathcal{G}^b,\mathcal{J}^b)$ are Jandl gerbes over $M$ with
respect to the same involution $k:M \to M$, and
suppose further that $(\mathcal{B},\eta_k):
(\mathcal{G}^a,\mathcal{J}^a) \to (\mathcal{G}^b,\mathcal{J}^b)$ is
an equivariant 1-isomorphism. Let $\mathcal{T}^b: \mathcal{G}^b \to
\mathcal{I}_{\rho}$ be a trivialization of $\mathcal{G}^a$ and let
$\mathcal{T}^a:= \mathcal{T}^b \circ \mathcal{B}$ be the
induced trivialization of $\mathcal{G}^a$. We obtain the
$k$-equivariant line bundles $(R^a,\phi^a)$ and $(R^b,\phi ^b)$ over
$M$ in the manner described above.

\begin{lemma}
\label{lem3}
The 2-isomorphism $\eta_k$ induces an isomorphism $\kappa: R^a \to R^b$ of line bundles over $M$ that respects the equivariant structures in the sense that the diagram
\begin{equation*}
\alxydim{@=1.2cm}{R^a \otimes k^{*}(R^a)^{*} \ar[r]^-{\phi^a}
\ar[d]_{\kappa \otimes k^{*}\kappa^{*-1}} & \trivlin \ar@{=}[d] \\
R^b \otimes k^{*}(R^b)^{*} \ar[r]_-{\phi^b} & \trivlin}
\end{equation*}
of isomorphisms of line bundles over $M$ is commutative.
\end{lemma}

\begin{proof}
The 2-isomorphism $\eta_k: k^{*}\mathcal{B}^{\dagger} \circ
\mathcal{A}_k^a \Rightarrow \mathcal{A}_k^b \circ \mathcal{B}$
induces an isomorphism
\begin{equation*}
\alxydim{@C=1.5cm}{k^{*}\mathcal{B}^{\dagger} \circ \mathcal{A}_k^a
\circ \mathcal{B}^{-1} \ar@{=>}[r]^-{\eta_k \circ \id} &
\mathcal{A}_k^b \circ \mathcal{B} \circ \mathcal{B}^{-1}
\ar@{=>}[r]^-{\id \circ i_r^{-1}} & \mathcal{A}_k^b \circ
\id_{\mathcal{G}^b} \ar@{=>}[r]^-{\lambda_{\mathcal{A}_k^b}} &
\mathcal{A}_k^b\text{.}}
\end{equation*}
The composition of the above
1-morphisms with $(\mathcal{T}^b)^{-1}$ from the right
and with $k^{*}(\mathcal{T}^b)^\dagger$
from the left yields a
2-isomorphism $\eta_k':\mathcal{R}^a\Rightarrow\mathcal{R}^b$
according to \erf{69}. Then, we have $\kappa := \bun(\eta_k')$. It
is straightforward to check that the 2-isomorphism $\eta_k'$ and
the two 2-isomorphisms $\psi^a$ and $\psi^b$ from \erf{70} fit into
a commutative diagram, such that applying the functor $\bun$ yields
the assertion we had to show.
\end{proof}

\section{Descent Theory for Jandl Gerbes}

\label{sec2}

In this section, we consider an orientifold group
$(\Gamma,\epsilon)$ whose normal subgroup
$\Gamma_0 := \mathrm{ker}(\epsilon)$ of $\Gamma$ acts on $M$ without
fixed points, so that the quotient $M':=M/\Gamma_0$ is equipped with
a canonical smooth-manifold structure such that the projection
$p: M \to M'$ is a smooth map. We remark that there is a remaining
smooth group action of $\Gamma':=\Gamma/\Gamma_0$ on $M'$. We also
still have a group homomorphism $\epsilon': \Gamma' \to \Z_2$, so
that $(\Gamma',\epsilon')$ is an orientifold group for $M'$.

\begin{theorem}
\label{th2} Let $(\Gamma,\epsilon)$ be an orientifold group for $M$
with $\Gamma_0$ acting without fixed points, and let
$(\Gamma',\varepsilon')$ be the quotient orientifold group for
the quotient $M':=M/\Gamma_0$. Then, there is a canonical bijection
\begin{equation*}
\alxy{\left \lbrace \txt{Equivalence classes\\of $(\Gamma,\epsilon)$-equivariant\\bundle
gerbes
over $M$} \right \rbrace \ar[r]^-{\cong} & \left \lbrace \txt{Equivalence
classes\\of $(\Gamma',\epsilon')$-equivariant\\bundle gerbes over $M'$} \right \rbrace\text{.}}
\end{equation*}
\end{theorem}
\noindent
Notice that Theorem \ref{th2} unites two interesting cases:
\begin{enumerate}
\item
The original group homomorphism $\epsilon:\Gamma \to \Z_2$ is
constant $\epsilon(\gamma)=1$. In this case,
$\Gamma_0=\Gamma$ and $(\Gamma',\epsilon')$ is the trivial
(orientifold) group. Here, Theorem \ref{th2}
reduces to the well-known bijection
\begin{equation*}
\alxy{\left \lbrace \txt{Equivalence classes\\of $\Gamma$-equivariant\\bundle
gerbes
over $M$} \right \rbrace \ar[r]^-{\cong} & \left \lbrace \txt{Isomorphism
classes of\\bundle gerbes over $M'$} \right \rbrace}\text{.}
\end{equation*}
This bijection was used in \cite{gawedzki2} to
construct bundle gerbes on non-simply connected Lie groups
$G/\Gamma_0$ from bundle gerbes over the universal covering
group $G$.

\item
The original group homomorphism $\epsilon:\Gamma \to \Z_2$ is
non-trivial. In this case, $\Gamma'=\Z_2$ and
$\epsilon'=\id$. Here Theorem \ref{th2} reduces to a
bijection
\begin{equation*}
\alxy{\left \lbrace \txt{Equivalence classes\\of $(\Gamma,\epsilon)$-equivariant\\bundle
gerbes
over $M$} \right \rbrace \ar[r]^-{\cong} & \left \lbrace \txt{Equivalence
classes of\\Jandl gerbes over $M'$} \right \rbrace\text{.}}
\end{equation*}
We will use this bijection in \cite{gawedzki7} to construct Jandl gerbes  over non-simply connected Lie groups.
\end{enumerate}

In the sequel of this section, we shall prove Theorem \ref{th2} assuming,
for simplicity, that all equivariant structures are normalized.
 First, we start with a given $(\Gamma,\epsilon)$-equivariant bundle
gerbe over $M$ and construct an associated quotient bundle gerbe
$\mathcal{G}'$ over $M'$ along the lines of \cite{gawedzki1}.

By Lemma \ref{lem1}, we may assume that the
$(\Gamma,\epsilon)$-equivariant structure is descended. If $\pi:Y
\to M$ is the surjective submersion of  $\mathcal{G}$, the fibre
products of the surjective submersion $\omega:Y \to M'$ with
itself, defined by $\omega:= p \circ \pi$, are disjoint unions
\begin{equation*}
Y \times_{M'} Y\cong \bigsqcup_{\gamma \in \Gamma_0}Z^{\gamma}\quad\text{ and }\quad Y \times_{M'}
Y \times_{M'} Y \cong \bigsqcup_{(\gamma_1,\gamma_2)\in
\Gamma_0^2} Z^{\gamma_1,\gamma_2}\text{.}
\end{equation*}
We recall that the $(\Gamma,\epsilon)$-equivariant structure
on  $\mathcal{G}$ in particular has
a line bundle $A_{\gamma}$ over $Z^{\gamma}$ of curvature
\begin{equation}
\label{9}
\mathrm{curv}(A_{\gamma})
= \epsilon(\gamma)\pi_{2}^{*}C - \pi_1^{*}C
\end{equation}
for each $\gamma\in\Gamma$, and for each pair $(\gamma_1,\gamma_2)\in\Gamma^2$ an isomorphism
\begin{equation*}
\varphi_{{\gamma_1,\gamma_2}}:\pi_{12}^{*}A_{\gamma_1}
\otimes \pi_{23}^{*}A_{\gamma_2}^{\epsilon(\gamma_1)}
\to \pi_{13}^{*}A_{\gamma_1\gamma_2}
\end{equation*}
of line bundles over $Z^{\gamma_1,\gamma_2}$, such that  diagram
(\ref{7}) is commutative.

\begin{definition}
\label{def5}
 The \emph{quotient
bundle gerbe $\mathcal{G}'$} over $M'$
is defined as follows:
\begin{itemize}
\item[(i)]
its surjective submersion $\omega: Y \to M'$ is the composition of
the surjective submersion $\pi:Y \to M$ of $\mathcal{G}$ with the
quotient map $p:M \to M'$;

\item[(ii)]
its 2-form is the 2-form $C\in \Omega^2(Y)$ of $\mathcal{G}$;

\item[(iii)]
its line bundle $A$ over $Y \times_{M'} Y$ is given by the line
bundle $A|_{Z^{\gamma}} := A_\gamma$ over each component
$Z^{\gamma}$ of $Y \times_{M'} Y$;

\item[(iv)]
its isomorphism is given by the isomorphism
$\varphi_{\gamma_1,\gamma_2}$ over each component
$Z^{\gamma_1,\gamma_2}$ of $Y \times_{M'} Y \times_{M'} Y$.

\end{itemize}
The axioms (G1) and (G2) for the quotient bundle
gerbe follow from (\ref{9}) and (\ref{7}), respectively.
\end{definition}

In the case of non-trivial $\varepsilon$, we enhance the quotient
bundle gerbe $\mathcal{G}'$ to a Jandl gerbe. Let us, for
simplicity, denote by $\Gamma_{-} \subset \Gamma$ the subset of
elements $\gamma\in \Gamma$ with $\epsilon(\gamma)=-1$. To define a
Jandl structure $\mathcal{J}'$ on the quotient bundle gerbe
$\mathcal{G}'$, we use the line bundles $A_\gamma$ over $Z^{\gamma}$
for $\gamma\in \Gamma_{-}$ (which have not been used in Definition
\ref{def5}), and the isomorphisms $\varphi_{\gamma_1,\gamma_2}$ for
elements $\gamma_1,\gamma_2\in \Gamma$ with either $\gamma_{1}\in
\Gamma_{-}$ or $\gamma_2\in\Gamma_{-}$ (which have
not been used yet either). The 1-isomorphism
\begin{equation}
\label{47}
\mathcal{A}'_k: \mathcal{G}' \to k^{*}\mathcal{G}^{\prime*}
\end{equation}
is
defined as follows: the fibre product $P:=Y \times_{M'} Y_k$ of the surjective
submersions of the two bundle gerbes can be written as
\begin{equation*}
P \cong \bigsqcup_{\gamma\in \Gamma_{-}} Z^{\gamma}\text{,}
\end{equation*}
and the line bundle $A'_k$ over $P$ is defined as $A'_k|_{Z^\gamma} :=
A_\gamma$. It has the correct curvature $\mathrm{curv}(A) = -\pi_2^{*}C - \pi_1^{*}C$
in the sense of axiom (1M1). The two-fold fibre product has
the components
\begin{equation*}
P \times_{M'} P \cong \bigsqcup_{\gamma_1,\gamma_2,\gamma_3\in
\Gamma_{-}} Z^{\gamma_1,\gamma_2,\gamma_3}\text{.}
\end{equation*}
Now, we have to define an isomorphism
$\alpha$ of line bundles over $P \times_{M'} P$, which is an
isomorphism
\begin{equation*}
\alpha|_{Z^{\gamma_1,\gamma_2,\gamma_3}}: \pi_{13}^{*}A_{\gamma_1\gamma_2} \otimes \pi_{34}^{*}A_{\gamma_3} \to \pi_{12}^{*}A_{\gamma_1}
\otimes \pi_{24}^{*}A^{*}_{\gamma_2\gamma_3}\text{,}
\end{equation*}
on the component $Z^{\gamma_1,\gamma_2,\gamma_3}$, where we
have a dual line bundle because the target of the isomorphism
$\mathcal{A}$ is the dual bundle gerbe. We define this isomorphism
as the composition of
\begin{equation*}
\pi_{134}^{*}\varphi_{\gamma_1\gamma_2,\gamma_3}:\pi_{13}^{*}A_{\gamma_1\gamma_2} \otimes \pi_{34}^{*}A_{\gamma_3}
\to \pi_{14}^{*}A_{\gamma_1\gamma_2\gamma_3}\text{,}
\end{equation*}
with no dual line bundle on the left
since $\epsilon(\gamma_1\gamma_2)=1$, with the inverse of
\begin{equation*}
\pi_{124}^{*}\varphi_{\gamma_1,\gamma_2\gamma_3}:\pi_{12}^{*}A_{\gamma_1} \otimes \pi_{24}^{*}A_{\gamma_1\gamma_3}^{*}
\to \pi_{14}^{*}A_{\gamma_1\gamma_2\gamma_3}\text{.}
\end{equation*}
The isomorphism $\alpha$ defined in this manner
satisfies axiom (1M2) for 1-morphisms
due to the commutativity condition (\ref{7}) for the isomorphisms
$\varphi_{\gamma_1,\gamma_2}$. This completes the definition of the
1-isomorphism $\mathcal{A}'_k$.

We are left with the definition of the 2-isomorphism
$\varphi'_{k,k}: k^{*}\mathcal{A}_k^{\prime\dagger} \circ
\mathcal{A}'_k \Rightarrow \id_{\mathcal{G}'}$ for which we use the
remaining structure, namely the 1-isomorphisms
$\varphi_{\gamma_1,\gamma_2}$ with
$\gamma_1,\gamma_2\in \Gamma_{-}$. The 1-morphism
$k^{*}\mathcal{A}_k^{\prime\dagger} \circ \mathcal{A}'_k$
has a surjective submersion $\omega:W\rightarrow
Y\times_{M'}Y$ for $W=Y\times_{M'}Y_k\times_{M'}Y$. Upon the identification
\begin{equation*}
W \cong \bigsqcup_{\gamma_1,\gamma_2
\in \Gamma_{-}} Z^{\gamma_1,\gamma_2}\text{,}
\end{equation*}
$\omega$ is induced by the natural maps $Z^{\gamma_1,\gamma_2}\rightarrow
Z^{\gamma_1\gamma_2}$. Over the component $Z^{\gamma_1,\gamma_2}$,
the line bundle of $k^{*}\mathcal{A}_k^{\prime\dagger} \circ \mathcal{A}'_k$
is equal to $\pi_{12}^{*}A_{\gamma_1} \otimes \pi_{23}^{*}A_{\gamma_2}^{*}$
and we define the bundle isomorphism $\varphi'_{k,k}|_{Z^{\gamma_1,\gamma_2}}$
as
\begin{equation*}
\varphi_{\gamma_1,\gamma_2}: \pi_{12}^{*}A_{\gamma_1} \otimes
\pi_{23}^{*}A_{\gamma_2}^{*} \to \pi_{13}^{*}A_{\gamma_1\gamma_2}.
\end{equation*}
Indeed, the 1-isomorphism
$\id_{\mathcal{G}}$ has the line bundle of the bundle gerbe
$\mathcal{G}'$ which is $A_{\gamma_1\gamma_2}$ over
$Z^{\gamma_1\gamma_2}$. The axiom for
$\varphi'_{k,k}$ can be deduced from the commutativity condition for
the 2-isomorphisms $\varphi_{\gamma_1,\gamma_2}$.

Finally, we have to assure that the 2-isomorphism $\varphi'_{k,k}$
satisfies equation (\ref{40}) for Jandl structures.
To see this, we have to express the natural 2-isomorphisms
$\rho_{\mathcal{A}}$ and $\lambda_{\mathcal{A}}$ by the given
2-isomorphisms $\varphi_{\gamma_1,\gamma_2}$. According to their
definition,
 we find $\rho_{\mathcal{A}}|_{Z^{\gamma_1,\gamma_2}} = \varphi
_{\gamma_1,\gamma_2}$ and
$\lambda_{\mathcal{A}}|_{Z^{\gamma_1,\gamma_2}} =
\varphi_{\gamma_2,\gamma_1}$ for $\gamma_1\in \Gamma_{-}$ and
$\gamma_2\in \Gamma_0$. Then, equation (\ref{40}) reduces to the
commutativity condition for the 2-isomorphisms
$\varphi_{\gamma_1,\gamma_2}$. This completes
the definition of the Jandl structure $\mathcal{J}'$ on the quotient
bundle gerbe $\mathcal{G}'$.

The second step in the proof of Theorem \ref{th2} is to
demonstrate that the procedure
described above is well-defined on equivalence classes. For this
purpose, we show that an equivariant 1-morphism
\begin{equation*}
(\mathcal{B},\eta_{\gamma}): (\mathcal{G}^a,\mathcal{J}^a) \to (\mathcal{G} ^b,\mathcal{J}^b)
\end{equation*}
between $(\Gamma,\epsilon)$-equivariant bundle gerbes over $M$
induces a 1-morphism $\mathcal{B}'$ between the quotient bundle
gerbes $\mathcal{G}^{a\,\prime}$ and $\mathcal{G}^{b\,\prime}$. We
may assume again that the 1-morphism $\mathcal{B}$ is descended.
Then, it consists of a line bundle $B$ over  $Z:=Y^a \times_M Y^b$,
and of an isomorphism $\beta$ of line bundles over $Z \times_M Z$.
The additional 2-isomorphisms correspond to bundle isomorphisms
\begin{equation*}
\eta_\gamma: \pi_{12}^{*}A^a_{\gamma} \otimes \pi_{24}^{*}B^{\epsilon(\gamma)} \to \pi_{13}^{*}B \otimes \pi_{34}^{*}A^b_{\gamma}\text{.}
\end{equation*}
of line bundles over $Y^a \times_M Y^a_{\gamma} \times_M Y^b
\times_M Y^b_{\gamma}$, see (\ref{34}).

The  quotient 1-morphism $\mathcal{B}'$ is defined as follows. Its
surjective submersion is the disjoint union $\tilde Z$ of $\tilde
Z^{\gamma} := Z \times_{Y^b} (Z^b)^{\gamma}\cong Y^a\times_MY^b
\times_MY^b_\gamma$ over all
$\gamma\in \Gamma_0$, together with the projection
\begin{equation*}
\pi_{13}: \tilde Z^{\gamma} \to Y^a \times_M Y^b_{\gamma}
\end{equation*}
whose codomain is the $\gamma$-component of the fibre product of the surjective submersions of the two quotient bundle
gerbes. Its line bundle $B'$ is defined as
$B'|_{\tilde Z^\gamma}\equiv B'_\gamma := \pi_{12}^{*}B \otimes \pi_{23}^{*}A^b_{\gamma}$, which
has the correct curvature:
\begin{multline*}
\mathrm{curv}(B'_\gamma) = \pi_{12}^*\mathrm{curv}(B) + \pi_{23}^*
\mathrm{curv}(A^b_{\gamma})\\ = \pi_2^{*}C^b
- \pi_1^{*}C^a + \pi_3^{*}C^b - \pi_2^{*}C^b
= \pi_3^{*}C^b - \pi_1^{*}C^a\text{.}
\end{multline*}
In order to define the isomorphism
of $\mathcal{B}'$, we have to consider the fibre product
\begin{equation*}
\tilde Z \times_{M'} \tilde Z \cong
\bigsqcup_{\gamma,\gamma',\gamma'' \in \Gamma_0} Y^a \times_M Y^b
\times_M Y^b_{\gamma} \times_M Y^a_{\gamma\gamma'} \times_M
Y^b_{\gamma\gamma'} \times_M Y^b_{\gamma\gamma'\gamma''}
\end{equation*}
and set
\begin{equation*}
\alxydim{@R=1.2cm@C=0.6cm}{\pi_{14}^{*}A^a_{\gamma\gamma'} \otimes \pi_{456}^{*}
B_{\gamma''} \ar@{=}[r]
& \pi_{14}^{*}A^a_{\gamma\gamma'}
\otimes \pi_{45}^{*}B \otimes \pi_{56}^{*}A^b_{\gamma''}
\ar[d]^{\pi_{1245}^{*}\eta_{\gamma\gamma'} \otimes  \id} \\ &
\pi_{12}^{*}B \otimes \pi_{25}^{*}A^b_{\gamma\gamma'} \otimes
\pi_{56}^{*}A^b_{\gamma''} \ar[d]^{\id \otimes
\pi_{256}^{*}\varphi^b_{\gamma\gamma',\gamma''}} \\ & \pi_{12}^{*}B
\otimes \pi_{26}^{*}A^b_{\gamma\gamma'\gamma''} \ar[d]^{\id \otimes
\pi_{236}^{*}\varphi_{\gamma,\gamma'\gamma''}^{b-1}}
\\ & \pi_{12}^{*}B \otimes \pi_{23}^{*}A^b_\gamma \otimes \pi_{36}^{*}A^b_{\gamma'\gamma''} \ar@{=}[r]
& \pi_{123}^{*} B'_\gamma \otimes \pi_{36}^{*}A^b_{\gamma'\gamma''}\text{.}}
\end{equation*}
This isomorphism satisfies axiom
(1M2) due to the
commutativity of the diagram for the isomorphisms $\eta_{\gamma}$
from Definition \ref{def4} and the one for the
$\varphi^b_{\gamma_1,\gamma_2}$ from Definition \ref{def1}.

In the case of $\varepsilon$
non-trivial, we enhance the quotient 1-morphism $\mathcal{B}'$ to
an equivariant 1-morphism
\begin{equation*}
(\mathcal{B}',\eta_{k}'): (\mathcal{G}^{a\prime},\mathcal{J}^{a\prime}) \to (\mathcal{G}^{b\,\prime},\mathcal{J}^{b\,\prime})
\end{equation*}
between Jandl gerbes over $M'$. To this end, we have
to define the 2-isomorphism
\begin{equation}
\label{48}
\eta'_k: k^{*}\mathcal{B}'^{\dagger} \circ \mathcal{A}'^a_k \Rightarrow \mathcal{A}'^b_k \circ \mathcal{B}'
\end{equation}
for $k$ the non-trivial group element of $\Gamma'\cong\Z_2$, and
$\mathcal{A}^{\prime a}_k$ and $\mathcal{A}_k^{\prime b}$ the
1-isomorphisms \erf{47} of the quotient Jandl structures on
$\mathcal{G}^{a\prime}$ and $\mathcal{G}^{b\prime}$, respectively.
Collecting all definitions, we establish that the 1-morphism
$k^{*}\mathcal{B}^{\prime\dagger} \circ \mathcal{A}^{\prime a}_k$ on
the left has the submersion $Z^{l} := P^a \times_{Y^a_k} \tilde Z_k$
which, for $\tilde Z\cong\bigsqcup_{\gamma\in\Gamma_0}Y^a
\times_M Y^b\times_M Y^b_{\gamma}$, becomes a disjoint union over the
fibre products $Y^a \times_M Y^a_{\gamma'} \times_{M} Y_{\gamma'}^b
\times_M Y^{b}_{\gamma''}$ for $\gamma',\gamma''\in\Gamma_{-}$. Over
this space, it has a line bundle $A^{l}$ defined componentwise as
$\pi_{12}^{*}A_{\gamma'}^a \otimes \pi_{23}^{*}B^{*} \otimes
\pi_{34}^{*}(A^b_{\gamma'^{-1}\gamma''})^{*}$. The 1-morphism
$\mathcal{A}_k^{\prime b} \circ \mathcal{B}'$ on the right has the
submersion $Z^r := \tilde Z \times_{Y^{b}} P^b$ which is the
disjoint union of the fibre products $Y^a \times_M Y^b \times_M
Y^b_{\gamma} \times Y^b_{\gamma''}$, with $\gamma\in\Gamma_0$ and
$\gamma''\in\Gamma_-$, and over that, a line bundle $A^r$ defined
componentwise as $\pi_{12}^{*}B \otimes \pi_{23}^{*}A^b_{\gamma}
\otimes \pi_{34}^{*}A_{\gamma^{-1}\gamma''}^b$. The 2-morphism
\erf{48} is now defined on the surjective submersion
\begin{equation*}
W := \bigsqcup_{\gamma\in\Gamma_{0},\gamma',\gamma''\in \Gamma_-}
Y^a \times_M Y_{\gamma'}^a \times_M Y^b \times_M Y^b_{\gamma'} \times_M
Y^{b}_{\gamma}\times_M Y^b_{\gamma''}
\end{equation*}
with the  projections $\pi_{1246}$ to $Z^l$ and $\pi_{1356}$ to
$Z^r$. We then declare the following isomorphism between the
pullbacks of $A^l$ and $A^r$ to $W$:
\begin{equation*}
\alxydim{@R=1.2cm}{\pi_{1246}^*A^l  \ar@{=}[r] & 
\pi_{12}^{*}A_{\gamma'}^a \otimes \pi_{24}^{*}B^{*}
\otimes \pi_{46}^{*}(A^b_{\gamma'^{-1}\gamma''})^{*}
 \ar[d]^{\pi_{1234}^{*}\eta_{\gamma'} \otimes \id}
& \\ &
\pi_{13}^{*}B\otimes\pi_{34}^*A_{\gamma'}^b \otimes
\pi_{46}^{*}(A^b_{\gamma'^{-1}\gamma''})^{*}
\ar[d]^{\id\otimes\pi_{346}^*\varphi_{\gamma',\gamma'^{-1}\gamma''}^b}
& \\ &
\pi_{13}^{*}B\otimes\pi_{36}^*A^b_{\gamma''}
\ar[d]^{\id\otimes\pi_{356}^*(\varphi_{\gamma,\gamma^{-1}\gamma''}^b){}^{-1}}
& \\ &
\pi_{13}^{*}B\otimes\pi_{35}^*A^b_{\gamma}\otimes
\pi_{56}^*A^b_{\gamma^{-1}\gamma''} \ar@{=}[r] &
\pi^*_{1356}A^r\textrm{.}}
\end{equation*}
It involves the isomorphism $\eta_{\gamma'}$ belonging to the
equivariant 1-morphism, see \erf{34}. Since $\gamma' \in \Gamma_{-}$
here, we have, by now, used all the structure of
$(\mathcal{B},\eta_{\gamma})$. It is straightforward to check that
these isomorphisms satisfy the compatibility axiom and make the
diagram \erf{17} commutative.

We have, so far, obtained a well-defined map
\begin{equation*}
q\;\;: \alxy{\left \lbrace \txt{Equivalence classes\\of
$(\Gamma,\epsilon)$-equivariant\\bundle gerbes over $M$} \right
\rbrace \ar[r] & \left \lbrace \txt{Equivalence classes\\of
$(\Gamma',\varepsilon')$-equivariant\\bundle gerbes over $M'$}
\right \rbrace\text{.}}
\end{equation*}
In order to finish the proof of Theorem
\ref{th2}, we now show that $q$ is surjective and injective. Let us
start with the observation that the pullback
$\mathcal{G}:=p^{*}\mathcal{G}'$ of any bundle gerbe $\mathcal{G}'$
over $M'$ along the projection $p:M \to M'$ has a canonical
$\Gamma_0$-equivariant structure. This comes from the fact that
$\mathcal{G} =\gamma\act\mathcal{G}$ for all $\gamma\in \Gamma_0$,
so that $\mathcal{A}_{\gamma} := \id_{\mathcal{G}}$ and
\begin{equation*}
\varphi_{\gamma_1,\gamma_2}:=\lambda_{\id_{\mathcal{G}}} = \rho_{\id_{\mathcal{G}}}: \id_{\mathcal{G}} \circ \id_{\mathcal{G}} \Rightarrow \id_{\mathcal{G}}
\end{equation*}
define a $\Gamma_0$-equivariant structure.

Let us have a closer look at this $\Gamma_0$-equivariant
structure. To this end, notice that $\mathcal{G}$ has the
surjective submersion $Y:= M \times_{M'} Y' \to M$ whose fibre
products have, each, an obvious projection $Y^{[k]}\to
Y'^{[k]}$. The line bundle $L$ and the isomorphism $\mu$ of
$\mathcal{G}$ are pullbacks of $L'$ and $\mu'$ along this
projection. The 1-isomorphisms $A_{\gamma}$ have the surjective
submersion $Z^{\gamma} \cong Y^{[2]}$ and the line bundles
$A_{\gamma} := L$, see Example \ref{ex3}. The 2-isomorphisms
$\varphi_{\gamma_1,\gamma_2}$ have the surjective submersions
$Z^{\gamma_1,\gamma_2} \cong Y^{[3]}$ and the isomorphisms
$\varphi_{\gamma_1,\gamma_2} := \mu$.

We can, next, pass to the quotient bundle gerbe $\mathcal{G}''$. Its
surjective submersion is $Y'':=Y$ whose fibre products $Y''^{[k]} =
\Gamma_0^{k-1}\times Y^{[k]}$ come, each, with an obvious
projection to $Y^{[k]}$. In fact, 
the line bundle $L''$ and the multiplication $\mu''$ for
$\mathcal{G}''$ are pullbacks of $L$ and $\mu$, respectively, along
these projections. Summarizing, the quotient bundle gerbe
$\mathcal{G}''$ has the structure of $\mathcal{G}'$ pulled back
along the composed projections $Y''^{[k]}\to Y'^{[k]}$, which are
fibre-preserving. It is well-known that such bundle gerbes are
isomorphic.

Let, now, $\mathcal{J}'$ be a Jandl
structure on $\mathcal{G}'$ consisting of a line bundle $A'_k$ over
$Z' := Y' \times_{M'} Y'_{k}$, and of an isomorphism $\alpha$
of line bundles over $Z'^{[2]}$. To the above canonical
$\Gamma_0$-equivariant structure, we add
1-isomorphisms $\mathcal{A}_{\gamma}: \mathcal{G} \to
\gamma\act\mathcal{G}$ for $\gamma \in \Gamma_{-}$. The relevant
fibre product $Z^{\gamma} = Y \times_M Y_{\gamma}$ projects to $Z'$,
so that the line bundle $A_{\gamma}$ is defined as the pullback of
$A'_k$ along this projection. Similarly, the isomorphism $\alpha'$
pulls back to the isomorphism of $\mathcal{A}_{\gamma}$. The Jandl
structure $\mathcal{J}'$ also contains a 2-isomorphism
$\varphi'_{k,k}$, which naturally pulls back to the 2-isomorphisms
$\varphi_{\gamma_1,\gamma_2}$ that we need to complete the
definition of a canonical $(\Gamma,\varepsilon)$-equivariant
structure $\mathcal{J}$ on the pullback bundle gerbe $\mathcal{G}$.
We conclude, along the lines of the above discussion, that 
$(\mathcal{G},\mathcal{J})$ descends to a Jandl
gerbe over $M'$ which is equivariantly isomorphic to
$(\mathcal{G}',\mathcal{J}')$.

It remains to prove that the map $q$ is injective. Thus, we assume 
that two $(\Gamma,\varepsilon)$-equivariant
bundle gerbes $(\mathcal{G}^a,\mathcal{J}^a)$ and
$(\mathcal{G}^b,\mathcal{J}^b)$ descend to a pair of
equivariantly isomorphic Jandl gerbes
$(\mathcal{G}^{a\prime},\mathcal{J}^{a\prime})$ and
$(\mathcal{G}^{b\prime},\mathcal{J}^{b\prime})$ over $M'$. Let
$(\mathcal{B}',\eta'_k)$ be an
equivariant 1-isomorphism between the latter Jandl gerbes. 
If we assume it to be descended,
$\mathcal{B}'$ is based on the fibre product
\begin{equation*}
Z' = Y^{a\prime} \times_{M'} Y^{b\prime} \cong \bigsqcup_{\gamma\in \Gamma_0} Y^a \times_M Y^b_{\gamma}
\end{equation*}
and the line bundle $B'$ over $Z'$ has
components $B'_{\gamma}$. Its isomorphism is
\begin{equation}
\label{60}
\beta'_{\gamma_1,\gamma_2,\gamma_3}:
\pi_{13}^{*}A^a_{\gamma_1\gamma_2} \otimes \pi_{34}^{*}B'_{\gamma_3}
\to \pi_{12}^{*}B'_{\gamma_1} \otimes
\pi_{24}^{*}A^b_{\gamma_2\gamma_3}\text{.}
\end{equation}
The additional 2-isomorphism $\eta'_{k}$ is, in the notation of
Section \ref{sec4}, an isomorphism of line bundles over
\begin{equation*}
Z_1^{\prime k} \times_{P'} Z_2^{\prime k} \cong
\bigsqcup_{\gamma_1,\gamma_2,\gamma_3\in\Gamma_{-}} Y^a
\times_M Y_{\gamma_1}^a \times_{M} Y^b_{\gamma_1\gamma_2}
\times_M Y^b_{\gamma_1\gamma_2\gamma_3}\text{,}
\end{equation*}
with components
\begin{equation}
\label{61} (\eta'_k)_{\gamma_1,\gamma_2,\gamma_3}:
\pi_{12}^{*}A^a_{\gamma_1} \otimes
\pi_{24}^{*}B'^{*}_{\gamma_2\gamma_3} \to \pi_{13}^{*}
B'_{\gamma_1\gamma_2} \otimes \pi_{34}^{*}A^b_{\gamma_3}\,,
\end{equation}
according to \erf{34}. Let us, now, construct an equivariant
1-isomorphism
\begin{equation*}
(\mathcal{B},\eta_{\gamma}): (\mathcal{G}^a,\mathcal{J}^a)
\to (\mathcal{G}^b,\mathcal{J}^b)
\end{equation*}
out of this structure. Its existence will prove
that the map $q$ is injective.

\label{67}

The 1-isomorphism $\mathcal{B}: \mathcal{G}^a \to \mathcal{G}^b$ has
the surjective submersion $Y^a \times_M Y^b$ and we take the line
bundle $B'_1$ over that space as its line bundle $B$ (all the other
line bundles $B'_{\gamma}$ with $\gamma\neq 1$
are to be ignored). Similarly, the isomorphism $\beta'_{1,1,1}$
serves as the isomorphism $\beta$ of $\mathcal{B}$. It remains to
construct the 2-isomorphisms $\eta_{\gamma}$. For
$\gamma\in\Gamma_{-}$, these we define them as the
isomorphisms
\begin{equation*}
(\eta'_{k})_{\gamma,\gamma^{-1},\gamma}: \pi_{12}^{*}A^a_{\gamma}
\otimes \pi_{24}^{*}B'^{*}_{1} \to \pi_{13}^{*}B'_{1} \otimes
\pi_{34}^{*}A^b_{\gamma} \end{equation*} from \erf{61}, while for
$\gamma\in \Gamma_0$, as the
isomorphisms
\begin{equation*}
\beta'_{1,\gamma,1}: \pi_{13}^{*}A^a_{\gamma} \otimes \pi_{34}^{*}B'_{1}
\to \pi_{12}^{*}B'_{1} \otimes \pi_{24}^{*}A^b_{\gamma}
\end{equation*}
from \erf{60} (pulled back along the map that exchanges the second
and the third factor). Finally, all the relations that these
isomorphisms should obey are readily seen to be satisfied.

\section{Twisted-equivariant Deligne Cohomology}

\label{sec3}

In this section, we relate the geometric theory developed in
Sections \ref{sec1} and \ref{sec2} to its cohomological counterpart
introduced in \cite{gawedzki6}.

\subsection{Local Data}

\label{sec3_1}

Let $(\Gamma,\epsilon)$ be an orientifold
group for $M$, and let $\mathcal{G}$
be a bundle gerbe over $M$ with
$(\Gamma,\epsilon)$-equivariant  structure $\mathcal{J}$,
 consisting
of 1-isomorphisms $\mathcal{A}_{\gamma}$ and of 2-isomorphisms
$\varphi_{\gamma_1,\gamma_2}$. We assume that there is a covering
$\mathfrak{O}=\lbrace O_i \rbrace_{i\in I}$ and a left action of
$\Gamma$ on the index set $I$ such that $\gamma(O_i)=O_{\gamma i}$,
and such that there exist sections $s_i: O_i \to Y$. We define
\begin{equation*}
M_{\mathfrak{O}}:=\bigsqcup_{i\in
I} O_i
\end{equation*}
and construct the smooth map
\begin{equation*}
s:M_{\mathfrak{O}}
\to Y:(x,i)\mapsto s_i(x)\text{.}
\end{equation*}
There are induced maps $s: M_{\mathfrak{O}}^{[k]} \to Y^{[k]}$ on
all fibre products, where $M_{\mathfrak{O}}^{[k]}$ is just the
disjoint union of all non-empty $k$-fold intersections
$O_{i_1...i_k}:= O_{i_1} \cap ... \cap O_{i_k}$ of open sets in
$\mathfrak{O}$. They may be used to pull back the line bundle $L$,
the 2-form $C$ and the isomorphism $\mu$ of $\mathcal{G}$. Choose,
now, sections $\sigma_{ij}:O_{ij} \to s^{*}L$ (of unit length) and
define smooth functions $g_{ijk}: O_{ijk}\to U(1)$ by
\begin{equation*}
s^{*}\mu(\sigma_{ij} \otimes \sigma_{jk})
= g_{ijk} \cdot \sigma_{ik}\text{,}
\end{equation*}
extract local connection 1-forms $A_{ij}\in \Omega^1(O_{ij})$ such that
\begin{equation*}
s^*\nabla\sigma_{ij}=\tfrac{1}{\mathrm{i}}A_{ij}\,\sigma_{ij}\,,
\end{equation*}
where $\nabla$ stands for the covariant derivative,
and define 2-forms $B_i := s_i^{*}C \in
\Omega^2(O_i)$. Axiom (G1) gives the equation
\begin{equation}
\label{151}
\mathrm{d}A_{ij}=B_j - B_i\quad\text{ on }O_{ij}\text{.}
\end{equation}
Since $\mu$
preserves connections, one obtains
\begin{equation}
\label{152}
A_{ij}-A_{ik}+A_{jk}=\mathrm{i}g_{ijk}^{-1}\mathrm{d}g_{ijk}\quad\text{ on }O_{ijk}\text{,}
\end{equation}
and axiom (G2) infers the cocycle condition
\begin{equation}
\label{153} g_{ijl}\cdot g_{jkl}=g_{ikl} \cdot g_{ijk}\quad\text{ on
}O_{ijkl}\text{.}
\end{equation}
One can always choose the sections
$\sigma_{ij}$ such that $A_{ij}$ and
$g_{ijk}$ have the antisymmetry property
$A_{ij}=-A_{ji}$ and
$g_{ijk}=g_{jik}^{-1}=g_{ikj}^{-1}=g_{kji}^{-1}$.

We continue by extracting local data of the 1-isomorphisms
\begin{equation*}
\mathcal{A}_{\gamma}: \mathcal{G} \to \gamma\act\mathcal{G}
\end{equation*}
consisting, each, of a line bundle $A_{\gamma}$ over $Z^{\gamma}=Y
\times_M Y_{\gamma}$ and of an isomorphism
\begin{equation*}
\alpha_{\gamma}: \pi_{13}^{*}L
\otimes \pi_{34}^{*}A_{\gamma} \to \pi_{12}^{*}A_{\gamma}
\otimes \pi_{24}^{*}L^{\epsilon(\gamma)}
\end{equation*}
of line bundles over $Z^{\gamma} \times_M Z^{\gamma}$, see the
discussion in Section \ref{sec4}. Note that $s_{\gamma^{-1} i} \circ
\gamma^{-1}: O_i \to Y_{\gamma}$  is a section into $Y_{\gamma}$,
\,i.e. $\pi_{\gamma} \circ (s_{\gamma^{-1} i} \circ
\gamma^{-1})=\id_{O_{i}}$, and so we get the map
$s_{\gamma}:M_{\mathfrak{O}} \to Y_{\gamma}:(x,i) \mapsto
s_{\gamma^{-1} i}(\gamma^{-1}(x))$ compatible with the projections
to $M$.
 Note that also  $\sigma_{\gamma^{-1} i\,\gamma ^{-1}j} \circ
 \gamma^{-1}:O_{ij}\to s_{\gamma}^{*}L$ is a
 section. Furthermore, we have a mixed map
\begin{equation}
z_{\gamma}:M_{\mathfrak{O}}\to
Z^{\gamma}:(x,i) \mapsto (s_i(x),s_{\gamma^{-1} i}(\gamma^{-1}(x)))
\label{10}
\end{equation}
into the space $Z^{\gamma}$.
Note that $\pi_{1} \circ z_{\gamma}=s$
and $\pi_{2} \circ z_{\gamma}=s_{\gamma}$, so
that the pullback of $\alpha_{\gamma}$
along $z_{\gamma}$ is an isomorphism
\begin{equation*}
z_{\gamma}^{*}\alpha_{\gamma}:s^{*}L
\otimes \pi_{2}^{*}z_{\gamma}^{*}A_{\gamma} \to \pi_{1}^{*}z_{\gamma}^{*}A_{\gamma}
\otimes s_{\gamma}^{*}L^{\epsilon(\gamma)}
\end{equation*}
of line bundles over  $M_{\mathfrak{O}} \times_M M_{\mathfrak{O}}$
(the two maps $\pi_1,\pi_2$ above are the canonical projections
from $M_{\mathfrak{O}} \times_M M_{\mathfrak{O}}$ to
$M_{\mathfrak{O}}$). Choose new unit-length sections
$\sigma^{\gamma}_i:O_i \to
z_\gamma^{*}A_{\gamma}$ and obtain local connection 1-forms
$\Pi_i^{\gamma}\in \Omega^1(O_i)$, as well as smooth functions
$\chi^{\gamma}_{ij}:O_{ij} \to U(1)$ by the relation
\begin{equation*}
z_{\gamma}^{*}\alpha_{\gamma}(\sigma_{ij} \otimes \sigma^{\gamma}_j)=\chi^{\gamma}_{ij}\cdot
 (\sigma^{\gamma}_i
 \otimes (  \sigma^{\epsilon(\gamma)}_{\gamma ^{-1}i\gamma^{-1}j} \circ
 \gamma^{-1}) )\text{.}
\end{equation*}
In order to simplify the notation in the following discussion, we write
\begin{equation}
\label{54}
\gamma f_i :=( (\gamma^{-1})^{*}f_{\gamma^{-1} i})^{\epsilon(\gamma)}
\quad\text{ and }\quad
\gamma\Pi_i := \epsilon(\gamma)(\gamma^{-1})^{*}\Pi_{\gamma^{-1}i}
\end{equation}
for $U(1)$-valued functions $f_i$ and 1-forms $\Pi_i$,
respectively, and likewise for components of generic $p$-form-valued
\v Cech cochains encountered below.
In this notation, axiom
(1M1) gives
\begin{equation}
\label{154}
\gamma B_i -B_i = \mathrm{d}\Pi_{i}^{\gamma}\text{.}
\end{equation}
Since $\alpha_{\gamma}$ preserves
connections, one obtains
\begin{equation}
\label{11} \gamma A_{ij} -A_{ij} = \Pi_j^{\gamma} - \Pi_i^{\gamma}
-\mathrm{i}\chi_{ij}^{\gamma-1}\mathrm{d}\chi_{ij}^{\gamma}\textrm{,}
\end{equation}
and axiom (1M2) gives
\begin{equation}
\gamma g_{ijk} \cdot g_{ijk}^{-1}
=\chi_{ij}^{\gamma-1} \cdot \chi_{ik}^{\gamma} \cdot \chi_{jk}^{\gamma-1}\text{.}
\label{12}
\end{equation}
One can, again, choose the sections such that
$\chi^{\gamma}_{ij}=(\chi_{ji}^{\gamma})^{-1}$.

Finally, we extract local data from the
2-isomorphisms $\varphi_{\gamma_1,\gamma_2}$. Consider the space
$Z^{\gamma_1,\gamma_2}=Y \times _M Y_{\gamma_1} \times_M
Y_{\gamma_1\gamma_2}$ and the isomorphism
\begin{equation*}
\varphi_{\gamma_1,\gamma_2}:\pi_{12}^{*}A_{\gamma_1}
\otimes \pi_{23}^{*}A_{\gamma_2}^{\epsilon(\gamma_1)}
\to \pi_{13}^{*}A_{\gamma_1\gamma_2}
\end{equation*}
of line bundles over $Z^{\gamma_1,\gamma_2}$.
We use the map
\begin{eqnarray*}
z_{\gamma_1,\gamma_2}:M_{\mathfrak{O}}
&\to& Z^{\gamma_1,\gamma_2}\\(x,i)&\mapsto&
(s_{i}(x),s_{\gamma_1
^{-1}i}(\gamma_1^{-1}(x)),s_{\gamma_2^{-1}\gamma_1^{-1}i}(\gamma_2^{-1}(\gamma_1^{-1}(x))))
\end{eqnarray*}
to pull back the isomorphism $\varphi_{\gamma_1,\gamma_2}$ to
$M_{\mathfrak{O}} $. Note that $\pi_{12}\circ
z_{\gamma_1,\gamma_2}=z_{\gamma_1}$, $\pi_{23}\circ
z_{\gamma_1,\gamma_2}=z_{\gamma_2} \circ \gamma_1^{-1}$ and
$\pi_{13} \circ z_{\gamma_1,\gamma_2}=z_{\gamma_1\gamma_2}$. Hence,
we may use the sections $\sigma_{i}^{\gamma}:O_i \to
z_{\gamma}^{*}A_\gamma$ to extract smooth
functions $f^{\gamma_1,\gamma_2}_i:O_i \to U(1)$ by the relation
\begin{equation*}
z_{\gamma_1,\gamma_2}^{*}\varphi_{\gamma_1,\gamma_2}(\sigma_i^{\gamma_1}
\otimes (\sigma^{\gamma_2}_{\gamma_1^{-1}i}\circ
\gamma_1^{-1})^{\epsilon(\gamma_1)} )=f^{\gamma_1,\gamma_2}_i\cdot
\sigma_{i}^{\gamma_1\gamma_2}\text{.}
\end{equation*}
From the requirement
that $\varphi_{\gamma_1,\gamma_2}$ respect connections, it
follows that
\begin{equation}
\gamma_1\Pi^{\gamma_2}_i
-\Pi^{\gamma_1\gamma_2}_i+\Pi^{\gamma_1}_i=\mathrm{i}(f^{\gamma_1,\gamma_2}_i)^{-1}\mathrm{d}
f_i^{\gamma_1,\gamma_2} \label{14}
\end{equation}
The compatibility condition (\ref{8}) for the 2-morphism $\varphi_{{\gamma_1,\gamma_2}}$ becomes
\begin{equation}
\gamma_1\chi^{\gamma_2}_{ij}\cdot(\chi^{\gamma_1\gamma_2}_{ij})^{-1}\cdot
\chi^{\gamma_1}_{ij}=(f_i^{\gamma_1,\gamma_2})^{-1}\cdot
f^{\gamma_1,\gamma_2}_j\textrm{,} \label{15}
\end{equation}
and the condition imposed on the morphisms
$\varphi_{\gamma_1,\gamma_2}$ in Definition
\ref{def1}, equivalent to the commutativity of diagram (\ref{7}), reads
\begin{equation}
\gamma_1 f_{i}^{\gamma_2,\gamma_3} \cdot (f_i^{\gamma_1\gamma_2,\gamma_3})^{-1}
\cdot f_i^{\gamma_1,\gamma_2\gamma_3} \cdot (f_i^{\gamma_1,\gamma_2})^{-1}=1\text{.}
\label{13}
\end{equation}

Summarizing, the bundle gerbe has local data
$c:=(B_i,A_{ij},g_{ijk})$, and the $(\Gamma,\epsilon)$-equivariant
structure has local data
$b_{\gamma}:=(\Pi^{\gamma}_i,\chi^{\gamma}_{ij})$ and
$a_{\gamma_1,\gamma_2} := (f_i^{\gamma_1,\gamma_2})$. If the
$(\Gamma,\epsilon)$-equivariant
 structure is normalized, the mixed map $z_1:M_{\mathfrak{O}}
\to Z^1=Y^{[2]}$ defined in (\ref{10}) is $z_1=\Delta \circ s$, so
we may choose $\sigma_i^1:=\Delta^{*}\sigma_{ii}$ and obtain
$\Pi^1_i=A_{ii}=0$. With this choice, we find,
for $\alpha_1=(1 \otimes \pi_{124}^{*}\mu^{-1}) \circ
(\pi_{134}^{*}\mu \otimes 1)$, the local datum
$\chi^{1}_{ij}=g_{iij}^{-1}\cdot g_{ijj}=1$. For
$\varphi_{1,\gamma}$, we get
$f^{1,\gamma}_i=(\chi^{\gamma}_{ii})^{-1}g_{iii}=1$, and,
analogously, $f^{\gamma,1}_i=\chi^{\gamma}_{ii}g_{iii}=1$. Finally,
$f^{1,1}_{i}=g_{iii}=1$.

Let us now extract local data of an equivariant 1-morphism
\begin{equation*}
(\mathcal{B},\eta_{\gamma}): (\mathcal{G},\mathcal{J}) \to (\mathcal{G}',\mathcal{J}')
\end{equation*}
between two $(\Gamma,\epsilon)$-equivariant bundle gerbes over $M$.
We may choose an open covering $\mathfrak{O}$ of $M$ with an action
of $\Gamma$ on its index set as above, such that it admits sections
$s_i:O_i \to Y$ and $s_i':O_i \to Y'$ for both bundle gerbes. The
1-morphism $\mathcal{B}$ provides a vector bundle $B$ over $Z:=Y
\times_M Y'$ of some rank $n$. Generalizing the mixed map
(\ref{10}), we have a map
\begin{equation}
\label{36}
z: M_{\mathfrak{O}} \to Z: (x,i) \mapsto (s_i(x),s'_i(x))\text{.}
\end{equation}
Just as above, we choose an orthonormal frame of sections 
$\sigma^a_i: O_i \to z^{*}B$, $\,a=1,\dots,n$, 
and obtain local connection 1-forms
$\Lambda_i\in\Omega^1(O_i,\mathfrak{u}(n))$ with values in the
set of hermitian $(n\times n)$-matrices, alongside
smooth functions $G_{ij}: O_{ij} \to U(n)$
defined by the relation
\begin{eqnarray}
(z^*\nabla)\sigma_i^a&=&\frac{_1}{^\mathrm{i}}
(\Lambda_i)_b^{\ a}\,\sigma_i^b\cr
z^{*}\beta(\sigma_{ij} \otimes \sigma_j^a)&=&(G_{ij})_b^{\ a}\, 
(\sigma_i^b\otimes \sigma_{ij}')\text{,}
\nonumber
\end{eqnarray}
where $\beta$ is the
isomorphism of $\mathcal{B}$ over $Z \times _M Z$ and $\sigma_{ij}$
and $\sigma_{ij}'$ are the sections chosen to
extract local data of the two bundle gerbes $\mathcal{G}$ and
$\mathcal{G}'$. Analogously to equations (\ref{154}), (\ref{11}) and
(\ref{12}), we have
\begin{eqnarray}
B_{i}'-B_i &=& \frac{_1}{^n}\mathrm{tr}(
\mathrm{d}\Lambda_{i})\textrm{,}
\nonumber\\
A'_{ij} -A_{ij} &=&   \Lambda_j  -G_{ij}^{-1}\cdot \Lambda_i\cdot
G_{ij} -\mathrm{i}G_{ij}^{-1}\mathrm{d}G_{ij}\textrm{,} \label{39}
\\\nonumber
g_{ijk}'\cdot g_{ijk}^{-1}
&=& G_{ik} \cdot G_{jk}^{-1} \cdot G_{ij}^{-1} \text{.}
\end{eqnarray}
Again, the sections $\sigma_i^{a}$ can be chosen such that
$G_{ij}=G_{ji}^{-1}$. The 2-isomorphisms
$\eta_{\gamma}:\gamma\act\mathcal{B} \circ \mathcal{A}_{\gamma}
\Rightarrow \mathcal{A}'_{\gamma} \circ \mathcal{B}$ are,
following the discussion in Section \ref{sec4}, isomorphisms
\begin{equation*}
\eta_{{\gamma}}: \pi_{12}^{*}A_{\gamma} \otimes \pi_{24}^{*}B^{\epsilon(\gamma)}\to \pi_{13}^{*}B\otimes \pi_{34}^{*}A'_{\gamma}
\end{equation*}
of vector bundles over $Z_1^{\gamma} \times_{P} Z^{\gamma}_2 \cong Y
\times_M Y_{\gamma} \times_M Y' \times_M Y_{\gamma}'$, see
(\ref{34}). We compose a map
\begin{equation*}
z_{\gamma}:M_{\mathfrak{O}} \to Z_1^{\gamma} \times_{P} Z^{\gamma}_2
: (x,i) \mapsto (s_i(x),s_{\gamma^{-1}
i}(\gamma^{-1}(x)),s'_i(x),s'_{\gamma^{-1}
i}(\gamma^{-1}(x)))\text{,}
\end{equation*}
 into that space and define smooth functions
$H^{\gamma}_i: O_i \to U(n)$ by
\begin{equation*}
z_{\gamma}^{*}\eta_{\gamma}(\sigma^{\gamma}_i \otimes
(\sigma^a_{\gamma^{-1}i}\circ \gamma^{-1})^{\epsilon(\gamma)}) =
(H^{\gamma}_i)_b^{\ a}\,(\sigma^b_i \otimes \sigma_i'^{\gamma})\text{.}
\end{equation*}
Since $\eta_{\gamma}$ preserves the connections, we obtain
\begin{equation}
\label{42}
\gamma\Lambda_{i} = (H_i^{\gamma})^{-1}\cdot\Lambda_i\cdot 
H_i^{\gamma} +  \mathrm{i} (H_i^{\gamma})^{-1} \mathrm{d}H_i^{\gamma}+ \Pi_i'^{\gamma}-\Pi_i^{\gamma}\text{.}
\end{equation}
The commutativity of diagram (\ref{32}) implies that
\begin{equation}
\label{43}
\gamma G_{ij}=  \chi'^{\gamma}_{ij} \cdot (\chi_{ij}^{\gamma})^{-1} \cdot (H_i^{\gamma})^{-1} \cdot G_{ij} \cdot H_{j}^{\gamma}\text{.}
\end{equation}
Here, we have extended the notation of \erf{54} to
$U(n)$-valued functions and $\mathfrak{u}(n)$-valued 1-forms in the
following way:
\begin{eqnarray}
\gamma\Lambda_i&:=&\begin{cases}(\gamma^{-1})^{*}\Lambda_{\gamma^{-1}i} & \text{if }\epsilon(\gamma)=1 \\
-(\gamma^{-1})^{*}\overline{\Lambda_{\gamma^{-1}i}} & \text{if }\epsilon(\gamma)=-1 \end{cases}
\nonumber\\[-0.2cm]\label{53}
\\[-0.2cm]\nonumber
\gamma G_i&:=&
\begin{cases} (\gamma^{-1})^{*}G_{\gamma^{-1} i} & \text{if } \epsilon(\gamma)=1 \\
 (\gamma^{-1})^{*}\overline{G_{\gamma^{-1} i}} & \text{if }\epsilon(\gamma)=-1
\end{cases}
\end{eqnarray}
with the overbar denoting the complex conjugation. These definitions
coincide with \erf{54} for $n=1$. The commutativity  of
diagram (\ref{33}) leads to
\begin{equation}
\label{41}
H_i^{\gamma_1\gamma_2} \cdot f_i^{\gamma_1,\gamma_2} = f_i'^{\gamma_1,\gamma_2} \cdot H_i^{\gamma_1} \cdot \gamma_1H_{i}^{\gamma_2}\text{.}
\end{equation}
Thus, an equivariant 1-morphism has
local data $\beta:=(\Lambda_i,G_{ij})$ and
$\eta_{\gamma}:=(H_i^{\gamma})$
satisfying (\ref{42}), (\ref{43}) and (\ref{41}).

Finally, for completeness, assume that
$\phi:(\mathcal{B},\eta_{\gamma}) \Rightarrow
(\mathcal{B}',\eta_\gamma')$ is an equivariant
2-morphism, i.e. $\phi:B \to B'$ is a morphism of vector bundles
over $Z$ subject to the two conditions of
Section \ref{sec4}. We use the map $z:M_{\mathfrak{O}} \to Z$ from
(\ref{36}) to pull back $\phi$ and to
extract smooth functions $(U_i)_{a'}^{\ a}$ defined on $O_i$ by
\begin{equation*}
z^{*}\phi(\sigma^a_i) = (U_i)_{a'}^{\ a}\,\sigma'^{a'}_i\text{.}
\end{equation*}
Here, $a=1,\dots,n$ and $a'=1,\dots,n'$, where $n$ and $n'$ are ranks
of the vector bundles $B$ and $B'$, respectively. The condition that
$\phi$ preserves hermitian metrics and the connections yields
\begin{equation}
\label{52}
U^\dagger_i\cdot U_i=1\quad\text{ and }\quad
\Lambda_i = U_i^\dagger\cdot\Lambda'_i\cdot U_i + \mathrm{i}U_i^\dagger
\mathrm{d}U_i\text{.}
\end{equation}
The two remaining
conditions, namely the commutativity of diagrams (\ref{37}) and
(\ref{62}) impose the relations
\begin{equation}
\label{51}
 G_{ij}=U_i^\dagger \cdot G_{ij}'\cdot U_j
\quad\text{ and }\quad
H_i^{\gamma} = U_i^\dagger\cdot H'^{\gamma}_i \cdot \gamma U_{i}\text{.}
\end{equation}

\subsection{Deligne Cohomology-valued Group Cohomology}

It is convenient to put the local data extracted above into the
context of Deligne hypercohomology. We denote  by $\mathcal{U}$ the
sheaf of smooth $U(1)$-valued functions on $M$, and by $\Lambda^{q}$
the sheaf of $q$-forms on $M$. The Deligne
complex in degree 2, denoted by $\mathcal{D}(2)$, is the complex
\begin{equation}
\label{150}
\alxydim{@C=1.5cm}{0 \ar[r] & \mathcal{U} \ar[r]^-{\frac{1}{\mathrm{i}}\mathrm{dlog}}
& \Lambda^1 \ar[r]^{\mathrm{d}} & \Lambda^2}
\end{equation}
of sheaves over $M$. Together with the \v Cech complex of the open
cover $\mathfrak{O}$, it gives a double
complex whose total complex $K(\mathcal{D}(2))$
\begin{equation}
\label{150a}
\alxydim{@C=1.5cm}{0 \ar[r] & A^0 \ar[r]^-{D_0}
& A^1 \ar[r]^{D_1} & A^2 \ar[r]^{D_2} & A^3}
\end{equation}
has the cochain
groups
\begin{eqnarray*}
A^0 &=& C^0(\mathcal{U})\textrm{,} \\
A^1 &=& C^0(\Lambda^1) \oplus C^1(\mathcal{U})\textrm{,} \\
A^2 &=& C^0(\Lambda^2) \oplus C^1(\Lambda^1) \oplus C^2(\mathcal{U})\textrm{,} \\
A^3 &=& C^1(\Lambda^2) \oplus C^2(\Lambda^1) \oplus
C^3(\mathcal{U})\textrm{,}
\end{eqnarray*}
and the differentials
\begin{eqnarray*}
D_0(f_i) &=& (-\mathrm{i}f_i^{-1}\mathrm{d}f_i\;,\; f_j^{-1}\cdot f_i)\textrm{,} \\
D_1(\Pi_i,\chi_{ij}) &=& (\mathrm{d}\Pi_i\;,\;
-\mathrm{i}\chi_{ij}^{-1}\mathrm{d}\chi_{ij}+\Pi_{j}
- \Pi_{i}\;,\;\chi_{jk}^{-1}\cdot \chi_{ik}\cdot \chi_{ij}^{-1})\textrm{,} \\
D_2(B_i,A_{ij},g_{ijk}) &=& (\mathrm{d}A_{ij}-B_{j} + B_i \;,\;
-\mathrm{i}g_{ijk}^{-1}\mathrm{d}g_{ijk}+A_{jk}-A_{ik}+A_{ij}\;,
\nonumber\\
&& \hspace{5cm} g_{jkl}^{-1} \cdot g_{ikl} \cdot g_{ijl}^{-1}\cdot g_{ijk})\textrm{.}
\end{eqnarray*}
The cohomology of this complex is the hypercohomology of the double
complex induced from (\ref{150}). Its
groups are denoted by $\mathbb{H}^k(M,\mathcal{D}(2))$. The local
data $c=(B_i,A_{ij},g_{ijk})$ extracted above from the
bundle gerbe $\mathcal{G}$ are an element of $A^2$, and the properties
(\ref{151}), (\ref{152}) and
(\ref{153}) show  that $D_2(c)=0$. This is
the Deligne cocycle of a bundle gerbe
\cite{gawedzki3}.

In \cite{gawedzki6}, we turned the complex $A^{n}$
into a complex of left $\Gamma$-modules, where the action of
$\Gamma$ is given by \erf{54} and its extension to higher order
forms. Thus, $\gamma c$ are the local data of
$\gamma\act\mathcal{G}$ for the same choice of
sections. The local data
$b_{\gamma}=(\Pi^{\gamma}_i,\chi^{\gamma}_{ij})$ extracted above
from 1-isomorphisms $\mathcal{A}_{\gamma}$
give an element $b_{\gamma} \in A^1$, and the
properties (\ref{154}), (\ref{11}) and (\ref{12}) amount
to the relations
\begin{equation}
\label{160}
\gamma c - c = D_1b_{\gamma}\text{.}
\end{equation}
Furthermore, the local data $a_{\gamma_1,\gamma_2} =
(f_i^{\gamma_1,\gamma_2})$ extracted from
2-isomorphisms $\varphi_{\gamma_1,\gamma_2}$ give an
element $a_{\gamma_1,\gamma_2} \in A^0$, and its properties
(\ref{14}) and (\ref{15}) are equivalent to the identities
\begin{equation}
\label{161}
\gamma_1 b_{\gamma_2} - b_{\gamma_1\gamma_2} + b_{\gamma_2} = -D_0a_{\gamma_1\gamma_2}\text{.}
\end{equation}
The additional constraint (\ref{13}) reads:
\begin{equation}
\label{162} \gamma_1 a_{\gamma_2,\gamma_3} -
a_{\gamma_1\gamma_2,\gamma_3} + a_{\gamma_1,\gamma_2\gamma_3} -
a_{\gamma_1,\gamma_2}=0
\end{equation}
(where the right-hand side represents the trivial \v Cech
cochain $(1)$ in the additive notation). As recognized in
\cite{gawedzki6}, the equations above show that group cohomology is
relevant in the present context. For each $\Gamma$-module
$A^n$, we form the group of $\Gamma$-cochains
$C^{k}(A^n)= \mathrm{Map}(\Gamma^k, A^n)$, together with the
$\Gamma$-coboundary operator
\begin{equation*}
\delta: C^{k}(A^n) \to C^{k+1}(A^n)
\end{equation*}
defined as
\begin{multline*}
(\delta n)_{\gamma_1,...,\gamma_{k+1}} =
\gamma_1n_{\gamma_2,...,\gamma_{k+1}} -
n_{\gamma_1\gamma_2,...,\gamma_{k+1}}
\\+...+(-1)^{k}n_{\gamma_1,...,\gamma_{k}\gamma_{k+1}} +
(-1)^{k+1}n_{\gamma_1,...,\gamma_k}\text{.}
\end{multline*}
When expressed in terms
of this coboundary operator, equations (\ref{160}), (\ref{161}) and
(\ref{162}) become
\begin{equation}
\label{164}
(\delta c)_{\gamma} = D_1b_{\gamma}
\quad\text{, }\quad
(\delta b)_{\gamma_1,\gamma_2} = -D_0a_{\gamma_1,\gamma_2}
\quad\text{ and }\quad
(\delta a)_{\gamma_1,\gamma_2,\gamma_3}=0\text{.}
\end{equation}
These are equations (2.13)-(2.15) of \cite{gawedzki6}.
In group cohomology, a cochain
$n_{\gamma_1,...,\gamma_n}$ is called \emph{normalized} if
$n_{\gamma_1,...,\gamma_n}=0$ whenever
some $\gamma_i=1$. Notice that the cochains $\eta_{\gamma}$ and
$a_{\gamma_1,\gamma_2}$ are normalized if the
$(\Gamma,\epsilon)$-equivariant structure on $\mathcal{G}$ is
normalized.

The coboundary operator $\delta$ commutes with the differential
$D_n$, and so the groups $C^k(A^n)$ form again a
double complex. Its hypercohomology is denoted by
$\mathbb{H}^n(\Gamma,K(\mathcal{D}(2))_{\varepsilon})$, where
the subscript $\varepsilon$ on $K(\mathcal{D}(2))$ indicates that
the action of $\Gamma$ on the latter module is the one inherited
from (\ref{54}). The collection
$(c,b_{\gamma},a_{\gamma_1,\gamma_2})$, representing the bundle
gerbe $\mathcal{G}$ with $(\Gamma,\epsilon)$-equivariant structure,
defines an element in the degree 2 cochain group of this total
complex.  Equations (\ref{164}) together with
$D_2c=0$ show that it is even a cocycle,
defining a class in
$\mathbb{H}^2(\Gamma,K(\mathcal{D}(2))_{\varepsilon})$.

In the same way as above, local data $\beta=(\Lambda_i,G_{ij})$ and
$\eta_{\gamma}=(H_i^{\gamma})$ of an equivariant 1-isomorphism fit
into this framework. The restriction to 1-\emph{iso}morphisms means
that the functions $\Lambda_i$ and $H_i^{\gamma}$ take values in
$U(1)$, and that the 1-forms $\Lambda_{i}$ are real-valued. Notice
that the conventions (\ref{53}) coincide with the definitions (\ref{54})
in the abelian case. Then, $\beta\in A^1$ and $\eta_{\gamma} \in A^0$.
Equations (\ref{39}), equations (\ref{42}) and (\ref{43}), and
equation (\ref{41}) mean
\begin{multline}
\label{44} c'=c+D_{1}\beta \quad\text{, }\quad
b_{\gamma}'=b_{\gamma} + (\delta\beta)_{\gamma} + D_0\eta_{\gamma}
\\\quad\text{ and }\quad
a'_{\gamma_1,\gamma_2}=a_{\gamma_1,\gamma_2}-(\delta\eta)_{\gamma_1,\gamma_2}\text{.}
\end{multline}
These are equations (2.16)-(2.18) of
\cite{gawedzki6}. Thus, $(\Gamma,\epsilon)$-equivariant bundle
gerbes which are related by an equivariant 1-isomorphism define the
same class in
$\mathbb{H}^2(\Gamma,K(\mathcal{D}(2))_{\varepsilon})$.

\begin{proposition}
\label{prop1}
The map
\begin{equation*}
\left \lbrace \txt{Equivalence classes\\of $(\Gamma,\epsilon)$-equivariant\\bundle
gerbes
over $M$} \right \rbrace \to \mathbb{H}^2(\Gamma,K(\mathcal{D}(2))_{\varepsilon})
\end{equation*}
defined by extracting local data as described above is a bijection.
\end{proposition}

\begin{proof}
This follows from the usual reconstruction of bundle gerbes,
1-morphisms and 2-morphisms from given local data, see, e.g.,
\cite{stevenson1}. The reconstructed objects have the property that
they admit local data from which they were reconstructed.
Thus, in order to see the surjectivity, one reconstructs a
bundle gerbe $\mathcal{G}$, the 1-morphisms $\mathcal{A}_{\gamma}$
and the 2-morphisms $\varphi_{\gamma_1,\gamma_2}$ from given local
data. The cocycle condition assures that all necessary diagrams are
commutative, so that one obtains a twisted-equivariant bundle gerbe
whose local data are the given one. To see the injectivity, assume
that the class of the cocycle $c$ of a given twisted-equivariant
bundle gerbe $(\mathcal{G},\mathcal{J})$ is trivial,
$c=D_1(d)$. Then, one can reconstruct an equivariant 1-isomorphism
$(\mathcal{G},\mathcal{J}) \to (\mathcal{I}_0,\mathcal{J}_0)$ from
the cochain $d$.
\end{proof}

The geometric descent theory from Section \ref{sec2} implies,
via Proposition \ref{prop1}, results for the cohomology theories,
namely a bijection
\begin{equation*}
\mathbb{H}^2(\Gamma,K(\mathcal{D}(2))_{\varepsilon})\ \cong\
\mathbb{H}^2(\Gamma',K(\mathcal{D}(2))_{\varepsilon'})\textrm{,}
\end{equation*}
whenever the normal subgroup $\Gamma_0:=\mathrm{ker}(\varepsilon)$
acts without fixed points so that $(\Gamma',\varepsilon')$ is an
orientifold group for the quotient manifold $M'$. In the next
section, we will use Proposition \ref{prop1} in the opposite
direction.

\subsection{Classification Results}

In this section, we present a short summary of the
classification results from \cite{gawedzki6}. In general, there are
obstructions to the existence of a $(\Gamma,\epsilon)$-equivariant
structure on a bundle gerbe
$\mathcal{G}$, and if these vanish, there may be inequivalent
choices thereof. We use Proposition \ref{prop1} to study these
issues in a purely cohomological way. To this end, we are
looking for the image and the kernel of the homomorphism
\begin{equation}
\label{58}
\mathrm{pr}: \mathbb{H}^2(\Gamma,K(\mathcal{D}(2))_{\varepsilon}) \to \mathbb{H}^2(M,\mathcal{D}(2))
\end{equation}
which sends a twisted-equivariant Deligne class to the underlying ordinary
Deligne class.

As shown in \cite{gawedzki6}, the image can be characterized by
hierarchical obstructions to the existence of
twisted-equivariant structures on a given bundle gerbe
$\mathcal{G}$. If we assume that the curvature $H$ of $\mathcal{G}$
is $(\Gamma,\varepsilon)$-equivariant in the sense that $\gamma\act
H = H$ for all $\gamma\in \Gamma$, these obstructions are classes
\begin{multline}
\label{59}
o_1 \in H^2(M,U(1))
\quad\text{, }\quad
o_2  \in H^2(\Gamma,H^1(M,U(1))_{\varepsilon})
\\\quad\text{ and }\quad
o_3  \in H^3(\Gamma,H^0(M,U(1))_{\varepsilon})\text{.}
\end{multline}
The latter two are $\Gamma$-cohomology groups, with the action of $\Gamma$
on the coefficients induced from \erf{54}. The class $o_2$ is
well-defined if $o_1$ vanishes, and $o_3$ is well-defined if $o_1$
and $o_2$ vanish.

Let us now discuss the kernel of the homomorphism \erf{58}, i.e. the
question what the set of equivalence classes of
twisted-equivariant bundle gerbes with isomorphic underlying bundle
gerbe looks like. We infer that $\mathrm{pr}$ is induced from the
projection
\begin{equation}
\label{24}
p^{n}:\bigoplus_{p+q=n} C^p(A^q) \to A^n
\end{equation}
of chain complexes, whose cohomologies are $\mathbb{H}^n(\Gamma,
K(\mathcal{D}(2))_{\varepsilon})$ and
$\mathbb{H}^{n}(M,\mathcal{D}(2))$, respectively. The kernel of
$p^{n}$ forms, again, a complex whose cohomology will be
denoted by $\mathcal{H}^n$. Explicitly, a class in
$\mathcal{H}^2$ is represented by a pair
$(b_{\gamma},a_{\gamma_1,\gamma_2})$ with $b_{\gamma}\in A^1$ and
$a_{\gamma_1,\gamma_2}\in A^0$ such that
\begin{equation*}
\mathrm{D}_1b_{\gamma}=0
\quad\text{, }\quad
(\delta b)_{\gamma_1,\gamma_2}=-\mathrm{D}_0 a_{\gamma_1,\gamma_2}
\quad\text{ and }\quad
(\delta a)_{\gamma_1,\gamma_2\gamma_3}=0\text{.}
\end{equation*}
Equivalent representatives satisfy $b_{\gamma}' = b_{\gamma} +
\mathrm{D}_0\eta_{\gamma}$ and $a'_{\gamma_1,\gamma_2} =
a_{\gamma_1,\gamma_2} - (\delta \eta)_{\gamma_1,\gamma_2}$ for a
collection $\eta_{\gamma}\in A^0$. Comparing this with equations
(20) in \cite{gawedzki6}, we conclude that the group
$H_{\Gamma}$ which we considered there is obtained
from $\mathcal{H}^2$ by additionally identifying cocycles
$(b_{\gamma},a_{\gamma_1,\gamma_2})$ and
$(b'_{\gamma},a'_{\gamma_1,\gamma_2})$ if there exists a $\beta\in
A^1$ such that $\mathrm{D}_1\beta=0$ and
$b_{\gamma}'=b_{\gamma} + (\delta\beta)_{\gamma}$.
\begin{lemma}
\label{lem2}
The group  $\mathcal{H}^2$  fits into the exact sequences
\begin{equation*}
\alxydim{@C=0.9cm}{0 \ar[r] & \mathcal{H}^2\,/\,H^1(M,U(1))
\ar[r] & \mathbb{H}^2(\Gamma, K(\mathcal{D}(2))_{\varepsilon}) \ar[r]^-{\mathrm{pr}} & \mathbb{H}^{2}(M,\mathcal{D}(2))}
\end{equation*}
and
\begin{equation*}
\alxydim{}{0 \ar[r] & H^2(\Gamma,H^0(M,U(1))_{\varepsilon}) \ar[r]
& \mathcal{H}^2 \ar[r] & C^1(H^1(M,U(1)))\text{.}}
\end{equation*}
\end{lemma}

\begin{proof}
The first sequence is just a piece of the long exact sequence
obtained from the short exact sequence which is (\ref{24}) extended
by its kernel to the left, together with the well-known
identification $\mathbb{H}^k(M,\mathcal{D}(2)) \cong H^k(M,U(1))$
for $k=0,1$ \cite{brylinski1}. The second sequence can be obtained
by the same trick: we project out another factor
$q^n: \mathrm{ker}(p^n) \to C^1(A^{n-1})$ from the exact
sequence of complexes, yielding a short exact sequence
\begin{equation*}
\alxydim{}{0 \ar[r] & \mathrm{ker}(q)^{\bullet} \ar[r] &
\mathrm{ker}(p)^{\bullet} \ar[r]^-{q} & C^1(A^{\bullet-1})
\ar[r] & 0}
\end{equation*}
of chain complexes. The interesting part of its long
exact sequence is
\begin{equation*}
\alxydim{}{C^1(H^0(M,U(1))) \ar[r]^-{\delta} &
H^{2}(\mathrm{ker}(q)) \ar[r] & \mathcal{H}^2 \ar[r] &
C^1(H^1(M,U(1)))\text{,}}
\end{equation*}
for which an easy computation shows that
$H^{2}(\mathrm{ker}(q))$ coincides with $\mathrm{ker}(\delta|_{C^2(H^0(M,U(1)))})$. \end{proof}

We have now derived results on the image and the kernel
of the projection \erf{58}. When the
underlying manifold is 2-connected, $H^2(M,U(1))=H^1(M,U(1))=0$ and
$H^0(M,U(1))_{\varepsilon}=U(1)_{\varepsilon}$ as $\Gamma$-modules,
so that the obstructions \erf{59} and Lemma \ref{lem2} boil down to
\begin{proposition}
Let $M$ be a 2-connected smooth manifold and let $\mathcal{G}$ be a bundle gerbe
over $M$ with $(\Gamma,\epsilon)$-equivariant curvature.
\begin{enumerate}
\item[(a)]
$\mathcal{G}$
admits $(\Gamma,\epsilon)$-equivariant structures if and only if the third
obstruction class $o_3 \in H^3(\Gamma,U(1)_{\varepsilon})$ vanishes.

\item[(b)]In the latter case, equivalence classes of
$(\Gamma,\epsilon)$-equi\-va\-ri\-ant bundle gerbes whose underlying
bundle gerbe is isomorphic to $\mathcal{G}$ are parameterized by the
group $H^2(\Gamma,U(1)_{\varepsilon})$.
\end{enumerate}
\end{proposition}

\noindent This  was the starting point for the calculations in
finite-group cohomology of \cite{gawedzki6}. Namely, on a
compact connected simple and simply connected Lie group, there
is a canonical family $\mathcal{G}_k$ of bundle gerbes with
$(\Gamma,\epsilon)$-equivariant curvature for $\Gamma$  a semidirect
product of $\mathbb{Z}_2$ (generated by the $\zeta$-twisted
inversion $g\to\zeta\cdot g^{-1}$, with $\zeta$ from the center of
$G$) and a subgroup of the center of $G$. Since these Lie groups
are 2-connected, the obstruction classes and the classifying groups
for $(\Gamma,\epsilon)$-equivariant structures on $\mathcal{G}_k$ may be
computed by calculations in finite-group cohomology.

\section{Equivariant Gerbe Modules}

\label{sec5}

Gerbe modules can be described conveniently as 1-morphisms
\cite{waldorf1}:
\begin{definition}
Let $\mathcal{G}$ be a bundle gerbe over $M$. A \uline{$\mathcal{G}$-module}  is a 1-morphism
\begin{equation*}
\mathcal{E}:\mathcal{G} \to \mathcal{I}_{\omega}\text{.}
\end{equation*}
The rank of the vector bundle of $\mathcal{E}$ is called the
\uline{rank} of the bundle-gerbe module, and the 2-form
$\omega$ is called its \uline{central curvature}. \end{definition}

Let us extract the details of this definition. We assume that the
1-morphism $\mathcal{E}$ is descended in the sense that it consists
of a vector bundle $E$ over $Y\cong Y \times_M M$. Similarly as
in Lemma \ref{lem1}, this can be assumed up to natural
2-isomorphisms, see Theorem 1 in \cite{waldorf1}. By axiom
(1M1), the curvature of the vector bundle $E$
satisfies
\begin{equation*}
\frac{_1}{^n}\mathrm{tr}(\mathrm{curv}(E)) = \pi^{*}\omega - C\text{.}
\end{equation*}
The $\mathcal{G}$-module consists also of an isomorphism
\begin{equation*}
\rho: L \otimes \pi_2^{*}E \to \pi_1^{*}E
\end{equation*}
of vector bundles over $Y^{[2]}$ which satisfies, by axiom
(1M2), the condition
\begin{equation}
\label{21} \pi_{13}^{*}\rho \circ (\mu \otimes \id)=
\pi_{12}^{*}\rho \circ (\id \otimes \pi_{23}^{*}\rho)\text{.}
\end{equation}
The latter resembles the axiom for an action $\rho$ of a monoid $L$
on a module $E$, hence the terminology. The above definition of a
bundle-gerbe module coincides with the usual one, see, e.g.,
\cite{bouwknegt1,gawedzki4}.

\begin{definition}
\label{def3} Let $(\Gamma,\epsilon)$ be an orientifold group for $M$
and let  $(\mathcal{G},\mathcal{J})$ be a
$(\Gamma,\epsilon)$-equivariant bundle gerbe over $M$. A
\uline{$(\mathcal{G},\mathcal{J})$-module} is a 2-form $\omega$ on
$M$ that satisfies the condition $\gamma\act\omega =\omega$
for all $\gamma\in\Gamma$, together with an equivariant 1-morphism
\begin{equation*}
(\mathcal{E},\rho_{\gamma}): (\mathcal{G},\mathcal{J}) \to (\mathcal{I}_{\omega},\mathcal{J}_{\omega})\text{,}
\end{equation*}
where $\mathcal{J}_{\omega}$ is the canonical $(\Gamma,\epsilon)$-equivariant structure on the trivial bundle gerbe $\mathcal{I}_{\omega}$ from Example \ref{ex1}.
\end{definition}

Thus,  a $(\mathcal{G},\mathcal{J})$-module is a $\mathcal{G}$-module
$\mathcal{E}:\mathcal{G} \to \mathcal{I}_{\omega}$ together with a 2-morphism
\begin{equation*}
\rho_{\gamma}: \gamma\act\mathcal{E} \circ \mathcal{A}_{\gamma} \Rightarrow  \mathcal{E}
\end{equation*}
for every $\gamma\in\Gamma$, such that the diagram
\begin{equation}
\label{23}
\alxydim{@C=2.5cm@R=1.2cm}{\gamma_1\act\gamma_2\act\mathcal{E} \circ \gamma_1\act\mathcal{A}_{\gamma_2} \circ \mathcal{A}_{\gamma_1} \ar@{=>}[d]_-{\gamma\act_1\rho_{\gamma_2} \circ \id_{\mathcal{A}_{\gamma_1}}} \ar@{=>}[r]^-{\id_{\gamma_1\gamma_2\act \mathcal{E}} \circ \varphi_{\gamma_1,\gamma_2}} & \gamma_1\gamma_2\act \mathcal{E} \circ \mathcal{A}_{\gamma_1\gamma_2} \ar@{=>}[d]^-{\rho_{\gamma_1\gamma_2}} \\   \gamma_1\act \mathcal{E}\circ \mathcal{A}_{\gamma_1} \ar@{=>}[r]_-{ \rho_{\gamma_1}} &  \mathcal{E}}
\end{equation}
is commutative for all $\gamma_1,\gamma_2\in\Gamma$. We say that a
$(\mathcal{G},\mathcal{J})$-module is \emph{normalized} if the
equivariant 1-morphism $(\mathcal{E},\rho_{\gamma})$ is normalized.

We already discussed equivariant 1-morphisms in
terms of vector bundles and isomorphisms of vector bundles in
Section \ref{sec4}, so that we only have to apply
these results to the particular case at hand. We recall that the
$(\Gamma,\epsilon)$-equivariant structure on $\mathcal{G}$ consists
of a line bundle $A_{\gamma}$ over $Z^{\gamma}$ of curvature
$\mathrm{curv}(A_{\gamma}) = \epsilon(\gamma)\pi_{2}^{*}C -
\pi_1^{*}C$ for each $\gamma\in\Gamma$, and of isomorphisms
\begin{equation*}
\alpha_{\gamma}: \pi_{13}^{*}L
\otimes \pi_{34}^{*}A_{\gamma} \to \pi_{12}^{*}A_{\gamma}
\otimes \pi_{24}^{*}L^{\epsilon(\gamma)}
\end{equation*}
of line bundles over $Z^{\gamma} \times_M Z^{\gamma}$ subject
to various conditions. The
$\mathcal{G}$-module  $\mathcal{E}: \mathcal{G} \to
\mathcal{I}_{\omega}$ consists of a vector bundle $E$ over $Y$ and
of an isomorphism $\rho:L \otimes \pi_2^{*}E \to \pi_1^{*}E$ over
$Y^{[2]}$ satisfying (\ref{21}). The 2-morphisms $\rho_{\gamma}$ are
isomorphisms
\begin{equation}
\rho_{{\gamma}}: \pi_{12}^{*}A_{\gamma} \otimes \pi_{2}^{*}E^{\epsilon(\gamma)}\to \pi_{1}^{*}E
\end{equation}
of vector bundles over $Z_1^{\gamma} \times_{P} Z^{\gamma}_2$, see
(\ref{34}), which is just $Z^{\gamma}$ here.
The compatibility condition (\ref{32}) now reads
\begin{equation}
\alxydim{@C=1.2cm@R=1.2cm}{\pi_{13}^{*}L \otimes \pi_{34}^{*}A_{\gamma}
\otimes \pi_{4}^{*}E^{\epsilon(\gamma)} \ar[d]_{\id \otimes \pi_{34}^{*}\rho_{{\gamma}}} \ar[r] & \pi_{12}^{*}A_{\gamma} \otimes \pi_{2}^{*}E^{\epsilon(\gamma)} \ar[d]^{\pi_{12}^{*}\rho_{{\gamma}}}
\\ \pi_{13}^{*}L \otimes \pi_{3}^{*}E \ar[r]_{\pi_{13}^{*}\rho} &
\pi_1^{*}E\textrm{,}}
\end{equation}
and the commutative diagram (\ref{23}), which is a
specialisation of (\ref{33}), becomes
\begin{equation}
\label{100}
\alxydim{@C=2cm@R=1.2cm}{\pi_{12}^{*}A_{\gamma_1} \otimes \pi_{23}^{*}A^{\epsilon(\gamma_1)}_{\gamma_2} \otimes \pi_3^{*}E^{\epsilon(\gamma_1\gamma_2)} \ar[r]^-{\pi_{12}^{*}\varphi_{\gamma_1,\gamma_2}
\otimes \id} \ar[d]_{\id \otimes \pi_{23}^{*}\rho^{\epsilon(\gamma_1)}_{{\gamma_2}}} & \pi_{13}^{*}A_{\gamma_1\gamma_2}
\otimes \pi_{3}^{*}E^{\epsilon(\gamma_1\gamma_2)} \ar[d]^{\pi_{13}^{*}\rho_{{\gamma_1\gamma_2}}} \\ \pi_{12}^{*}A_{\gamma_1} \otimes \pi_2^{*}E^{\epsilon(\gamma_1)} \ar[r]_{\pi_{12}^{*}\rho_{{\gamma_1}}} & \pi_1^{*}E\text{.}}
\end{equation}

\begin{definition}
\label{def6} A $(\mathcal{G}^a,\mathcal{J}^a)$-module
$(\mathcal{E}^a,\rho_{\gamma}^a)$ and a
$(\mathcal{G}^b,\mathcal{J}^b)$-module
$(\mathcal{E}^b,\rho_{\gamma}^b)$ are called \emph{equivalent} if
there exists an equivariant 1-isomorphism
$(\mathcal{B},\eta_{\gamma}):(\mathcal{G}^a,\mathcal{J}^a) \to
(\mathcal{G}^b,\mathcal{J}^b)$ and an equivariant 2-isomorphism
\begin{equation*}
\nu:(\mathcal{E}^b,\rho_{\gamma}^b) \circ (\mathcal{B},\eta_{\gamma}) \Rightarrow (\mathcal{E}^a,\rho_{\gamma}^a)\text{.}
\end{equation*}
\end{definition}
In particular, the bundle gerbes $\mathcal{G}^a$ and $\mathcal{G}^b$
are isomorphic,  the 2-forms $\omega^a$ and $\omega^b$ of the two
gerbe modules coincide, and the two modules have the same rank. If
the 1-isomorphism $\mathcal{B}$ has a line bundle $B$ over $Z$ and
an isomorphism $\beta$, this equivariant 2-morphism is just an
isomorphism
\begin{equation}
\label{63}
\nu: B \otimes \pi_2^{*}E^b \to \pi_1^{*}E^a
\end{equation}
of line bundles over $Z$ that satisfies the usual axiom for
2-isomorphisms and the additional equivariance condition \erf{62},
which now becomes the commutative diagram
\begin{equation*}
\alxydim{@C=2cm@R=1.2cm}{
\pi_{12}^{*}A^a_{\gamma} \otimes (\pi_{24}^{*}B \otimes \pi_4^{*}E^b)^{\epsilon(\gamma)}\ \ar[r]^-{\pi_{34}^*\rho^b_\gamma\circ\eta_{\gamma}\otimes\id}
\ar[d]_{\id \otimes \pi_{24}^{*}\nu^{\epsilon(\gamma)}}
&\ \pi_{13}^{*}B \otimes \pi_3^{*}E^b
\ \ar[d]^{\pi_{13}^{*}\nu\otimes
\id}  \\ \pi_{12}^{*}A^{a}_{\gamma} \otimes \pi_{2}^{*}(E^a)^{\epsilon(\gamma)} \ar[r]_-{\rho^a_{\gamma}}
& \pi_{1}^{*}E^a\text{.}}
\end{equation*}

Concerning the local data of a $(\mathcal{G},\mathcal{J})$-module,
we only have to specialize the local data of an
equivariant 1-morphism to the case in which the
second bundle gerbe is a trivial one equipped with its canonical
$(\Gamma,\epsilon)$-equivariant structure, see Section \ref{sec3_1}.
Thus, let $c=(B_i,A_{ij},g_{ijk})$ be local data of
the bundle gerbe $\mathcal{G}$ with respect to some invariant open
cover $\mathfrak{O}$, and let
$b_{\gamma}=(\Pi^{\gamma}_i,\chi^{\gamma}_{ij})$ and
$a_{\gamma_1,\gamma_2} = (f_i^{\gamma_1,\gamma_2})$ be local data of
the $(\Gamma,\epsilon)$-equivariant structure. Evidently, the
trivial bundle gerbe $\mathcal{I}_{\omega}$ has local data
$c'=(\omega,0,1)$, and its canonical equivariant structure
$\mathcal{J}_{\omega}$ has local data $b_{\gamma}'=(0,1)$ and
$a_{\gamma_1,\gamma_2}'=(1)$. A $(\mathcal{G},\mathcal{J})$-module
of rank $n$, i.e. an equivariant 1-morphism
\begin{equation*}
(\mathcal{E},\rho_{\gamma}): (\mathcal{G},\mathcal{J}) \to (\mathcal{I}_{\omega},\mathcal{J}_{\omega})\text{,}
\end{equation*}
has local data $\beta=(\Lambda_i,G_{ij})$ and
$\eta_{\gamma}=(H_i^{\gamma})$ satisfying equations \erf{39},
\erf{42}, \erf{43} and \erf{41}. Explicitly, we have 1-forms
$\Lambda_i \in \Omega^1(O_i,\mathfrak{u}(n))$, and smooth functions
$G_{ij}:O_{ij} \to U(n)$ and $H_i^{\gamma}:O_{i} \to U(n)$. The
equations are
\begin{multline*}
\omega = B_i + \frac{_1}{^n}\mathrm{tr}( \mathrm{d}\Lambda_{i})
\quad\text{, }\quad
 \Lambda_j = G_{ij}^{-1} \Lambda_i G_{ij}  -A_{ij} +\mathrm{i}G_{ij}^{-1}\mathrm{d}G_{ij}
\\\quad\text{ and }\quad
 G_{ij}\cdot G_{jk}
= g_{ijk} \cdot G_{ik}   \text{.}
\end{multline*}
These are just the relations for an ordinary
$\mathcal{G}$-module, see equations (2.3) in \cite{gawedzki4}.
Equivariance is expressed by the relations
\begin{eqnarray}
\nonumber
\gamma\Lambda_i
&=&(H^\gamma_i)^{-1}
\Lambda_i H^\gamma_i+\mathrm{i}(H^\gamma_i)^{-1}\mathrm{d}H^\gamma_i-\Pi^\gamma_i\textrm{,} \\
\label{79}
\gamma G_{ij}&=&(H^\gamma_i)^{-1}
 \cdot G_{ij} \cdot H^\gamma_j\cdot(\chi^\gamma_{ij})^{-1}\textrm{,}
 \\\nonumber
H^{\gamma_1\gamma_2}_i&=&H_i^{\gamma_1} \cdot \gamma_1 H_{i}^{\gamma_2}\cdot (f_i^{\gamma_1,\gamma_2})^{-1}\text{,}
\end{eqnarray}
where we have used the conventions \erf{54} and \erf{53}.

If a $(\mathcal{G},\mathcal{J})$-module
$(\mathcal{E},\rho_{\gamma})$ and a
$(\mathcal{G}',\mathcal{J}')$-module $(\mathcal{E}',\rho_{\gamma}')$
are equivalent in the sense of Definition \ref{def6}, and
$(c,b_{\gamma},a_{\gamma_1,\gamma_2})$ and
$(c',b_{\gamma}',a_{\gamma_1,\gamma_2}')$ are local data of the two
gerbes, there exist local data $(R_i,u_{ij})$ and $(h_i^{\gamma})$
of the equivariant 1-isomorphism $(\mathcal{B},\eta_{\gamma})$
satisfying equations (\ref{44}). There are also functions $U_i: O_i
\to U(n)$ coming from the equivariant 2-morphism $\nu$. If
$\beta=(\Lambda_i,G_{ij})$ and $\rho_{\gamma}=(H_i^{\gamma})$ are
local data of $(\mathcal{E},\rho_{\gamma})$, and, similarly,
$\beta'$ and $\rho'_{\gamma}$ are those of
$(\mathcal{E}',\rho_{\gamma}')$, equations (\ref{52}) and (\ref{51})
take the form
\begin{multline*}
\Lambda'_i=U_i^{-1}\Lambda_i U_i+\mathrm{i}U_i^{-1}\mathrm{d}U_i-R_i
\quad\text{,}\quad
G'_{ij}=U_i^{-1} \cdot G_{ij} \cdot U_j \cdot u_{ij}^{-1}\,
\\\quad\text{ and }\quad
H'^\gamma_i=U_i^{-1} \cdot H^\gamma_i
\cdot \gamma U_{i}\cdot(h_i^\gamma)^{-1}\text{.}
\end{multline*}

A particular situation that we shall discuss explicitly is the
orientifold group $(\Z_2,\id)$, and a bundle gerbe $\mathcal{G}$
with (normalized) Jandl structure
$\mathcal{J}=(\mathcal{A}_k,\varphi_{k,k})$. In this situation, we
call a $(\mathcal{G},\mathcal{J})$-module a \emph{Jandl module}.
Given such a (normalized) Jandl module $(\mathcal{E},\rho_k)$,
$\,\mathcal{E}:\mathcal{G} \to \mathcal{I}_{\omega}$ is a
bundle-gerbe module whose curvature satisfies
$k^{*}\omega=-\omega$, and there is a single 2-isomorphism
\begin{equation*}
\rho_k: k^{*}\mathcal{E}^{\dagger} \circ \mathcal{A}_k
\Rightarrow \mathcal{E}
\end{equation*}
such that the diagram
\begin{equation}
\label{49}
\alxydim{@C=2cm@R=1.2cm}{\mathcal{E} \circ k^{*}\mathcal{A}_k^{\dagger}
\circ \mathcal{A}_k \ar@{=>}[r]^-{\id \circ \varphi_{k,k}} \ar@{=>}[d]_-{k^{*}\rho_k^{\dagger} \circ \id} & \mathcal{E}
\circ \id_{\mathcal{G}} \ar@{=>}[d]^-{\lambda_{\mathcal{E}}}
\\ k^{*}\mathcal{E}^{\dagger} \circ \mathcal{A}_k \ar@{=>}[r]_-{\rho_k}
& \mathcal{E}}
\end{equation}
of 2-isomorphisms is commutative. Still more
specifically, we assume that there is a trivialization
$\mathcal{T}:\mathcal{G} \to \mathcal{I}_{\rho}$. As discussed in
Section \ref{sec15}, the trivialization and the Jandl structure
induce a $k$-equivariant line bundle $(R,\phi)$ over $M$ of
curvature $-(k^{*}\rho+\rho)$. This
was obtained by applying the functor $\bun$
of \erf{55} to the 1-isomorphism
\begin{equation*} \mathcal{R} =k^{*}\mathcal{T}^{\dagger}\circ
\mathcal{A}_{k} \circ \mathcal{T}^{-1}: \mathcal{I}_{\rho} \to
\mathcal{I}_{-k^{*}\rho}\text{.}
\end{equation*}
In the same way, we form the 1-morphism $\mathcal{E} \circ
\mathcal{T}^{-1}: \mathcal{I}_{\rho} \to \mathcal{I}_{\omega}$,
and get a vector bundle $E:= \bun(\mathcal{E} \circ \mathcal{T}^{-1})$.
We have, further, a 2-isomorphism
\begin{multline*}
\alxydim{@C=1cm@R=1.2cm}{ k^{*}(\mathcal{E} \circ \mathcal{T}^{-1})^{\dagger} \circ \mathcal{R} \ar@{=}[r]&
k^{*}\mathcal{E}^{\dagger} \circ
k^{*}\mathcal{T}^{\dagger-1} \circ k^{*}\mathcal{T}^{\dagger} \circ
\mathcal{A}_k \circ \mathcal{T}^{-1} \ar@{=>}[d]^-{\id \circ i_l
\circ \id} \\ & k^{*}\mathcal{E}^{\dagger} \circ
\id_{k^*\mathcal{G}^*} \circ \mathcal{A}_k \circ \mathcal{T}^{-1} \ar@{=>}[d]^-{\id \circ \rho_{\mathcal{A}_{k}}
\circ \id}
\\& k^{*}\mathcal{E}^{\dagger} \circ  \mathcal{A}_k \circ
\mathcal{T}^{-1} \ar@{=>}[r]^-{\rho_k \circ \id}  & \mathcal{E}
\circ \mathcal{T}^{-1}}
\end{multline*}
that induces, via $\bun$, an isomorphism
\begin{equation}
\label{56}
\vartheta: R \otimes k^{*}E^{*} \to E
\end{equation}
of vector bundles over $M$. Finally, diagram \erf{49} implies that
this morphism is compatible with the equivariant structure $\phi$ on
$R$ in the sense that the diagram
\begin{equation*}
\alxydim{@R=1.2cm@C=0.4cm}{R \otimes k^{*}R^{*} \otimes E \ar[rd]_-{\phi \otimes \id} \ar[rr]^-{\id \otimes k^{*}\vartheta^{*}} && R \otimes k^{*}E^{*} \ar[dl]^{\vartheta} \\ & E &}
\end{equation*}
of morphisms of vector bundles over $M$ is commutative. Summarizing,
every Jandl module for a trivialized Jandl
gerbe gives rise to a vector bundle together with an isomorphism
\erf{56}.

In Section \ref{sec2}, we described the descent
theory of twisted-equivariant bundle gerbes as a way to obtain (all)
Jandl gerbes over a smooth manifold $M'$. In the same way, we have

\begin{proposition}
\label{prop2} Let $(\Gamma,\epsilon)$ be an orientifold group for
$M$ with $\Gamma_0$ acting without fixed points, and let
$(\Gamma',\varepsilon')$ be the quotient orientifold group for
the quotient $M':=M/\Gamma_0$. Then, there is a canonical bijection
\begin{equation*}
\alxy{\left \lbrace \txt{Equivalence classes of\\equivariant modules for\\ $(\Gamma,\varepsilon)$-equivariant\\bundle gerbes over $M$} \right \rbrace \ar[r]^-{\cong} & \left \lbrace \txt{Equivalence classes of equi-\\variant modules for $(\Gamma',\varepsilon')$-equi-\\variant bundle gerbes over $M'$} \right \rbrace}\text{.}
\end{equation*}
\end{proposition}
\noindent Notice that the latter proposition unites (as did Theorem
\ref{th2} before) the two cases of $\Gamma_0=\Gamma$ and
$\Gamma/\Gamma_0=\Z_2$.

Since a $(\mathcal{G},\mathcal{J})$-module is nothing but an
equivariant 1-morphism between $(\mathcal{G},\mathcal{J})$ and
$(\mathcal{I}_{\omega},\mathcal{J}_{\omega})$, we can apply the
descent theory developed in Section \ref{sec2}
and pass to an associated quotient 1-morphism.
The only thing to notice is that the quotients are
$(\mathcal{I}_{\omega})'\cong \mathcal{I}_{\omega'}$ for the
descended 2-form $\omega'\in\Omega^2(M')$, and
$(\mathcal{J}_{\omega})'\cong \mathcal{J}_{\omega'}$. But this is
clear since all involved line bundles are the trivial ones, and all
involved isomorphisms are identities.
This way, it becomes obvious how the map
in Proposition \ref{prop2} is defined and that it is surjective.
It remains to check that it is well-defined on
equivalence classes and injective. For this purpose, we have to
amend the discussion of Section \ref{sec2} by
providing a descent construction for equivariant 2-morphisms.
Suppose that we have a $(\mathcal{G}^a,\mathcal{J}^a)$-module
$(\mathcal{E}^a,\rho_{\gamma}^a)$ and an equivalent
$(\mathcal{G}^b,\mathcal{J}^b)$-module
$(\mathcal{E}^b,\rho_{\gamma}^b)$, i.e. there is an equivariant
1-isomorphism
$(\mathcal{B},\eta_{\gamma}):(\mathcal{G}^a,\mathcal{J}^a) \to
(\mathcal{G}^b,\mathcal{J}^b)$ and an equivariant 2-isomorphism
\begin{equation}
\label{65}
\nu:(\mathcal{E}^b,\rho_{\gamma}^b) \circ (\mathcal{B},\eta_{\gamma}) \Rightarrow (\mathcal{E}^a,\rho_{\gamma}^a)\text{.}
\end{equation}
We have to construct an equivariant 2-isomorphism
\begin{equation}
\label{66}
\nu': (\mathcal{E}^{b\prime},\rho_{k}^{b\prime}) \circ (\mathcal{B}',\eta_{k}') \Rightarrow (\mathcal{E}^{a \prime},\rho_{k}^{a\prime})
\end{equation}
which guarantees that the quotient Jandl modules are equivalent.
Notice that the 1-morphism on the right side has the vector bundle
$E^a$ over $Y^a$, and the one on the left side has a vector bundle
over the disjoint union of $Y^a \times_M Y^b \times_M Y^b_{\gamma}$
over all $\gamma\in \Gamma_0$, which is defined componentwise as
$\pi_{12}^{*}B \otimes \pi_{23}^{*}A_{\gamma}^b\otimes
\pi_{3}^{*}E^b$. This follows from the definition of quotient
1-morphisms and from Definition \ref{def2}. Thus, the 2-morphism we
have to construct has components $\nu'_{\gamma} : \pi_{12}^{*}B
\otimes \pi_{23}^{*}A_{\gamma}^b \otimes \pi_{3}^{*}E^b \to
\pi_{1}^{*}E^a$, and we define them as
\begin{equation*}
\alxydim{@C=1.5cm}{\pi_{12}^{*}B \otimes \pi_{23}^{*}A_{\gamma}^b
\otimes \pi_{3}^{*}E^b \ar[r]^-{\id \otimes \rho_{\gamma}^b} &
\pi_{12}^{*}B \otimes \pi_2^{*}E^b \ar[r]^-{\nu} &
\pi_{1}^{*}E^a\text{,}}
\end{equation*}
where $\nu$ comes from the given 2-morphism as in
\erf{63}. It is straightforward to check that
this, indeed, defines an equivariant 2-isomorphism.

Conversely, if an equivariant 1-isomorphism
\begin{equation*}
(\mathcal{B}',\eta'_k):
(\mathcal{G}^{a\prime},\mathcal{J}^{a\prime}) \to
(\mathcal{G}^{b\prime},\mathcal{J}^{b\prime})
\end{equation*}
is given, every
equivariant 2-isomorphism \erf{66} immediately induces an
equivariant 2-isomorphism \erf{65} for
$(\mathcal{B},\eta_{\gamma})$ the equivariant 1-isomorphism
constructed on \prf{67}. This shows that the map from Proposition
\ref{prop2} is injective.

\section{Holonomy for unoriented Surfaces}

We show that a Jandl gerbe over a smooth manifold
$M$ together with Jandl-gerbe modules over  submanifolds
of $M$ provides a well-defined notion of holonomy for
unoriented surfaces with boundary, for example
for the Möbius strip. This notion merges the holonomy for
unoriented closed surfaces from \cite{schreiber1} with that
of the holonomy for oriented surfaces with boundary
from \cite{gawedzki1,carey2}. In the first subsection, we discuss
its definition in terms of geometric structures, and then we develop
expressions in terms of local data.

\subsection{Geometrical Definition}

In short, holonomy arises by pulling back a bundle gerbe
$\mathcal{G}$ along a smooth map $\phi:\Sigma \to M$ to a
surface $\Sigma$, where it becomes trivializable for
dimensional reasons. For any choice of a trivialization
$\mathcal{T}:\phi^{*}\mathcal{G} \to \mathcal{I}_{\rho}$, there is a
number
\begin{equation}
\label{26}
\mathrm{Hol}_{\mathcal{G}}(\phi,\Sigma)\,:=\, \exp
\left ( \im \int_{\Sigma} \rho \right )\, \in\, U(1).
\end{equation}
The integral requires $\Sigma$ to be oriented, and its independence
of the choice of $\mathcal{T}$ requires $\Sigma$ to be closed.

If $\Sigma$ has a boundary, the expression
\erf{26} is no longer well-defined since a boundary term emerges
under a change of the trivialization.  We shall assume for
simplicity that the boundary has only one connected component.
Compensating the boundary term then requires  choices of a
\emph{$\mathcal{G}$-brane} \cite{gawedzki1,carey2,gawedzki4}, a
submanifold $Q \subset M$ together with a $\mathcal{G}|_{Q}$-module
$\mathcal{E}: \mathcal{G}|_{Q} \to \mathcal{I}_{\omega}$. The maps
$\phi:\Sigma \to M$ which we take into account are now
supposed to satisfy $\phi(\partial \Sigma) \subset Q$. If $E$ is the
vector bundle $\bun(\phi^{*}\mathcal{E} \circ \mathcal{T}^{-1})$
over $\partial\Sigma$ constructed in Section \ref{sec5}, the formula
\begin{equation}
\label{28} \mathrm{Hol}_{\mathcal{G},\mathcal{E}}(\phi,\Sigma)\, :=\,
\exp \left ( \im \int_{\Sigma} \rho \right ) \cdot \mathrm{tr}\left
( \mathrm{Hol}_E(\partial\Sigma) \right )\, \in\, \C\textrm{,}
\end{equation}
written in terms of the vector bundle $E$ of $\mathcal{E}$ is invariant
under changes of the trivialization $\mathcal{T}$. If the boundary
is empty, it reduces to \erf{26}. A generalization to
several $\mathcal{G}$-branes in the case of more than one boundary
component is straightforward.

If $\Sigma$ is unoriented, e.g., if it is unorientable, it is
important to notice that there is a unique two-fold covering
$\mathrm{pr}:\hat\Sigma \to \Sigma$, called the \emph{oriented
double}, where $\hat\Sigma$ is oriented and equipped
with an orientation-reversing
involution $\sigma:\hat\Sigma \to\hat\Sigma$ that permutes the
sheets of $\hat\Sigma$ so that $\Sigma=\hat\Sigma/\sigma$. To
obtain holonomy for unoriented surfaces, two changes in the above
setup have to be made \cite{schreiber1}. First, the bundle gerbe
$\mathcal{G}$ has to be equipped with a Jandl structure, i.e. a
twisted-equivariant structure with respect to an involution $k:M \to
M$. Second, the holonomy is taken for smooth maps
$\hat\phi:\hat\Sigma \to M$ which are equivariant in the sense that
the diagram
\begin{equation*}
\alxydim{@=1.2cm}{\hat\Sigma \ar[r]^{\hat \phi}
\ar[d]_{\sigma} & M \ar[d]^{k}
\\ \hat\Sigma \ar[r]_{\hat\phi}
& M}
\end{equation*}
is commutative. This is just the stack-theoretic way to talk about a
smooth map $\Sigma \to M/k$ without requiring that the quotient
$M/k$ be a smooth manifold.

The pullback of the Jandl gerbe
$(\mathcal{G},\mathcal{J})$ along $\hat\phi$ is a Jandl gerbe over
the surface $\hat\Sigma$, and hence trivializable. As discussed in
Section \ref{sec15}, any choice of a trivialization
$\mathcal{T}:\hat\phi^{*}\mathcal{G} \to \mathcal{I}_{\rho}$ defines
a $\sigma$-equivariant line bundle $(\hat R,\hat \varphi)$ over
$\hat\Sigma$ of curvature $-(\sigma^{*}\rho + \rho)$, which, in
turn, descends to a line bundle $R$ over $\Sigma$. To define the
holonomy, we further need to choose a
\emph{fundamental domain} $F$ of $\Sigma$ in $\hat\Sigma$. This is a
submanifold $F\subset \hat\Sigma$ with (possibly piecewise smooth)
boundary such that
\begin{equation}
\label{83}
F \cap \sigma(F) \subset \partial F
\quad\text{ and }\quad
F \cup \sigma(F) = \hat\Sigma\text{,}
\end{equation}
see Figure \ref{fig1} for an example.
\begin{figure}[h]
\begin{center}
\includegraphics[width=\textwidth]{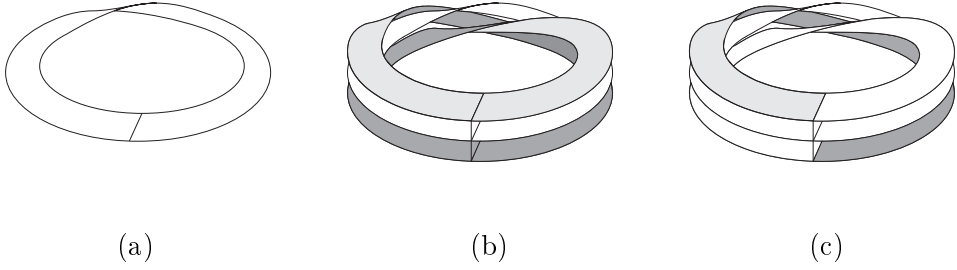}\setlength{\unitlength}{1pt}\begin{picture}(0,0)(466,669)\end{picture}
\vspace{0.3cm}
\\
(a)\hspace{4.2cm}(b)\hspace{4cm}(c)
\end{center}
\caption{(a) shows the Möbius strip. (b) is a Möbius strip (in the middle layer) together with its oriented double. The latter is an ordinary strip with a bright and a dark side. (c) shows a fundamental domain.}
\label{fig1}
\end{figure}
In the case of a closed surface $\Sigma$, it is a key observation
that the involution $\sigma$
restricts to an orientation-\emph{preserving} involution on the
boundary $\partial F$, so that the quotient $\bar F := \partial
F/\sigma$ is a closed oriented 1-dimensional submanifold of $\Sigma$
\cite{schreiber1}. Then, the holonomy is defined by
\begin{equation}
\label{27}
\mathrm{Hol}_{\mathcal{G},\mathcal{J}}(\hat\phi,\Sigma) := \mathrm{exp} \left (
\mathrm{i} \int_{F}\rho
\right ) \cdot \mathrm{Hol}_{R}(\bar F)\text{.}
\end{equation}
This expression is invariant under changes
of the trivialization $\mathcal{T}$ and
of the fundamental domain $F$ \cite{schreiber1}. In the case
when the surface $\Sigma$ is orientable, the oriented double
$\hat\Sigma$ has two global sections $s:\Sigma\to\hat\Sigma$
 intertwined by composition with the involution $\sigma$. The choice
 $F:=s(\Sigma)$ with $\partial F=\emptyset$ satisfies $\mathrm{Hol}_{\mathcal{G},\mathcal{J}}(\hat\phi,\Sigma) = \mathrm{Hol}_{\mathcal{G}}(\hat\phi \circ
 s,\Sigma)$, where on the right
side $\Sigma$ is taken with the orientation pulled back by $s$ from
$\hat\Sigma$.

Below, we introduce a simultaneous generalization of the formul\ae ~\erf{27}
and \erf{28} appropriate for unoriented surfaces \emph{with}
boundary. In addition to the choice of a
Jandl structure on the bundle gerbe $\mathcal{G}$, the following new
structure will be required.

\begin{definition}
Let $\mathcal{G}$ be a bundle gerbe over $M$ and let $\mathcal{J}$
be a Jandl structure on $\mathcal{G}$ with involution $k:M \to M$. A
$(\mathcal{G},\mathcal{J})$-brane is a submanifold $Q\subset M$ such
that $k(Q) = Q$, together with a
$(\mathcal{G},\mathcal{J})|_Q$-module $(\mathcal{E},\rho_k)$.
\end{definition}

We consider maps $\hat\phi:\hat\Sigma \to M$ that satisfy the
boundary condition $\hat\phi(\partial\hat\Sigma)\subset Q$.
As auxiliary data, we choose a trivialization
$\mathcal{T}:\hat\phi^{*}\mathcal{G} \to \mathcal{I}_{\rho}$ and
obtain the associated $\sigma$-equivariant line bundle $(\hat
R,\hat\varphi)$ over $\hat\Sigma$.   The pullback of the Jandl
module $(\mathcal{E},\rho_k)$ along $\hat\phi$ to
$\partial\hat\Sigma$ yields a Jandl module for the
trivialized Jandl gerbe: as discussed in Section \ref{sec5}, it
induces a vector bundle $E$ over $\partial\hat\Sigma$. A further
auxiliary datum is, again, a fundamental domain $F$ of $\Sigma$ in
$\hat\Sigma$. In order to account for the
boundary, we need to choose a lift (a closed one-dimensional
 submanifold) $\hat\ell\subset \partial\hat\Sigma$ of
 $\partial\Sigma$. It is easy to see that these lifts always exist.

\begin{remark}
If the boundary $\partial \Sigma$ consists of several components,
one can choose a separate $(\mathcal{G},\mathcal{J})$-brane for each
component $\ell$. It is easy to generalize the subsequent
discussion to this case.
\end{remark}

We now define a one-dimensional oriented closed submanifold $\bar F$
that generalizes the one used in the closed case.
As a set, it is defined to be
\begin{equation}
\label{50}
\bar F:=\mathrm{pr} (\partial F\,\backslash\,\hat\ell )\text{,}
\end{equation}
see Figure \ref{fig2}.
\begin{figure}[h]
\begin{center}
\includegraphics{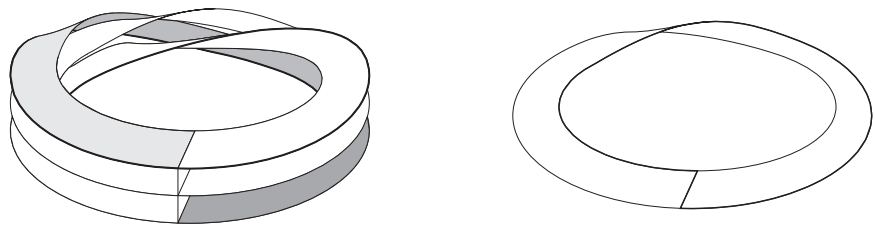}\setlength{\unitlength}{1pt}\begin{picture}(0,0)(426,545)\end{picture}
\end{center}
\caption{On the left, the Möbius strip with a fundamental domain as
in Figure \ref{fig1}, together with a lift $\hat\ell$ of the
boundary $\partial\Sigma$ (the thick line). On the right, the
associated one-dimensional oriented submanifold $\bar F$ of
$\Sigma$.} \label{fig2}
\end{figure}
This space is equipped with the structure of an oriented
one-dimensional piecewise smooth manifold as follows. Let $U \subset
\bar F$ be a small open neighborhood. If $U \cap
\partial\Sigma=\emptyset$, we have $U= \hat U/\sigma$ with $\hat
U:=\mathrm{pr}^{-1}(U)\subset \partial F$ so that $U$ is
one-dimensional and oriented like in the situation for a closed
surface. Otherwise, there exists a unique  continuous section
$s:U\to\partial F$ such that $s(U)\cap \hat\ell=\emptyset$. This
section induces the structure of a one-dimensional and oriented
manifold on $U$. It is easy to see that the orientations coincide on
intersections.

\begin{definition}
\label{def7} Let $\mathcal{J}$ be a Jandl structure on a bundle
gerbe $\mathcal{G}$ over $M$, let $(Q,\mathcal{E},\rho)$ be a
$(\mathcal{G},\mathcal{J})$-brane and let $\hat\phi: \hat\Sigma \to
M$ be an equivariant smooth map with $\hat\phi(\partial \hat\Sigma)
\subset Q$. Given a trivialization
\begin{equation*}
\mathcal{T}: \hat\phi^{*}\mathcal{G} \to
\mathcal{I}_{\rho}\textrm{,}
\end{equation*}
let $\bar R$ be the induced line bundle over $\Sigma$, and let
$E$ be the pullback vector bundle over $\partial\hat\Sigma$.
Choose, furthermore a fundamental
domain $F$ of $\Sigma$ in its oriented double $\hat\Sigma$
and a lift $\hat\ell$ of
the boundary of $\Sigma$. Then, the holonomy along
$\hat\phi$ is defined as
\begin{equation*}
\mathrm{Hol}_{\mathcal{G},\mathcal{J},\mathcal{E}}\big(\hat\phi,\Sigma\big) := \mathrm{exp} \left (
\mathrm{i} \int_{F}\rho
\right ) \cdot \mathrm{Hol}_{\bar R}(\bar F) \cdot  \mathrm{tr}\left ( \mathrm{Hol}_E(\hat\ell) \right ) \text{.}
\end{equation*}
\end{definition}

Obviously, the holonomy formulae \erf{27} and \erf{28} are reproduced
for an empty boundary or an oriented $\Sigma$, respectively. In particular,
formula \erf{26} is reproduced for an oriented closed surface.

\begin{theorem}
\label{th1} Definition \ref{def7} depends neither on the choice of
the trivialization $\mathcal{T}$ nor on the choice of the
fundamental domain $F$ nor on the choice of the lift $\hat\ell$.
\end{theorem}

\noindent We give a complete proof of this theorem in the next section in
terms of local data. Before we switch to local data, let us
elaborate on those properties of the holonomy formula that
can conveniently be discussed in the geometric setting.

For example, one can  check that the holonomy of Definition
\ref{def7} is independent of the choice of the trivialization. If
two trivializations $\mathcal{T}:\hat\phi^{*}\mathcal{G} \to
\mathcal{I}_{\rho}$ and $\mathcal{T}':\hat\phi^{*}\mathcal{G} \to
\mathcal{I}_{\rho'}$ are present, we form the composition
\begin{equation*}
\mathcal{T}' \circ \mathcal{T}^{-1}: \mathcal{I}_{\rho} \to \mathcal{I}_{\rho'}\text{.}
\end{equation*}
The functor $\bun$ induces a line bundle
$T:=\bun(\mathcal{T}' \circ \mathcal{T}^{-1})$ over $\hat\Sigma$ of
curvature $\rho'-\rho$. The holonomy of $T$ captures the difference
that arises in the first factor:
\begin{multline}
\label{72}
 \mathrm{exp} \left (
\mathrm{i} \int_{F}\rho'
\right ) =  \mathrm{exp} \left (
\mathrm{i} \int_{F}\rho
\right ) \cdot  \mathrm{exp} \left (
\mathrm{i} \int_{F}\rho' - \rho
\right )\\ = \mathrm{exp} \left (
\mathrm{i} \int_{F}\rho
\right )\cdot \mathrm{Hol}_T(\partial F)\text{.}
\end{multline}
Notice that $\hat Q:= \sigma^{*}T \otimes T$ is a line bundle over
$\hat\Sigma$ with a canonical $\sigma$-equivariant structure given
as the permutation of the two tensor factors.  From the definition
of $\bar F$, we find, for the holonomies of $T$ and the descent line
bundle $Q$,
\begin{equation}
\label{73} \mathrm{Hol}_T(\partial F) = \mathrm{Hol}_{Q}(\bar
F) \cdot \mathrm{Hol}_T(\hat\ell)\text{.}
\end{equation}
Let $\hat R$ and $\hat R'$ be
the $\sigma$-equivariant line bundles associated to the
trivializations $\mathcal{T}$ and $\mathcal{T}'$, respectively.
We then obtain an isomorphism
\begin{equation*}
\hat R \cong\hat Q \otimes \hat R'
\end{equation*}
of $\sigma$-equivariant line bundles over $\hat\Sigma$, see the
discussion after Definition 10 in \cite{waldorf1}. For the descent
line bundles, this implies an isomorphism $R \cong Q \otimes R'$,
so that
\begin{equation}
\label{76}
\mathrm{Hol}_{Q}(\bar F) \cdot \mathrm{Hol}_{R'}(\bar F)
= \mathrm{Hol}_{R}(\bar F)\text{.}
\end{equation}

Concerning the vector bundles $E$ and $E'$, note that we have a 2-isomorphism
\begin{equation*}
\mathcal{E} \circ \mathcal{T}'^{-1} \circ \mathcal{T}' \circ \mathcal{T}^{-1} \cong \mathcal{E} \circ \mathcal{T}^{-1}
\end{equation*}
which induces, via the functor $\bun$,
an isomorphism $E' \otimes T \cong E$ of vector bundles over
$\partial\hat\Sigma$. This shows that
\begin{equation}
\label{75} \mathrm{Hol}_T(\hat\ell) \cdot  \mathrm{tr}\left (
\mathrm{Hol}_{E'}(\hat\ell) \right )=  \mathrm{tr}\left (
\mathrm{Hol}_E(\hat\ell) \right )\textrm{.}
\end{equation}
Formul\ae ~\erf{72}-\erf{75} prove that the holonomy in Definition
\ref{def7} does not depend on the choice of the trivialization.

Another result on the holonomy is

\begin{proposition}
The holonomy for the unoriented surface $\Sigma$ determines a square
root of the holonomy for the oriented double,
\begin{equation*}
\left ( \mathrm{Hol}_{\mathcal{G},\mathcal{J},\mathcal{E}}\big(\hat\phi,\Sigma\big) \right ) ^2 = \mathrm{Hol}_{\mathcal{G},\mathcal{E}}(\hat\phi,\hat\Sigma)\text{.}
\end{equation*}
\end{proposition}

\begin{proof}
To see this, one chooses a fundamental domain $F$ and a lift
$\hat\ell$ for the first of the two factors on the left-hand
side, and makes the choices $F':=\sigma(F)$ and
$\hat\ell':=\sigma(\hat\ell)$ for the second factor.
The square on the left-hand side consists,
after reordering of the factors, of
\begin{equation*}
\exp \left ( \im \int_{F}\rho \right ) \cdot \exp\left (\im
\int_{F'} \rho \right ) \stackrel{\text{\erf{83}}}{=} \exp \left(
\im \int_{\hat\Sigma} \rho \right )
\end{equation*}
and $\mathrm{Hol}_{R}(\bar F)\cdot \mathrm{Hol}_{R}(\bar
F') = 1$ (the latter identity follows from the fact that the
submanifolds $\bar F$ and $\bar F'$ are the same sets, but with
opposite orientations), as well as of
$\mathrm{tr}(\mathrm{Hol}_E(\hat\ell)) \cdot
\mathrm{tr}(\mathrm{Hol}_E(\hat\ell')) =
\mathrm{tr}(\mathrm{Hol}_E(\partial\hat\Sigma))$. Altogether, this
reproduces the holonomy formula \erf{28} for
$\hat\Sigma$.
\end{proof}

Finally, let us discuss what happens to the holonomy when we
pass to equivalent background data.

\begin{proposition}
\label{prop3}
Suppose that $(\mathcal{B},\eta_k): (\mathcal{G}^a,\mathcal{J}^a)
\to (\mathcal{G}^b,\mathcal{J}^b)$ is an equivariant 1-isomorphism
between Jandl gerbes, that $(\mathcal{E}^a,\rho^a)$ and
$(\mathcal{E}^b,\rho^b)$ are Jandl modules for
$(\mathcal{G}^a,\mathcal{J}^a)$ and $(\mathcal{G}^b,\mathcal{J}^b)$,
respectively, and that
$\nu: \mathcal{E}^b \circ \mathcal{B} \Rightarrow \mathcal{E}^a$
is a 2-isomorphism.
Then,
\begin{equation}
\label{68}
\mathrm{Hol}_{\mathcal{G}^a,\mathcal{J}^a,\mathcal{E}^a}(\hat\phi,\Sigma) = \mathrm{Hol}_{\mathcal{G}^b,\mathcal{J}^b,\mathcal{E}^b}(\hat\phi,\Sigma)
\end{equation}
for any smooth equivariant map $\hat\phi:\hat\Sigma \to M$.
\end{proposition}

\begin{proof}
We fix the choices of the fundamental domain $F$ and of the
lift $\hat\ell$ for both sides of \erf{68}. To compute the holonomy
on the right-hand side, we choose a trivialization
$\mathcal{T}^b:\hat\phi^{*}\mathcal{G}^b \to \mathcal{I}_{\rho}$. It
induces a trivialization $\mathcal{T}^a := \mathcal{T}^b \circ
\mathcal{B}$ which we use to compute the
left-hand side. Since $\mathcal{T}^a$ and $\mathcal{T}^b$ have
the same 2-form $\rho$, the first factor of the holonomy formula
from Definition \ref{def7} is the same on both sides of \erf{68}.

Associated to the trivialized Jandl gerbes
$\hat\phi^{*}(\mathcal{G}^a,\mathcal{J}^a)$ and
$\hat\phi^{*}(\mathcal{G}^b,\mathcal{J}^b)$, there are
$\sigma$-equivariant line bundles over $\hat\Sigma$, as
discussed in Section \ref{71}.
By Lemma \ref{lem3}, these line bundles are isomorphic as
equivariant line bundles and hence induce isomorphic line bundles
$R^a$ and $R^b$ over $\Sigma$. Isomorphic line bundles
have equal holonomies, therefore also
the second factor of the holonomy formula from Definition \ref{def7}
coincides on both sides of \erf{68}.

Finally, we induce a 2-isomorphism
\begin{equation*}
\alxydim{@C=0.9cm}{\mathcal{E}^b \circ (\mathcal{T}^b)^{-1} \ar@{=>}[r] & \mathcal{E}^b \circ \mathcal{B} \circ (\mathcal{T}^b \circ \mathcal{B})^{-1} \ar@{=>}[r]^-{\nu \circ \id} & \mathcal{E}^a \circ (\mathcal{T}^a)^{-1}}
\end{equation*}
whose image under the functor $\bun$ yields an isomorphism $E^b \to
E^a$ of vector bundles over $\partial\hat\Sigma$. Again,
these vector bundles have equal holonomies, so that
also the third factor coincides on both sides.
\end{proof}

Of course, it follows that equivalent Jandl modules have equal
holonomies. We remark, however, that the 2-isomorphism $\nu$ in
Proposition \ref{prop3} does not have to be equivariant. In other
words, the holonomy from Definition \ref{def7} cannot
distinguish all  equivalence classes of Jandl modules.

\subsection{Local-Data Counterpart}

\def\hu#1{\underline{\hat{#1}}}
\def\el#1{\hat e_{\hat\ell}}

Here, we rewrite the holonomy for unoriented surfaces (with
boundary) from Definition \ref{def7} in terms of local data.
Thus, let $\mathfrak{O}=\lbrace O_i \rbrace_{i\in
I}$ be an open cover of $M$, with $k(O_i)=O_{ki}$, that permits to
extract local data, namely the data $c=(B_i,A_{ij},g_{ijk})$
of the bundle gerbe $\mathcal{G}$, the data
$b=(\Pi_i,\chi_{ij})$ and $a = (f_i)$ of the Jandl structure
$\mathcal{J}$ (see Section \ref{sec3_1}), and the data
$\beta=(\Lambda_i,G_{ij})$ and $\phi=(H_i)$ of the
$(\mathcal{G},\mathcal{J})|_Q$-module $(Q,\mathcal{E},\rho)$, see
Section \ref{sec5}. The local data of the bundle gerbe satisfy
relations \erf{151} - \erf{153}. For reader's convenience, let
us recall the relations between
the local data of the Jandl structure and those of
the gerbe module, specialized to the present
case of the orientifold group $(\Z_2,\id)$. Concerning the Jandl
structure, these are equations \erf{154} - \erf{12}, namely
\begin{multline}
\label{84} -k^{*} B_{ki} -B_i = \mathrm{d}\Pi_{i} \quad\text{,
}\quad -k^{*} A_{ki\,kj} -A_{ij} = \Pi_j - \Pi_i
-\mathrm{i}\chi_{ij}^{-1}\mathrm{d}\chi_{ij}
\\\quad\text{ and }\quad
k^{*} g_{ki\,kj\,kl}^{-1} \cdot g_{ijl}^{-1} =\chi_{ij}^{-1} \cdot
\chi_{il} \cdot \chi_{jl}^{-1}\text{,}
\end{multline}
as well as equations \erf{14} - \erf{13}, namely
\begin{multline}
\label{85} -k^{*}\Pi_{ki} +\Pi_i=\mathrm{i}f_i^{-1}\mathrm{d} f_i
\quad\text{, }\quad k^{*}\chi^{-1}_{ki\,kj}\cdot
\chi_{ij}=f_i^{-1}\cdot f_j \\\quad\text{ and }\quad k^{*} f_{ki}^{-1}
\cdot f_i^{-1}=1\text{.}
\end{multline}
Concerning the gerbe module, these are equations \erf{79},
\begin{multline}
\label{86} -k^{*}\overline{\Lambda_{ki}} =H_i^{-1} \cdot \Lambda_i
\cdot H_i+\mathrm{i}H_i^{-1}\mathrm{d}H_i-\Pi_i \quad\text{, }\quad\\
k^{*} \overline{G_{ki\,kj}}=H_i^{-1}
 \cdot G_{ij} \cdot H_j\cdot\chi_{ij}^{-1}
\quad\text{ and }\quad
1=H_i \cdot k^{*}\overline{H_{ki}}\cdot f_i^{-1}\text{.}
\end{multline}

From the open cover $\mathfrak{O}$ of $M$, an equivariant smooth map $\hat\phi:\hat\Sigma \to M$ induces
an open cover of $\hat\Sigma$ with open sets $\hat V_i := \hat\phi^{-1}(O_i)$.
 Let $T$ be a triangulation of $\Sigma$ which is subordinate to this cover in the following sense. The preimage of each triangular face $t\in T$ in
$\hat\Sigma$ is supposed to have two connected components, and we
require that if $\hat t$ is one of these components, there exists an
index $i(\hat t)\in I$ such that $\hat t \subset \hat{V}_{i(\hat
t)}$. The indices may be chosen such that
\begin{equation*} i(\sigma(\hat t))=ki(\hat t)\text{.}
\end{equation*}
For the edges $e$ and the vertices $v$, we make similar choices of
indices.

According to the prescription from Definition
\ref{def7}, we have to choose a fundamental domain. As described in
\cite{schreiber1}, this can be done by selecting one of the two
components of the preimage of each face $t\in T$, to be
denoted by $\hu t$. For a sufficiently well-behaved
triangulation (e.g., one with trivalent vertices),
\begin{equation}
\label{46}
F := \bigcup_{t \in T} \hu t
\end{equation}
is a smooth submanifold with piecewise smooth boundary, as required.
The subsequent discussion does
not use this assumption. Next, we have to  choose a trivialization
$\mathcal{T}:\hat\phi^{*}\mathcal{G} \to \mathcal{I}_{\rho}$. With
respect to the cover $\hat V_i$, it has local data $\theta =
(\Theta_i,\tau_{ij})$ with
\begin{equation}
\label{29} (\rho,0,1) = \phi^{*}c + D_1 \theta\text{.}
\end{equation}
Finally, we choose a lift $\hat\ell$ of $\partial\Sigma$.

Equipped with these choices of $F$, $\mathcal{T}$ and $\hat\ell$, we
start to translate the formula of Definition \ref{def7} into the
language of local data. The first factor is
\begin{equation*}
\mathcal{F}_1:=\exp\left ( \im \int_F \rho \right ) =  \exp \left (
\im \sum_{t\in T} \int_{\hu t}\rho
\right ) =  \exp \left (
\im \sum_{t\in T} \int_{\hu t} \hat\phi^{*}B_{i(\hu t)} + \mathrm{d}\Theta_{i(\hu t)}
\right ) \text{.}
\end{equation*}
Here, the orientation on $\underline{\hat t}$ is the one induced
from $\hat\Sigma$. Using Stokes' Theorem and \erf{29}, we obtain
\begin{multline*}
\mathcal{F}_1=\prod_{t\in T} \exp \left ( \im \int_{\hu{t}}
\hat\phi^{*}B_{i(\underline{\hat t})} \right ) \prod_{\hat e \in
\partial \hu{t}} \exp \left ( \im \int_{\hat e}
\hat\phi^{*}A_{i(\underline{\hat t})i(\hat e)} + \Theta_{i(\hat e)}
\right ) \\ \cdot\ \prod_{\hat v\in\partial \hat e}
\hat\phi^{*}g^{\epsilon(\hat v,\hat e)}_{i(\underline{\hat t})i(\hat
e)i(\hat v)}(\hat v) \cdot \tau_{i(\hat e)i(\hat v)}^{-\epsilon(\hat
v,\hat e)}(\hat v) \cdot \tau_{i(\underline{\hat t})i(\hat
v)}^{\epsilon(\hat v,\hat e)}(\hat v)\text{.}
\end{multline*}
Here, the edge $\hat e$ has the
orientation induced from the boundary of $\hu{t}$, and
$\epsilon(\hat v,\hat e)= \pm 1$ is negative if, in this
orientation, $\hat v$ is the starting point of $\hat e$, and
positive otherwise. We make two manipulations in this formula.
First, the very last factor can be dropped since every vertex in a
fixed face $\hu t$ appears twice, each time with a different sign
of $\epsilon$. Second, many edges
$\hat e$ appear twice and with different orientations, so that the
corresponding integrals of $\Theta_{i(\hat e)}$
cancel. More precisely, the edges which appear only once can be of
two types. If $e$ is a common edge of two faces $t_1$ and $t_2$, we
call $e$ \emph{orientation-reversing} whenever
$\hu{t_1}$ and $\hu{t_2}$ have no common edge (see
Figure \ref{fig3}),
\begin{figure}[h]
\begin{center}
\includegraphics{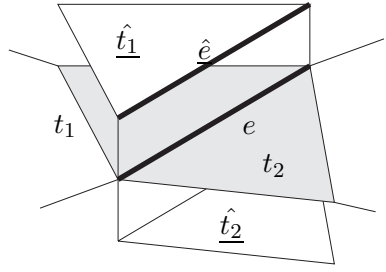}\setlength{\unitlength}{1pt}\begin{picture}(0,0)(308,306)\put(249.87430,341.49345){$t_2$}\put(170.88888,356.40822){$t_1$}\put(194.41672,388.30734){$\hu{t_1}$}\put(233.69924,318.26041){$\hu {t_2}$}\put(242.25420,355.96970){$e$}\put(225.33761,384.52082){$\hu e$}\end{picture}
\end{center}
\caption{An orientation-reversing edge $e$ between two faces $t_1$
and $t_2$ in $\Sigma$, which are lifted to different sheets. The
edge $e$ itself is, here, lifted to the same sheet as $t_1$.}
\label{fig3}
\end{figure}
and we denote by $E$ the set of orientation-reversing edges. For
each orientation-reversing edge $e\in E$, we
choose a lift $\hu e$. The second type of edges which appear only
once are those on the boundary of $\Sigma$. We denote the set of
these edges by $B$. Any $e\in B$ belongs to a unique face $t$, and
the lift $\hu t$ determines a lift $\hu e$ such that $\hu{e} \in
\partial\hu{t}$. We can now write the above formula as
\begin{eqnarray*}
\mathcal{F}_1 =\prod_{t\in T} \exp \left ( \im \int_{\underline{\hat t}}
\hat\phi^{*}B_{i(\underline{\hat t})} \right ) \prod_{\hat e \in
\partial \underline{\hat t}} \exp \left ( \im \int_{\hat e}
\hat\phi^{*}A_{i(\underline{\hat t})i(\hat e)}  \right ) \prod_{\hat
v\in\partial \hat e} \hat\phi^{*}g^{\epsilon(\hat v,\hat
e)}_{i(\underline{\hat t})i(\hat e)i(\hat v)}(\hat v)
\\
\cdot \prod_{e\in E  }\exp \left ( \im \int_{\hu{e}}
\Theta_{i(\hu{e})} + \sigma^{*}\Theta_{ki(\hu{e})} \right )
\prod_{\hat v\in\partial \hu{e}}  \tau_{i(\hu e)i(\hat
v)}^{-\epsilon(\hat v,\hu e)}(\hat v) \cdot \sigma^{*}\tau_{ki(\hu
e)\,ki(\hat v)}^{-\epsilon(\hat v,\hu e)}(\hat v)
\\
\cdot \prod_{e \in B} \exp \left ( \im \int_{\hu{e}}
\Theta_{i(\hu{e})}  \right )  \prod_{\hat v\in\partial \hu{e}}
\tau_{i(\hu e)i(\hat v)}^{-\epsilon(\hat v,\hu e)}(\hat v)\text{.}
\end{eqnarray*}

The second factor $\mathcal{F}_2 := \mathrm{Hol}_{R}(\bar F)$ is more
complicated because the line bundle $R$ over $\Sigma$ whose holonomy
we need is only given abstractly as a descended $\sigma$-equivariant
line bundle $(\hat R,\hat\varphi)$ over $\hat\Sigma$. As described
in \cite{schreiber1}, we can compute its holonomy by integrating the
local data of $\hat R$ along piecewise lifts of $\bar F$, and then
use the local data of its equivariant structure $\hat\varphi$ at
points where the lift has a jump. According to the definition of
$\hat R$, its local data can be determined from the one of the
trivialization $\mathcal{T}$ and the Jandl structure  by the formula
$-\theta + \hat\phi^{*}b + \sigma^{*}\theta$.
Thus, the line bundle $\hat R$ has local connection 1-forms
\begin{equation*}
\Psi_i := -\Theta_i + \hat\phi^{*}\Pi_i-\sigma^{*}\Theta_{ki}
\end{equation*}
on $\hat V_i$, and transition functions
\begin{equation*}
\psi_{ij} := \tau_{ij}^{-1} \cdot \hat\phi^{*}\chi_{ij} \cdot
\sigma^{*}\tau_{ki\,kj}^{-1}
\end{equation*}
on $\hat V_i \cap \hat V_j$. Its equivariant structure $\hat\varphi$
has local data $f_i$. Mimicking the definition \erf{50} of $\bar F$,
we fix the subset $\bar B \subset B$ consisting of those edges $e$
for which $\hu e$ is \emph{not} contained in the lift $\hat\ell$ of
$\partial\Sigma$. Note that $\bar F$ is
simply-covered by the lifts $\hu{e}$ of edges in $E \cup \bar B$ we
chose before. As remarked above, these lifts
typically do not patch together, i.e. there exist vertices $v\in e_1
\cap e_2$ between pairs of edges $e_1,e_2 \in E
\cup \bar B$ such that $\hu{e_1} \cap \hu{e_2}=\emptyset$. In order
to take care of these, let us choose, for \emph{all} vertices $v$,
a lift $\hu v$. Then, the second factor is
\begin{equation}
\label{87} \mathcal{F}_2 = \prod_{e\in E \cup \bar B} \exp \left ( \im
\int_{\hu{e}} \Psi_{i(\hu e)}  \right ) \prod_{\hat v \in \partial
\hu{e}} \psi^{-\epsilon(\hu e,\hat v)}_{i(\hu{e})i(\hat v)}(\hat v)
\cdot \prod_{\hu v\in \partial \hu{e}}
\hat\phi^{*}f^{\epsilon(\hu{e},\hu{v})}_{i(\hu{v})}(\hu v)\text{.}
\end{equation}
Here, the sign in the exponent of the transition functions
$\psi_{ij}$ is due to our conventions for the relation between
connection 1-forms and transition functions of a line bundle. Notice
that a vertex $v$ contributes to the last factor
of \erf{87} only if it belongs to the adjacent edges in
$E\cup\bar B$ whose lifts do \emph{not} patch together.

As for the third factor, local data of the vector bundle $E$
are provided by the expression
$\hat\phi^*\beta - \theta$. For each edge $e\in B$,
we denote by $\el e$ the
corresponding lift such that $\el e \subset
\hat\ell$. Then,
\begin{multline*}
\mathcal{F}_3 := \mathrm{tr}\left ( \mathrm{Hol}_E(\hat \ell) \right) =
\mathrm{tr} \; \mathcal{P}  \prod_{e\in B }\hspace{-0.05cm}\exp
\left ( \im \int_{\el e} \hat\phi^{*}\Lambda_{i(\el e)} -
\Theta_{i(\el e)} \right )\hspace{-0.15cm}\\
\cdot\ \prod_{\hat v \in \partial
\el{e}} \hat\phi^{*}G_{i(\el{e})i(\hat v)}^{-\epsilon(\hat
v,\el{e})}(\hat v) \cdot \tau_{i(\el{e})i(\hat v)}^{\epsilon(\hat
v,\el{e})}(\hat v) \text{,}
\end{multline*}
where the symbol $\mathcal{P}$ indicates
that the edges $e$ have to be ordered according to the orientation
on $\hat\ell$ (which is the one induced from $\partial\hat\Sigma$).
Since we take the trace, it does not matter at which vertex one
starts.

We may now compute the product $\mathcal{F}_1 \cdot\mathcal{F}_2
\cdot\mathcal{F}_3$, which is, by construction, the holonomy of Definition
\ref{def7}. We claim that all occurrences of $\Theta_i$ and
$\tau_{ij}$ drop out: First of all, for each edge $e\in E$, the
contributions from $\mathcal{F}_1$ and $\mathcal{F}_2$ are
\begin{multline*}
\exp \left ( \im \int_{\hu{e}}  \Theta_{i(\hu{e})} +
\sigma^{*}\Theta_{ki(\hu{e})} \right ) \cdot \exp \left ( \im
\int_{\hu{e}} -\Theta_{i(\hu e)}  - \sigma^{*}\Theta_{ki(\hu e)}
\right ) \\ \cdot\ \prod_{\hat v\in\partial \hu{e}}  \tau_{i(\hu e)i(\hat
v)}^    {-\epsilon(\hat v,\hu e)}(\hat v) \cdot
\sigma^{*}\tau_{ki(\hu e)\,ki(\hat v)}^{-\epsilon(\hat v,\hu
e)}(\hat v) \cdot \tau^{\varepsilon(\hat v,\hu e)}_{i(\hu{e})i(\hat
v)}(\hat v) \cdot \sigma^{*}\tau^{\varepsilon(\hat v,\hu
e)}_{ki(\hu{e})\,ki(\hat v)}(\hat v)\text{,}
\end{multline*}
which is obviously equal to 1.
For each edge $e \in \bar B$, we have contributions from
all three factors, namely
\begin{multline*}
\exp \left ( \im \int_{\hu{e}}  \Theta_{i(\hu{e})}  \right )\cdot
\exp \left ( \im \int_{\hu{e}} -\Theta_{i(\hu e)}  -
\sigma^{*}\Theta_{ki(\hu e)}  \right ) \cdot \exp \left ( \im
\int_{\el e} - \Theta_{i(\el e)} \right )\\ \cdot\ \prod_{\hat
v\in\partial \hu{e}}  \tau_{i(\hu e)i(\hat v)}^{-\varepsilon(\hat
v,\hu e)}(\hat v) \cdot \tau^{\varepsilon(\hat v,\hu
e)}_{i(\hu{e})i(\hat v)} (\hat v)\cdot
\sigma^{*}\tau^{\varepsilon(\hat v,\hu e)}_{ki(\hu{e})\,ki(\hat
v)}(\hat v)  \cdot
 \prod_{\hat v \in \partial \el{e}}  \tau_{i(\el{e})i(\hat v)}^{\varepsilon(\hat v,\el{e})}(\hat v)\text{.}
\end{multline*}
By the definition of $\bar B$, we have $\hu e = \sigma (\el
e)$, with opposite orientations. Hence, these contributions also
cancel out. For the remaining edges $e \in B \setminus \bar
B$, we only have contributions from $\mathcal{F}_1$ and $\mathcal{F}_3$, namely
\begin{equation*}
\exp \left ( \im \int_{\hu{e}}  \Theta_{i(\hu{e})}  \right )  \prod_{\hat v\in\partial \hu{e}}  \tau_{i(\hu e)i(\hat v)}^{-\varepsilon(\hat v,\hu e)}(\hat v)\cdot \exp \left ( \im \int_{\el e} - \Theta_{i(\el e)} \right ) \prod_{\hat v \in \partial \el{e}}  \tau_{i(\el{e})i(\hat v)}^{\varepsilon(\hat v,\el{e})}(\hat v)\text{.}
\end{equation*}
Here, we know that $\hu e = \el e$, so that these terms cancel
out, too. Finally, we end up with the following \emph{local
holonomy formula}:
\begin{eqnarray}
\nonumber
&&\hspace{-0.7cm}\mathrm{Hol}_{\mathcal{G},\mathcal{J},\mathcal{E}}(\hat\phi,\Sigma)
\\\nonumber&& = \prod_{t\in T} \exp \left ( \im
\int_{\underline{\hat t}} \hat\phi^{*}B_{i(\underline{\hat t})}
\right ) \prod_{\hat e \in \partial \underline{\hat t}} \exp \left (
\im \int_{\hat e} \hat\phi^{*}A_{i(\underline{\hat t})i(\hat e)}
\right )  \prod_{\hat v\in\partial \hat e}
\hat\phi^{*}g^{\varepsilon(\hat v,\hat e)}_{i(\underline{\hat
t})i(\hat e)i(\hat v)}(\hat v)
\\ &&\hspace{1cm} \cdot\prod_{e\in E \cup \bar B}
\exp \left ( \im \int_{\hu{e}} \hat\phi^{*}\Pi_{i(\hu e)}  \right )
\prod_{\hat v \in \partial \hu{e}}
\hat\phi^{*}\chi^{-\varepsilon(\hat v,\hu e)}_{i(\hu{e})i(\hat
v)}(\hat v) \cdot  \prod_{\hu v\in \partial \hu{e}}
\hat\phi^{*}f^{\varepsilon(\hu{v},\hu{e})}_{i(\hu{v})}(\hu v)
\label{74}\\\nonumber && \hspace{1cm}\cdot\ \ \prod_{e\in B
}\mathrm{tr}\;\mathcal{P} \exp \left ( \im \int_{\el e}
\hat\phi^{*}\Lambda_{i(\el e)} \right ) \prod_{\hat v \in \partial
\el{e}} \hat\phi^{*}G_{i(\el{e})i(\hat v)}^{-\varepsilon(\hat
v,\el{e})}(\hat v)
\end{eqnarray}

Let us briefly review where the particular terms come from.
The first line pairs
up the local data of the bundle gerbe with  the triangulation
of the fundamental domain. The second line takes care of
the orientation-reversing edges. It pairs up the
local data of the Jandl structure with lifts of these edges, and
compensates inconsistent lifts. The third line pairs up the
local data of  the gerbe module with $\hat\ell$.

If the boundary is empty, $B=\emptyset$, then formula \erf{74}
reduces exactly to the one given in
\cite{schreiber1} for the holonomy of closed unoriented surfaces
written in terms of local data. If the surface is oriented, we can
make choices such that $E=\bar B=\emptyset$, so that the formula
reduces to the one given in \cite{gawedzki4} for the holonomy of
oriented surfaces with boundary. Finally, if the surface is oriented
and closed, only the first line survives and reduces to the formula
found in \cite{alvarez1, gawedzki3}.

We  remark that the local data $\theta$ of the trivialization have
vanished completely from the formula \erf{74}. 
This reflects the independence
of Definition \ref{def7} of the choice of trivializations 
demonstrated in the
previous section. Similarly, Proposition \ref{prop3} implies that
\erf{74} is independent of the choice of local data of the bundle gerbe,
the Jandl structure and the equivariant gerbe module. It is a nice exercise
to check this directly by showing that the local holonomy formula \erf{74}
is independent of the choices of the subordinated indices $i$, the lifts
$\hu e$ and $\hu v$, and that it remains unaltered when passing 
to a finer triangulation or to
cohomologically equivalent local data. In order to complete the proof
of Theorem \ref{th1}, it is then enough to show that the local expression
\erf{74} is independent of the choice of the lifts
$\hu t$ of the faces of the triangulation $T$ and of that of the 
lift $\hat\ell$ of the boundary $\ell$.

Let us first suppose that we replace
a triangle $\hu t$ by $\hu t'$ differing from $\sigma(\hu t)$
by the orientation or, in short, $\hu t'=-\sigma(\hu t)$. We write the
first line of the local holonomy formula
\erf{74} as
\begin{equation*}
\prod_{t\in T} H(\hu t)
\end{equation*}
where $H(\hu t)$ is the contribution of the face $t$  with the 
choice $\hu t$ of the lift.
One obtains after simple algebra employing
relations \erf{84}:
\begin{equation*}
H({\hu t}{}')=H(\hu t) \cdot \prod_{\hat e\in\partial \hu t} I(\hat e)
\quad
\text{ with }
\quad
 I(\hat e) := \exp \left ( \im \int_{\hat e} \hat\phi^{*}\Pi_{i(\hat e)} \right) \cdot \prod_{\hat v \in \partial \hat e }\hat\phi^{*}\chi_{i(\hat e)i(\hat v)}^{-\varepsilon(\hat v,\hat e)}(\hat v) \text{.}
\end{equation*}
To compute the changes in the second line, let $E_t := \partial t \cap (E \cup \bar B)$ be the set of those edges of $t$ that are either orientation-reversing or located on the boundary.
We may assume that the edges $\hu e$ chosen for them satisfy
$\hu e\in \partial \hu t$. Under the replacement
of $\hu t$ by $\hu t { }'$,
the set $E_t$ changes
to the complementary set of edges of $t$ and we may assume that the edges
$\hu e'$ chosen for them satisfy $\hu e' \in \partial
\hu t'$. We have using relations \erf{85}:
\begin{equation}
\label{77}
\prod_{e\in \partial t \setminus E_t} I(\hu e{}')=\prod_{e\in \partial t \setminus E_t} I(\hat e)^{-1} \cdot \prod_{\hat v\in\partial \hat e} \hat\phi^{*}f_{i(\hat v)}^{-\varepsilon(\hat v, \hat e)}(\hat v)\text{,}
\end{equation}
where $\hat e$ denotes the edge in
$\hu t$ projecting to $e\subset t$ (with the orientation
induced from $\hu t$) and where we have used the fact that,
for $e\not\subset E_t$, $\,\hu e'$ and $\sigma(\hat e)$ differ
only by the orientation, i.e. $\hu e'=-\sigma(\hat e)$. Note that
\begin{equation*}
 \prod_{\hat e\in\partial \hu t} I(\hat e)\ \cdot \prod_{e\in \partial t
\setminus E_t} I(\hat e)^{-1}\ =\  \prod_{e\in E_t} I(\hat e)
\end{equation*}
which is a needed expression, a part of the original second line. 
The remaining factors
\begin{equation*}
\prod_{e\in \partial t \setminus E_t}\ \prod_{\hat v\in\partial \hat e} \hat\phi^{*}f_{i(\hat v)}^{-\varepsilon(\hat v, \hat e)}(\hat v)
\end{equation*}
from \erf{77} compensate the remaining changes in the second line. Indeed,
again with $\hat e\in \partial \hu t$ and
$\hu v'=\sigma(\hu v)$,
\begin{eqnarray}
&&\hspace{-1cm}\prod_{e\in \partial t \setminus E_t} \cdot 
\prod_{\hu v\in\partial \hu e'} \hat\phi^{*}
f_{i(\hu v)}^{\varepsilon(\hu v, \hu e')}(\hat v) 
\ \cdot\hspace{-0.1cm}\prod_{e\in \partial t \setminus E_t} \cdot 
\prod_{\hat v\in\partial \hat e} \hat\phi^{*}f_{i(\hat v)}^{-\varepsilon
(\hat v, \hat e)}(\hat v)\cr
&&=\ \prod_{e\in \partial t \setminus E_t} \cdot 
\prod_{\hu v'\in\partial \hat e} \hat\phi^{*}
f_{i(\hu v')}^{\varepsilon(\hu v', \hat e)}(\hat v')
\ \cdot\hspace{-0.1cm}\prod_{e\in \partial t \setminus E_t} 
\cdot \prod_{\hat v\in\partial
\hat e} \hat\phi^{*}f_{i(\hat v)}^{-\varepsilon(\hat v, \hat e)}(\hat v)\cr
&&=\ \prod_{e\in \partial t \setminus E_t} \cdot \prod_{\hu v\in\partial
\hat e} \hat\phi^{*}f_{i(\hu v)}^{-\varepsilon(\hu v, \hat e)}(\hu v)
\nonumber\\ && =\ \prod_{e\in E_t} \cdot \prod_{\hu v\in\partial \hu e} \hat\phi^{*}f_{i(\hu v)}^{\varepsilon(\hu v, \hu e)}(\hu v)\text{,}
\nonumber
\end{eqnarray}
where the last equality follows from the identity
\begin{equation*}
\prod\limits_{\hat e\in
\partial\hu t}\cdot\prod\limits_{\hu v\in\partial\hat e}
\hat\phi^{*}f_{i(\hu v)}^{\varepsilon(\hu v, \hu e)}(\hu v)\ =\ 1\,\text{.}
\end{equation*}
We infer that the local holonomy formula does not change
when $\hu t$ is replaced by $\hu t'$.

Next we want to analyze the effect of the change of the lift of $\ell$
from $\hat \ell$ to $\hat\ell'
=-\sigma(\hat\ell)$. Only the lines two and three of the local holonomy
formula \erf{74} change in this case. In the third line $L_3$, 
we find a change of the form
\begin{equation*}
L_3' = L_3 \cdot \prod_{e\in B } I(\hat e_{\ell})^{-1}\text{.}
\end{equation*}
Here we have used the relations \erf{86} and the identities
\begin{equation*}
\mathrm{tr}(\overline{G})=\mathrm{tr}(G^{-1})
\quad\text{ and }\quad
\mathrm{tr}\left ( \mathrm{e}^{\im\, \overline{\Lambda}} 
\right ) = \mathrm{tr}\left( \mathrm{e}^{\im \Lambda} \right )
\end{equation*}
valid for  $G\in U(n)$ and  $\Lambda\in \mathfrak{u}(n)$. 
In the second line $L_{2}$ we find
\begin{equation*}
L_2' = \prod_{e\in B \setminus \bar B} I(\hat e_{\hat\ell})  \cdot
\prod_{\hu v\in \partial \hat e_{\hat\ell}} \hat\phi^{*}
f^{\varepsilon(\hu{v},\hat e_{\hat\ell} )}_{i(\hu{v})}(\hu v)
\end{equation*}
where $\hat e_{\hat\ell}\subset \hat\ell$
is taken with the orientation inherited from $\hat\ell$.
Multiplying both lines together, we obtain
\begin{equation*}
L_3'\cdot L_2'\ =\ L_3\,\cdot\prod_{e\in\bar B}
I(\hat e_{\hat\ell})^{-1}  \cdot
\prod\limits_{e\in B\setminus\bar B}\prod_{\hu v\in \partial \hat e_{\hat\ell}} \hat\phi^{*}
f^{\varepsilon(\hu{v},\hat e_{\hat\ell} )}_{i(\hu{v})}(\hu v)\,.
\end{equation*}
On the right-hand side, we may pass back from
$\,\hat e_{\hat\ell}\,$ to $\hat e_{\hat\ell'}=-\sigma(\hat e_{\hat\ell})$
using \erf{85}: 
to obtain
\begin{equation*}
\prod_{e\in\bar B} I(\hat e_{\hat\ell})^{-1}\ =\
\prod_{e\in\bar B} I(\hat e_{\hat\ell'})\
\prod\limits_{\hat v\in\partial\hat e_{\hat\ell'}}
\hat\phi^{*}f^{\varepsilon(\hat v,\hat e_{\hat\ell'})}_{i(\hat v)}(\hat v)
\end{equation*}
and, for $\hu v'=\sigma(\hu v)$,
\begin{multline*}
\prod\limits_{e\in B\setminus\bar B}\,\prod_{\hu v\in
\partial \hat e_{\hat\ell}} \hat\phi^{*}
f^{\varepsilon(\hu{v},\hat e_{\hat\ell} )}_{i(\hu{v})}(\hu v)
\ =\ \prod\limits_{e\in B\setminus\bar B}\,\prod_{\hu v'\in
\partial \hat e_{\hat\ell'}} \hat\phi^{*}
f^{\varepsilon(\hu{v}',\hat e_{\hat\ell'} )}_{i(\hu{v}')}(\hu v')\ \\=\
\prod\limits_{e\in\bar B}\,\prod_{\hu v'\in
\partial \hat e_{\hat\ell'}} \hat\phi^{*}
f^{-\varepsilon(\hu{v}',\hat e_{\hat\ell'} )}_{i(\hu{v}')}(\hu v')\,,
\end{multline*}
where the last equality follows from the obvious identity
\begin{equation*}
\prod\limits_{e\in B}\,\prod\limits_{\hu v'\in\partial\hat e_{\hat\ell'}}
f^{\varepsilon(\hu{v}',\hat e_{\hat\ell'} )}_{i(\hu{v}')}(\hu v')\,=\,1\,.
\end{equation*}
Upon performing all these transformations, we arrive at the formula:
\begin{equation*}
L_3'\cdot L_2'\ =\ L_3\,\cdot\prod_{e\in\bar B} I(\hat e_{\hat\ell'})  \cdot
\prod_{\hu v\in \partial \hat e_{\hat\ell'}} \hat\phi^{*}
f^{\varepsilon(\hu{v},\hat e_{\hat\ell'} )}_{i(\hu{v})}(\hu v)\,.
\end{equation*}
Noting that $\,\hat e_{\hat\ell'}=\hu e\,$ for $e\in\bar B$, we identify
the right-hand side with $L_3\cdot L_2$. This ends the proof of
the independence of the local expression \erf{74} of the lift $\hat\ell$.

Thus, altogether, the local holonomy formula \erf{74} is independent of all
the arbitrary choices made. Accordingly, also the geometric 
holonomy formula
from Definition \ref{def7} is manifestly associated only to the bundle gerbe,
its Jandl structure and its modules, and, of course to the equivariant map 
$\hat\phi:\hat\Sigma\rightarrow M$. This proves Theorem \ref{th1}.

\section{Conclusions}

We considered in this paper manifolds $M$ equipped with a closed
3-form $H$ and an orientifold-group action. The latter is an
action of a finite group $\Gamma$ on $M$ such that, for
$\gamma\in\Gamma$, one has $\gamma^*H=\epsilon(\gamma)H$ for a
homomorphism $\epsilon:\Gamma\rightarrow\{\pm1\}$. We
introduced the notion of a
$(\Gamma,\epsilon)$-equivariant (or twisted-equivariant) structure
on a gerbe $\mathcal{G}$ over $M$ with curvature
$H$. This notion
 extends that of a so-called Jandl structure introduced in \cite{schreiber1},
to which it reduces for $\Gamma=\{\pm1\}$ and $\epsilon(\pm1)=\pm 1$.

In the case of $\Gamma_0={\rm ker}(\epsilon)$
acting on $M$ without fixed points, equivalence
classes of $(\Gamma,\epsilon)$-equivariant gerbes over $M$ were
shown to descend to equivalence classes of gerbes over
$M'=M/\Gamma_0$, with $(\Gamma',\epsilon')$-equivariant structures
for $\Gamma'=\Gamma/\Gamma_0$ and $\epsilon'$ induced from
$\epsilon$. For $\Gamma_0=\Gamma$, this gives a way to construct
gerbes over $M'$ from gerbes over $M$, and for
$\Gamma/\Gamma_0=\mathbb{Z}_2$, it enables to
construct gerbes with Jandl structure over $M'$. Working with local
data, we showed that equivalence classes of
$(\Gamma,\epsilon)$-equivariant gerbes can be identified with
classes of the $2^{\text{nd}}$ hypercohomology group of a double
complex of chains on $\Gamma$ with values in the (real) Deligne
complex in degree 2. This identification permitted to study the
obstructions to the existence of $(\Gamma,\epsilon)$-equivariant
structures on a given gerbe $\mathcal{G}$ with curvature $H$. In the
case of 2-connected manifolds, the unique obstruction takes values
in the cohomology group $H^3(\Gamma,U(1)_\epsilon)$, where the
coefficient group $U(1)$ is taken with the action $(\gamma,u)\mapsto
u^{\epsilon(\gamma)}$ of $\Gamma$. If this
obstruction vanishes, equivalence classes of
$(\Gamma,\epsilon)$-equivariant structures on $\mathcal{G}$ are
parameterized by  cohomology classes in $H^2(\Gamma,U(1)_\epsilon)$.
This agrees with the purely local analysis of \cite{gawedzki6}.

In \cite{gawedzki6}, these results were applied to the case of
gerbes $\mathcal{G}_k$ with curvature
$H=\frac{k}{12\pi}\hskip0.05cm{\rm tr}\hskip0.05cm g^{-1}dg^{\wedge
3}$ over simple simply-connected compact Lie groups $G$,
for integer $k$. We considered there the orientifold
groups $\Gamma=\mathbb{Z}_2\ltimes Z$ with $Z$ a subgroup of the
center $Z(G)$ of $G$ acting on $G$ by multiplication, and the
non-trivial element of $\mathbb{Z}_2$ sending
$g\in G$ to $(\zeta g)^{-1}$ for $\zeta\in Z(G)$. In that
paper, we also computed the classes
$[u]\in H^3(\Gamma,U(1)_\epsilon)$ obstructing the existence
of $(\Gamma,\epsilon)$-equivariant structures on the gerbes $\mathcal{
G}_k$ and found the trivializing chains $v$ such that $u=\delta v$
whenever $[u]$ vanishes. These data enter an explicit construction
of $(\Gamma,\epsilon)$-equivariant structures on the gerbes $\mathcal{
G}_k$ that will be described in \cite{gawedzki7} in analogy to the
construction of \cite{gawedzki2} for the orbifold group $\Gamma=Z$
with trivial $\epsilon$. Such structures on $\mathcal{G}_k$ permit to
construct orientifold WZW models for closed surfaces.

With applications to the boundary field theories in view, we
discussed above twisted-equivariant gerbe modules,
and their equivalence, as well as the
descent theory for them. These results will be used to construct
boundary orientifold WZW models. The construction, extending the one
of \cite{gawedzki4} for the orbifold case, is postponed to
\cite{gawedzki7}. We also plan to compare in \cite{gawedzki7} our
geometric approach to WZW orientifolds to the algebraic ones of
\cite{fuchs,brunner2}.

The $(\Gamma,\epsilon)$-equivariant structures on gerbes and gerbe
modules are used to define the contribution of the $H$-flux to the
Feynman amplitudes of the orientifold sigma models. Such
contributions describe the gerbe holonomy along surfaces in $M$
defined by classical fields, with contributions from gerbe modules
in the case of surfaces with boundary. We discussed above the
holonomy for surfaces in the particular case of Jandl structures in
both geometric and local terms, extending the
discussion of \cite{schreiber1} to the boundary case. In
\cite{gawedzki7}, we shall relate the surface holonomy to the more
standard loop-holonomy of connections on line and vector bundles
with $(\Gamma,\epsilon)$-equivariant structures over spaces of
closed and open curves (``strings'') in $M$. Such structures  will
be obtained from twisted-equivariant gerbes and gerbe modules by
transgression, see \cite{gawedzki1} for the discussion of the
orbifold case. They play an important role in the geometric
quantization of orientifold sigma models where the equivariant
sections of the bundles over the spaces of curves represent quantum
states of the theory. This is the approach that we will
adopt in \cite{gawedzki7} for the orientifolds of boundary 
WZW models.

\bibliographystyle{my-h-elsevier}

\end{document}